\documentclass[journal]{IEEEtran}

\usepackage{amsmath}
\usepackage{amssymb}
\usepackage{amsthm}
\usepackage{mathtools}

\usepackage{algorithm, algpseudocode}

\usepackage{tikz}
\usetikzlibrary{shapes.geometric,fit,backgrounds,arrows.meta}

\usepackage{pgfplots}
\pgfplotsset{width=10cm,compat=1.9}

\newtheorem{theorem}{Theorem}
\newtheorem*{theorem*}{Theorem}
\newtheorem{lemma}[theorem]{Lemma}
\newtheorem*{lemma*}{Lemma}
\newtheorem{corollary}[theorem]{Corollary}
\newtheorem*{corollary*}{Corollary}

\newtheorem*{proposition*}{Proposition}

\newtheorem*{definition*}{Definition}

\newtheorem*{remark*}{Remark}

\newtheorem*{example*}{Example}

\DeclareMathOperator*{\lexmin}{lex\,min}

\usepackage{enumitem}
\newlist{assumption}{enumerate}{3}
\setlist[assumption]{
	label=(A\arabic*)
	,before=\it\color{black}
	,font=\normalfont
	,resume
	,leftmargin=*
	,align=left
}
\newlist{stability}{enumerate}{3}
\setlist[stability]{
	label=(S\arabic*)
	,before=\it\color{black}
	,font=\normalfont
	,resume
	,leftmargin=*
	,align=left
}

\begin{document}

\title{Prioritized Inverse Kinematics: Desired Task Trajectories in Nonsingular Task Spaces}

\author{Sang-ik An and Dongheui Lee
	\thanks{This work was supported in part by Technical University of Munich - Institute for Advanced Study, funded by the German Excellence Initiative. {\it(Corresponding Author: Sang-ik An.)}}
	\thanks{S. An was with the Human-Centered Assistive Robotics, Department of Electrical and Computer Engineering, Technical University of Munich, D-80333 (email: S.An.Robots@gmail.com).}
	\thanks{D. Lee is with the Human-Centered Assistive Robotics, Department of Electrical and Computer Engineering, Technical University of Munich, D-80333 Munich, Germany and also with Institute of Robotics and Mechatronics, German Aerospace Center (DLR) (email: dhlee@tum.de).}%
}

\maketitle

\begin{abstract}
	A prioritized inverse kinematics (PIK) solution can be considered as a (regulation or output tracking) control law of a dynamical system with prioritized multiple outputs.
	We propose a method that guarantees that a joint trajectory generated from a class of PIK solutions exists uniquely in a nonsingular configuration space.
	We start by assuming that desired task trajectories stay in nonsingular task spaces and find conditions for task trajectories to stay in a neighborhood of desired task trajectories in which we can guarantee existence and uniqueness of a joint trajectory in a nonsingular configuration space.
	Based on this result, we find a sufficient condition for task convergence and analyze various stability notions such as stability, uniform stability, uniform asymptotic stability, and exponential stability in both continuous and discrete times.
	We discuss why the number of tasks is limited in discrete time and show how preconditioning can be used in order to overcome this limitation.
\end{abstract}


\IEEEpeerreviewmaketitle

\section{Introduction}

\IEEEPARstart{T}{he} prioritized inverse kinematics (PIK) is a problem to find a joint trajectory that accomplishes goals of prioritized multiple tasks.
The PIK problem consists of two steps.
Firstly, we find a PIK solution that is a (regulation or output tracking) control law of a dynamical system with prioritized multiple outputs.
Then, we generate a joint trajectory by solving the differential equation whose right-hand side is the PIK solution.
When there are conflicts between tasks, available degrees of freedom (DOF) that are a common resource of multiple tasks are distributed according to the priority relations by means of projections of the joint velocity to the null spaces of higher priority tasks.
As a result, we can guarantee that the goal of a higher priority task is not touched by the actions of lower priority tasks.
The PIK problem has been studied intensively for decades in the robotics society and expanded into many areas such as constrained PIK \cite{Mansard2009}\cite{Kanoun2012}\cite{Escande2014}, task switching \cite{Lee2012}\cite{Petric2013}\cite{An2015}, prioritized control \cite{Khatib1987}\cite{Hsu1989}\cite{Sadeghian2014}\cite{Ott2015}, prioritized optimal control \cite{Nakamura1987a}\cite{Geisert2017}, learning prioritized tasks \cite{Saveriano2015}\cite{Calinon2018}\cite{Silverio2018}, etc.

In spite of its successful use in many practical applications, the theoretical aspect of the PIK problem has not been fully studied yet.
Antonelli \cite{Antonelli2009} analyzed task convergence of two popular PIK solutions under the assumptions that desired task trajectories are constant functions and a joint trajectory stays in a nonsingular configuration space.
However, desired task trajectories are usually not constant functions and it was not shown how we can guarantee that a joint trajectory stays in a nonsingular configuration space.
Bouyarmane and Kheddar \cite{Bouyarmane2018} showed that a PIK solution can be approximated in any accuracy by the multi-objective optimization with the weighted-sum scalarization and analyzed stability of the approximated PIK solution.
However, it was not proven that a joint trajectory generated from the approximated PIK solution converges to a joint trajectory generated from the PIK solution as the approximated PIK solution converges to the PIK solution.
Since a joint trajectory is generated by solving the differential equation whose right-hand side is the PIK solution, existence and uniqueness of a joint trajectory are related to smoothness of the PIK solution.
An and Lee \cite{An2019} proposed a generalization of the PIK problem as the multi-objective optimization with the lexicographical ordering by specifying three properties (dependence, uniqueness, and representation) for an objective function to be proper for the PIK problem.
Based on this definition, An and Lee \cite{An2019a} analyzed nonsmoothness, trajectory existence, task convergence, and stability of a class of PIK solutions in the general case that the PIK solution can possibly be nonsmooth; the task convergence theorem removes not only the assumption that desired task trajectories are constant functions but also the sufficient condition on the feedback gains found in Antonelli \cite{Antonelli2009}.
Even though the work by An and Lee \cite{An2019}\cite{An2019a} provides better understanding of the PIK problem theoretically, it is not easy to apply those results to practical applications.
Specifically, the existence theorem requires us to know the set of configurations in which the PIK solution is discontinuous.
Additionally, in order to guarantee task convergence, it still assumes that a joint trajectory stays in a nonsingular configuration space.
Also, stability analysis assumes that desired task trajectories are constant functions.
Lastly, there is a lack of analysis in discrete time in which most practical applications are solved.

In this paper, we propose a method that guarantees that a joint trajectory exists uniquly and stays in a nonsingular configuration space; analyze task convergence and stability in both continuous and discrete times; and show how preconditioning can be used in order to increase the number of tasks in discrete time.
In Section \ref{sec:prioritized_inverse_kinematics}, we briefly introduce the PIK problem and state the problem to solve.
In Section \ref{sec:continuous_time}, we present our method and analyze task convergence and stability in continuous time.
The basic idea comes from the fact that it is easier to assume that desired task trajectories stay in nonsingular task spaces because we can freely design desired task trajectories to have that property.
Then, we find conditions for task trajectories to stay in a neighborhood of desired task trajectories where a joint trajectory is guaranteed to exist uniquely and to stay in a nonsingular configuration space.
Once we have this property, we can find sufficient conditions for task convergence and various stability notions such as stability, uniform stability, uniform asymptotic stability, and exponential stability.
In Section \ref{sec:discrete_time}, we find sufficient conditions for task convergence and various stability notions in discrete time.
In Section \ref{sec:preconditioning}, we firstly analyze preconditioning and then show how preconditioning can be used in order to increase the number of tasks in discrete time.
We give concluding remarks in Section \ref{sec:conclusion}.

\section{Prioritized Inverse Kinematics}
\label{sec:prioritized_inverse_kinematics}

\subsection{Inverse Kinematics in Position}

Let $\mathbf{x} = (t,\mathbf{q})$ be a variable on $X = \mathbb{R}\times\mathbb{R}^n$.
We shall use $\mathbf{q}$ to denote both a point $\mathbf{q}\in\mathbb{R}^n$ and a trajectory $\mathbf{q}:[t_0,\infty)\to\mathbb{R}^n$ interchangeably and the meaning should be understood from the context.
We write $\mathbf{x}(t) = (t,\mathbf{q}(t))$ when $\mathbf{q}$ is a trajectory.
The inverse kinematics with multiple tasks in the position level is a problem to find a best joint trajectory $\mathbf{q}:[t_0,\infty)\to\mathbb{R}^n$ that satisfies the $a$-th forward kinematic equation
\begin{equation*}
\mathbf{f}_a(\mathbf{x}(t)) = \mathbf{p}_a(t)
\end{equation*}
at least approximately for all $a\in\overline{1,l} = \{1,2,\dots,l\}$ and $t\in[t_0,\infty)$ where $l\in\mathbb{N}\setminus\{1\}$ is the number of tasks, $t_0\in\mathbb{R}$ is the initial time, $\mathbf{f}_a:X\to\mathbb{R}^{m_a}$ is the $a$-th forward kinematic function that maps the joint position $\mathbf{q}\in\mathbb{R}^n$ into the $a$-th task position $\mathbf{f}_a(t,\mathbf{q})\in\mathbb{R}^{m_a}$ for each $t\in\mathbb{R}$, and $\mathbf{p}_a:\mathbb{R}\to\mathbb{R}^{m_a}$ is the $a$-th desired task trajectory.
Without loss of generality, we assume $m = m_1 + \cdots + m_l \le n$; if $m > n$, then we can redefine the forward kinematic functions as $\tilde{\mathbf{f}}_a(t,\tilde{\mathbf{q}}) = \tilde{\mathbf{f}}_a(t,\mathbf{q},q_{n+1},\dots,q_m) = \mathbf{f}_a(t,\mathbf{q})$ by introducing dummy variables or virtual joints $q_{n+1},\dots,q_m$.
The forward kinematic functions are defined by a mechanism in an environment and the desired task trajectories are designed according to a scenario.
Usually, various scenarios are applied for a mechanism in an environment, so we need to consider various $\mathbf{p} = (\mathbf{p}_1,\dots,\mathbf{p}_l)$ given $\mathbf{f} = (\mathbf{f}_1,\dots,\mathbf{f}_l)$.
Let $(\mathbb{R}^m)^{\mathbb{R}}$ be the set of all functions from $\mathbb{R}$ to $\mathbb{R}^m$ in which $\mathbf{p}$ is defined.
Solving the inverse kinematics with multiple tasks in the position level consists of two steps.
Initially, we find a map $\mathbf{u}:\mathbb{R}\times(\mathbb{R}^m)^\mathbb{R}\to\mathbb{R}^n$, called the inverse kinematics solution, whose value $\mathbf{u}(t,\mathbf{p})$ satisfies
\begin{equation}
\label{eqn:inverse_problem_position}
\mathbf{f}_a(t,\mathbf{u}(t,\mathbf{p})) = \mathbf{p}_a(t)
\end{equation}
at least approximately or equivalently minimizes the $a$-th task error
\begin{equation*}
	\mathbf{e}_a(t,\mathbf{q},\mathbf{p}) = \mathbf{p}_a(t) - \mathbf{f}_a(t,\mathbf{q})
\end{equation*}
with respect to $\mathbf{q}$ in some sense for all $(a,t,\mathbf{p})$.
Then, for every initial time $t_0\in\mathbb{R}$ and every desired task trajectory $\mathbf{p}\in(\mathbb{R}^m)^\mathbb{R}$, we can generate the best joint trajectory $\mathbf{q}:[t_0,\infty)\to\mathbb{R}^n$ by letting
\begin{equation*}
\mathbf{q}(t) = \mathbf{u}(t,\mathbf{p})
\end{equation*}
for all $t\in[t_0,\infty)$.
In most cases, solving the inverse problem \eqref{eqn:inverse_problem_position} is difficult because $\mathbf{f}$ is highly nonlinear and is not one-to-one.
Furthermore, it is hard to guarantee a certain level of smoothness of the joint trajectory (continuity, differentiability, etc) that is required for a mechanism to execute the joint trajectory.
So, Whitney \cite{Whitney1969} introduced the resolved motion rate control that solves the inverse kinematics in the velocity level.

\subsection{Inverse Kinematics in Velocity}

The inverse kinematics with multiple tasks in the velocity level is a problem to find a best joint trajectory $\mathbf{q}:[t_0,\infty)\to\mathbb{R}^n$ that is differentiable on $(t_0,\infty)$ and satisfies $\mathbf{q}(t_0) = \mathbf{q}_0$ and the $a$-th differential forward kinematic equation
\begin{equation*}
\mathbf{f}_{ta}(\mathbf{x}(t)) + \mathbf{F}_{qa}(\mathbf{x}(t))\dot{\mathbf{q}}(t) = \mathbf{r}_a(\mathbf{x}(t))
\end{equation*}
at least approximately for all $a\in\overline{1,l}$ and $t\in(t_0,\infty)$ where $(t_0,\mathbf{q}_0)\in X$ is the initial value, $\mathbf{F}_a = \begin{bmatrix}\mathbf{f}_{ta} & \mathbf{F}_{qa}\end{bmatrix}:X\to\mathbb{R}^{m_a\times(n+1)}$ is the $a$-th velocity mapping function that maps the joint velocity $\dot{\mathbf{q}}\in\mathbb{R}^n$ into the $a$-th task velocity $\mathbf{f}_{ta}(\mathbf{x}) + \mathbf{F}_{qa}(\mathbf{x})\dot{\mathbf{q}}\in\mathbb{R}^{m_a}$ for each $\mathbf{x}\in X$, and $\mathbf{r}_a:X\to\mathbb{R}^{m_a}$ is the $a$-th reference that represents the desired behavior of the $a$-th task velocity.
Similarly as before, we assume without loss of generality that $m\le n$.
Since we need to solve $l$ such inverse problems simultaneously, preconditioning matrix valued functions $\mathbf{F}_{q1},\dots,\mathbf{F}_{ql}$ may provide better solvability.
Let $\mathbf{R}:X\to GL_n(\mathbb{R}) = \{\mathbf{A}\in\mathbb{R}^{n\times n}\mid\det(\mathbf{A})\neq0\}$ be an arbitrary matrix-valued function and define $\mathbf{J}_a:X\to\mathbb{R}^{m_a\times n}$ as $\mathbf{J}_a(\mathbf{x}) = \mathbf{F}_{qa}(\mathbf{x})\mathbf{R}^{-1}(\mathbf{x})$.
Let $\mathbf{r}_a' = \mathbf{r}_a - \mathbf{f}_{ta}$.
Then, we can rewrite the $a$-th differential forward kinematic equation as
\begin{equation*}
\mathbf{J}_a(\mathbf{x}(t))\mathbf{R}(\mathbf{x}(t))\dot{\mathbf{q}}(t) = \mathbf{r}_a'(\mathbf{x}(t)).
\end{equation*}
A specific choice of $\mathbf{R}$ and its effect on the inverse kinematics will be discussed in Section \ref{sec:preconditioning}.
We denote $\mathbf{f}_t = (\mathbf{f}_{t1},\dots,\mathbf{f}_{tl})$, $\mathbf{F}_q = \begin{bmatrix} \mathbf{F}_{q1}^T & \cdots & \mathbf{F}_{ql}^T\end{bmatrix}^T$, $\mathbf{F} = \begin{bmatrix} \mathbf{f}_t & \mathbf{F}_q\end{bmatrix}$, $\mathbf{J} = \mathbf{F}_q\mathbf{R}^{-1}$, $\mathbf{r} = (\mathbf{r}_1,\dots,\mathbf{r}_l)$, and $\mathbf{r}' = \mathbf{r} - \mathbf{f}_t$.
We orthogonalize rows of $\mathbf{J}$ by performing the full QR decomposition of $\mathbf{J}^T(\mathbf{x})$ at each $\mathbf{x}\in X$ as in \cite[Lemma 1]{An2019}
\begin{equation}
\label{eqn:orthogonalization_of_J}
\underbrace{\begin{bmatrix} \mathbf{J}_1 \\ \vdots \\ \mathbf{J}_l \end{bmatrix}}_{\mathbf{J}(\mathbf{x})\in\mathbb{R}^{m\times n}} = \underbrace{\begin{bmatrix} \mathbf{C}_{11} & \cdots & \mathbf{0} & \mathbf{0} \\ \vdots & \ddots & \vdots & \vdots \\ \mathbf{C}_{l1} & \cdots & \mathbf{C}_{ll} & \mathbf{0} \end{bmatrix}}_{\mathbf{C}_e(\mathbf{x}) = [\mathbf{C}_{ij}(\mathbf{x})] \in \mathbb{R}^{m\times n}} \underbrace{\begin{bmatrix} \hat{\mathbf{J}}_1 \\ \vdots \\ \hat{\mathbf{J}}_{l+1} \end{bmatrix}}_{\hat{\mathbf{J}}_e(\mathbf{x})\in\mathbb{R}^{n\times n}}.
\end{equation}
$\mathbf{F}$ is defined by a mechanism in an environment, $\mathbf{R}$ is constructed from $\mathbf{F}_q$, and $\mathbf{r}$ is designed according to a scenario.
Usually, various scenarios are applied for a mechanism in an environment, so we need to consider various $\mathbf{r}$ given $\mathbf{F}$ and $\mathbf{R}$.
Let $(\mathbb{R}^m)^X$ be the set of all functions from $X$ to $\mathbb{R}^m$ in which $\mathbf{r}$ is defined.
Solving the inverse kinematics in the velocity level consists of two steps.
Initially, we find a map $\mathbf{u}:X\times(\mathbb{R}^m)^X\to\mathbb{R}^n$, called the inverse kinematics solution, whose value $\mathbf{u}(\mathbf{x},\mathbf{r})$ satisfies
\begin{equation}
\label{eqn:inverse_problem_velocity}
\mathbf{J}_a(\mathbf{x})\mathbf{R}(\mathbf{x})\mathbf{u}(\mathbf{x},\mathbf{r}) = \mathbf{r}'_a(\mathbf{x})
\end{equation}
at least approximately or equivalently minimizes the $a$-th residual
\begin{equation*}
	\mathbf{e}_a^\mathrm{res}(\mathbf{x},\dot{\mathbf{q}},\mathbf{r}) = \mathbf{r}_a'(\mathbf{x}) - \mathbf{J}_a(\mathbf{x})\mathbf{R}(\mathbf{x})\dot{\mathbf{q}}
\end{equation*}
with respect to $\dot{\mathbf{q}}$ in some sense for all $(a,\mathbf{x},\mathbf{r})$.
Then, for every initial condition $(t_0,\mathbf{q}_0)\in X$ and every reference $\mathbf{r}\in(\mathbb{R}^m)^X$, we can generate the best joint trajectory $\mathbf{q}:[t_0,\infty)\to\mathbb{R}^n$ by solving the differential equation
\begin{equation}
\label{eqn:differential_equation}
\dot{\mathbf{q}}(t) = \mathbf{u}(t,\mathbf{q}(t),\mathbf{r})
\end{equation}
for $t\in(t_0,\infty)$ with the initial value $\mathbf{q}(t_0) = \mathbf{q}_0$.

\subsection{Prioritized Inverse Kinematics}

When solving the inverse problem \eqref{eqn:inverse_problem_velocity}, the expression `at least approximately' can be understood very differently by each person.
We may specify the details on this expression by constructing an optimization problem.
Since we have multiple tasks, we need to introduce a vector-valued objective function $\boldsymbol{\pi} = (\pi_1,\dots,\pi_l):X\times\mathbb{R}^n\times(\mathbb{R}^m)^X\to[0,\infty]^l$ such that \eqref{eqn:inverse_problem_velocity} can be achieved for each $(\mathbf{x},\mathbf{r})$ at least approximately by choosing $\mathbf{u}(\mathbf{x},\mathbf{r})$ that minimizes $\pi_a(\mathbf{x},\mathbf{R}(\mathbf{x})\dot{\mathbf{q}},\mathbf{r})$ with respect to $\dot{\mathbf{q}}\in\mathbb{R}^n$.
Since $\mathbf{R}$ is invertible everywhere, it is equivalent to choose $\mathbf{R}(\mathbf{x})\mathbf{u}(\mathbf{x},\mathbf{r})$ that minimizes $\pi_a(\mathbf{x},\mathbf{y},\mathbf{r})$ with respect to $\mathbf{y}\in\mathbb{R}^n$.
There does not typically exist an optimal solution $\mathbf{y}^*\in\mathbb{R}^n$ that minimizes all objective functions $\pi_1(\mathbf{x},\mathbf{y},\mathbf{r}),\dots,\pi_l(\mathbf{x},\mathbf{y},\mathbf{r})$ with respect to $\mathbf{y}$ simultaneously.
Then, a canonical choice of $\mathbf{R}(\mathbf{x})\mathbf{u}(\mathbf{x},\mathbf{r})$ would be a Pareto optimal solution of the multi-objective optimization
\begin{equation*}
\min_{\mathbf{y}\in\mathbb{R}^n}(\pi_1(\mathbf{x},\mathbf{y},\mathbf{r}),\dots,\pi_l(\mathbf{x},\mathbf{y},\mathbf{r}),\|\mathbf{y}\|^2/2)
\end{equation*}
for each $(\mathbf{x},\mathbf{r})$.
Here, we introduced an axiliary objective function $\|\mathbf{y}\|^2/2$ to consider regularization of the inverse kinematics solution near singularities.
Among those inverse kinematics solutions, we put a special emphasis on the case that $\mathbf{R}(\mathbf{x})\mathbf{u}(\mathbf{x},\mathbf{r})$ is the optimal solution of the multi-objective optimization with the lexicographical ordering
\begin{equation}
	\label{eqn:multi-objective_optimization_with_lexicographical_ordering}
	\lexmin_{\mathbf{y}\in\mathbb{R}^n}(\pi_1(\mathbf{x},\mathbf{y},\mathbf{r}),\dots,\pi_l(\mathbf{x},\mathbf{y},\mathbf{r}),\|\mathbf{y}\|^2/2)
\end{equation}
because it allows us to consider priority relations between multiple tasks; the first task has priority over the second, the second has priority over the third, and so on.
We call the inverse kinematics with prioritized multiple tasks as the prioritized inverse kinematics (PIK).
However, not every objective function $\boldsymbol{\pi}$ is proper for the PIK problem.
For example, if $\pi_1(\mathbf{x},\mathbf{y},\mathbf{r}) = \|\mathbf{y}\|$, then we have a trivial optimal solution $\mathbf{y}^* = \mathbf{0}$ of \eqref{eqn:multi-objective_optimization_with_lexicographical_ordering} regardless of the choice of $\pi_2,\dots,\pi_l$.
So, An and Lee \cite{An2019}\cite{An2019a} proposed three properties (dependence, uniqueness, and representation) for an objective function to be (strongly or weakly) proper for the PIK problem.
We say that a map $\mathbf{u}$ is a PIK solution if there exists a proper objective function $\boldsymbol{\pi}$ such that $\mathbf{R}(\mathbf{x})\mathbf{u}(\mathbf{x},\mathbf{r})$ is the optimal solution of \eqref{eqn:multi-objective_optimization_with_lexicographical_ordering} for each $(\mathbf{x},\mathbf{r})$.
In this paper, we consider a class of PIK solutions that can be written as
\begin{equation}
\label{eqn:class_of_pik_solutions}
\mathbf{u} = \mathbf{R}^{-1}\hat{\mathbf{J}}^T\mathbf{C}_D^T\underbrace{\begin{bmatrix} \mathbf{L}_{11} & \mathbf{0} & \cdots & \mathbf{0} \\ \mathbf{L}_{21} & \mathbf{L}_{22} & \cdots & \mathbf{0} \\ \vdots & \vdots & \ddots & \vdots \\ \mathbf{L}_{l1} & \mathbf{L}_{l2} & \cdots & \mathbf{L}_{ll} \end{bmatrix}}_{\mathbf{L} = [\mathbf{L}_{ij}]:X\to\mathbb{R}^{m\times m}}\mathbf{r}'
\end{equation}
where $\hat{\mathbf{J}}:X\to\mathbb{R}^{m\times n}$ is the top ($m\times n$) block of $\hat{\mathbf{J}}_e$, $\mathbf{C}_D = \mathrm{diag}(\mathbf{C}_{11},\dots,\mathbf{C}_{ll}):X\to\mathbb{R}^{m\times m}$ is block diagonal whose diagonal blocks are $\mathbf{C}_{11},\dots,\mathbf{C}_{ll}$ starting from the top left corner, and $\mathbf{L} = [\mathbf{L}_{ij}]:X\to\mathbb{R}^{m\times m}$ is block lower triangular with $\mathbf{L}_{ab}:X\to\mathbb{R}^{m_a\times m_b}$.
Some examples of PIK solutions that can be written as \eqref{eqn:class_of_pik_solutions} can be found from \cite{An2019}\cite{An2019a}.

\subsection{Problem Statement}

The inverse kinematics with multiple tasks in the velocity level can be used to solve the inverse kinematics with multiple tasks in the position level.
Assume that the forward kinematic function $\mathbf{f}:X\to\mathbb{R}^m$ and the desired task trajectory $\mathbf{p}:\mathbb{R}\to\mathbb{R}^m$ are differentiable on $X$ and $\mathbb{R}$, respectively.
Then, we can define the velocity mapping function as
\begin{equation*}
	\mathbf{F} = \begin{bmatrix} \mathbf{f}_t & \mathbf{F}_q\end{bmatrix} = \begin{bmatrix} \frac{\partial \mathbf{f}}{\partial t} & \frac{\partial\mathbf{f}}{\partial\mathbf{q}}\end{bmatrix}.
\end{equation*}
A typical choice of the reference is 
\begin{equation}
	\label{eqn:reference_clik}
	\mathbf{r} = \dot{\mathbf{p}} + \mathbf{K}(\mathbf{p} - \mathbf{f})
\end{equation}
where $\mathbf{K} = \mathrm{diag}(k_1\mathbf{I}_{m_1},\dots,k_l\mathbf{I}_{m_l})\in\mathbb{R}^{m\times m}$ with $k_a\in(0,\infty)$ is the feedback gain matrix \cite{Chiacchio1991} and $\mathbf{I}_{m_a}\in\mathbb{R}^{m_a\times m_a}$ is the identity matrix.
The preconditioner function $\mathbf{R}:X\to GL_n(\mathbb{R})$ can be chosen arbitrarily.
Let $\mathbf{u}:X\times(\mathbb{R}^m)^X\to\mathbb{R}^n$ be a PIK solution in the form of \eqref{eqn:class_of_pik_solutions}.
Since the desired task trajectory and the reference are fixed, we may write $\mathbf{e}_a(\mathbf{x}) = \mathbf{e}_a(\mathbf{x},\mathbf{p})$, $\mathbf{e}_a^\mathrm{res}(\mathbf{x},\dot{\mathbf{q}}) = \mathbf{e}_a^\mathrm{res}(\mathbf{x},\dot{\mathbf{q}},\mathbf{r})$, and $\mathbf{u}(\mathbf{x}) = \mathbf{u}(\mathbf{x},\mathbf{r})$.
Let $\mathbf{q}:[t_0,\infty)\to\mathbb{R}^n$ be a joint trajectory satisfying the differential equation \eqref{eqn:differential_equation} with an initial value $(t_0,\mathbf{q}_0)\in X$.
The $a$-th reference $\mathbf{r}_a(\mathbf{x})$ consists of the feedforward term $\dot{\mathbf{p}}_a(t)$ and the feedback term $k_a(\mathbf{p}_a(t) - \mathbf{f}_a(\mathbf{x}))$ and $\mathbf{u}(\mathbf{x})$ minimizes residuals $\mathbf{e}_a^\mathrm{res}(\mathbf{x},\dot{\mathbf{q}})$ for $a\in\overline{1,l}$ with respect to $\dot{\mathbf{q}}$ under the priority relations in the sense of \eqref{eqn:multi-objective_optimization_with_lexicographical_ordering}.
So, we may expect
\begin{equation}
	\label{eqn:convergence_of_task_errors}
	\lim_{t\to\infty}\|\mathbf{e}_a(\mathbf{x}(t))\| = 0.
\end{equation}
Also, we may anticipate stability of the nonautonomous system \eqref{eqn:differential_equation} to guarantee that a small change in the initial configuration does not produce large changes in the joint trajectory.
In this paper, we study those issues in both continuous and discrete time in order to provide the theoretical foundation and the practical application of the PIK problem.
We also analyze preconditioning to show that the number of tasks can be increased by preconditioning in discrete time.
When we study those issues, we must keep in mind that existence and uniqueness of the joint trajectory $\mathbf{q}(t)$ should be guaranteed somehow.

We use $\|\cdot\|_p$ to denote the $p$-norm for vectors and the induced $p$-norm for matrices and $\|\cdot\| = \|\cdot\|_2$ for both vectors and matrices.
$\|\cdot\|_F$ is the Frobenius norm for matrices.
For a three dimensional array $\mathbf{M} = [m_{ijk}]\in\mathbb{R}^{a\times b\times c}$, we define $\|\mathbf{M}\| = (\sum_{k=1}^c\|[m_{ijk}]\|^2)^{1/2}$ and $\|\mathbf{M}\|_F = (\sum_{k=1}^c\|[m_{ikj}]\|_F^2)^{1/2}$.
We use $\sigma_i(\cdot)$, $\sigma_{\max}(\cdot)$, and $\sigma_{\min}(\cdot)$ to denote the $i$-th, the maximum, and the minimum singular values of matrices, respectively.
$2^S$, $\mathrm{cl}(S)$, $\mathrm{int}(S)$, and $\mathrm{bd}(S)$ are the power set, the closure, the interior, and the boundary of a set $S$, respectively.
We say that $\mathbf{u}(\mathbf{x})$ is linearly bounded if there exist $\gamma,c\in[0,\infty)$ satisfying $\|\mathbf{u}(t,\mathbf{q})\|\le\gamma\|\mathbf{q}\| + c$ for all $(t,\mathbf{q})\in X$.
We will need the following assumptions and lemma:
\begin{assumption}
	\item\label{ass:p_bounded_and_f_linearly_bounded} $\mathbf{p}$ is bounded and $\mathbf{f}$ is linearly bounded;
	\item\label{ass:dp_F_invR_locally_Lipschitz_and_bounded} $\dot{\mathbf{p}}$, $\mathbf{F}$, and $\mathbf{R}^{-1}$ are locally Lipschitz and bounded.
\end{assumption}
\begin{lemma}
	\label{lem:existence_of_solution_of_differential_equation}
	If $\mathbf{u}:\mathbb{R}\times\mathbb{R}^n\to\mathbb{R}^n$ is continuous and linearly bounded, then for every $(t_0,\mathbf{q}_0)\in\mathbb{R}\times\mathbb{R}^n$ there exists $\mathbf{q}:[t_0,\infty)\to\mathbb{R}^n$ satisfying $\mathbf{q}(t_0) = \mathbf{q}_0$ and $\dot{\mathbf{q}}(t) = \mathbf{u}(t,\mathbf{q}(t))$ for all $t\in(t_0,\infty)$.
	If additionally $\mathbf{u}$ is locally Lipschitz, then the trajectory $\mathbf{q}$ is unique \cite[pp. 178]{Clarke2008}.
\end{lemma}

\begin{figure}[t]
	\centering
	\resizebox{0.96\columnwidth}{!}{%
		\begin{tikzpicture}
		\draw[thick] (0,0) rectangle (3.6,3.6);
		\node[above] at (1.8,3.6) {$\mathbb{R}^n$};
		\draw[thick,fill=gray,fill opacity=0.5] (1.8,1.8) circle (1.4);
		\node at (1.8,2.65) {$\Theta(t)+r_\Theta B_n$};
		\draw[thick,fill=gray,fill opacity=0.5] (1.8,1.5) circle (0.9);
		\node at (1.8,2) {$\Theta(t)$};
		\draw[thick] (4.1,0) rectangle (9.9,3.6);
		\node[above] at (7.0,3.6) {$\mathbb{R}^m$};
		\draw[thick,fill=gray,fill opacity=0.5] (7.0,2) ellipse (1.6 and 1.3);
		\node at (7.0,2.8) {$\mathbf{f}(t,\Theta(t))$};
		\draw[thick,fill=gray,fill opacity=0.5] (8.4,1.5) circle (1.2);
		\node at (8.4,2.1) {$\mathbf{f}(t,\Theta_b(t))$};
		\draw[thick,fill=gray,fill opacity=0.5] (5.6,1.5) circle (1.2);
		\node at (5.0,2.1) {$P(t)$};
		\draw[fill] (1.8,1.5) circle (0.03);
		\node[below] at (1.8,1.5) {$\mathbf{q}$};
		\draw[fill] (6.1,1.7) circle (0.03);
		\node[below] at (6.1,1.7) {$\mathbf{f}(t,\mathbf{q})$};
		\draw[-latex] (1.8,1.5) to[bend left=5] (6.1,1.7);
		\draw[fill] (2.75,1) circle (0.03);
		\node[below right] at (2.75,1) {$\mathbf{q}'$};
		\draw[fill] (7.8,1.4) circle (0.03);
		\node[below right] at (7.8,1.4) {$\mathbf{f}(t,\mathbf{q}')$};
		\draw[-latex] (2.75,1) to[bend left=-10] (7.8,1.4);
		\draw[-latex] (2.2,3.8) to node[pos=0.37,above]{$\mathbf{f}(t,\cdot)$} (6.6,3.8);
		\end{tikzpicture}
	}
	\caption{A diagram that shows relations between sets when \ref{ass:set_valued_map_Theta} holds where $P(t) = \bigtimes_{b=1}^l(\mathbf{p}_b(t) + \theta_bB_{m_b})$ and $\Theta_b(t) = (\Theta(t)+r_\Theta B_n)\setminus\Theta(t)$.}
	\label{fig:diagram_for_relations_between_sets}
\end{figure}
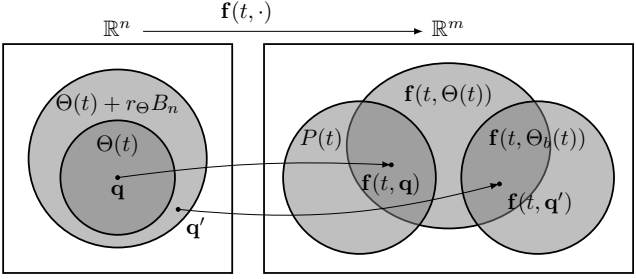

\section{Continuous Time}
\label{sec:continuous_time}

It is not always the case that a PIK solution in the form of \eqref{eqn:class_of_pik_solutions} is locally Lipschitz such that existence and uniqueness of a solution of \eqref{eqn:differential_equation} can be guaranteed by classical theorems such as Lemma \ref{lem:existence_of_solution_of_differential_equation}.
The reason comes from the fact that even if we assume that $\mathbf{J}$ is locally Lipschitz on $X$, its orthogonalization \eqref{eqn:orthogonalization_of_J} can be discontinuous such that a PIK solution in the form of \eqref{eqn:class_of_pik_solutions} has a source of discontinuity $\hat{\mathbf{J}}^T\mathbf{C}_D^T$ \cite{An2019a}.
We may solve this problem by constraining the joint trajectory on a subset of $X$ on which the PIK solution is locally Lipschitz.

\subsection{Domain Restriction}

Define $\dot{\mathbf{p}}_a' = \dot{\mathbf{p}}_a - \mathbf{f}_{ta}$, $\mathbf{A}_{ab} = \sum_{i=b}^a\mathbf{C}_{ai}\mathbf{C}_{ii}^T\mathbf{L}_{ib}$, and $\mathbf{b}_a = \dot{\mathbf{p}}_a' - \sum_{b=1}^a\mathbf{A}_{ab}\dot{\mathbf{p}}_b' - \sum_{b=1}^{a-1}k_b\mathbf{A}_{ab}\mathbf{e}_b$ for $1\le b\le a\le l$.
Note that $\mathbf{A}_{ab}:X\to\mathbb{R}^{m_a\times m_b}$ is the $(a,b)$-th block of $\mathbf{A} = \mathbf{C}\mathbf{C}_D^T\mathbf{L}$ for $a,b\in\overline{1,l}$.
By differentiating the $a$-th task error $\mathbf{e}_a(t,\mathbf{q}) = \mathbf{p}_a(t) - \mathbf{f}_a(t,\mathbf{q})$ with respect to $t$ and letting $\dot{\mathbf{q}} = \mathbf{u}(t,\mathbf{q})$, we can formulate the $a$-th error dynamics as
\begin{equation}
	\label{eqn:error_dynamics}
	\dot{\mathbf{e}}_a = -k_a\mathbf{A}_{aa}\mathbf{e}_a + \mathbf{b}_a.
\end{equation}
Assume that
\begin{assumption}
	\item\label{ass:set_valued_map_Theta} there exist $r_\Theta,\theta_1,\dots,\theta_l,\omega_1,\dots,\omega_l\in(0,\infty)$ and a continuous set-valued map $\Theta:\mathbb{R}\to2^{\mathbb{R}^n}$ such that for every $a\in\overline{1,l}$, $t\in\mathbb{R}$, and $\mathbf{q}\in\Theta(t)$
	\begin{enumerate}
		\item $\Theta(t)$ is closed and nonempty;
		\item $\bigtimes_{b=1}^l(\mathbf{p}_b(t) + \theta_bB_{m_b}) \cap \mathbf{f}(t,\Theta_b(t)) = \emptyset$;
		\item $(\mathbf{A}_{aa} + \mathbf{A}_{aa}^T)(t,\mathbf{q}) - 2\omega_a\mathbf{I}_{m_a} \ge 0$;
	\end{enumerate}
	\item\label{ass:L_is_locally_Lipschitz_and_bounded} $\mathbf{L}$ is locally Lipschitz and bounded on $\mathrm{gr}(\Theta)$
\end{assumption}
where $B_{m_a}$ is the closed unit ball in $\mathbb{R}^{m_a}$, $\mathbf{p}_a(t) + \theta_aB_{m_a} = \{\mathbf{p}_a'\in\mathbb{R}^{m_a}\mid \|\mathbf{p}_a(t) - \mathbf{p}_a'\|\le \theta_a\}$, $\Theta_b(t) = (\Theta(t)+r_\Theta B_n)\setminus\Theta(t)$, $\bigtimes_{b=1}^lA_b = A_1\times\cdots\times A_l$ is the Cartesian product of sets $A_1,\dots,A_l$, and $\mathrm{gr}(\Theta) = \{(t,\mathbf{q})\in X\mid \mathbf{q}\in\Theta(t)\}$.
We provide a diagram in Figure \ref{fig:diagram_for_relations_between_sets} in order to ease understanding of relations between sets.
We call $\Theta(t)$ and $\mathbf{f}(t,\Theta(t))\setminus\mathbf{f}(t,\Theta_b(t))$ a nonsingular configuration space and nonsingular task spaces, respectively.
Since $\Theta$ is assumed to be continuous, whenever a continuous joint trajectory $\mathbf{q}(t)$ leaves $\Theta(t)$, the task trajectories $\mathbf{f}(t,\mathbf{q}(t))$ move from $\mathbf{f}(t,\Theta(t))\setminus\mathbf{f}(t,\Theta_b(t))$ to $\mathbf{f}(t,\Theta_b(t))$.
Therefore, if we can confine task trajectories on nonsingular task spaces, then we can guarantee that a joint trajectory stays in a nonsingular configuration space.

We prove that $\mathbf{u}$ is locally Lipschitz and linearly bounded on $\mathrm{gr}(\Theta)$.
It is immediate to check that $\mathbf{u}$ is linearly bounded on $\mathrm{gr}(\Theta)$ from \ref{ass:p_bounded_and_f_linearly_bounded}, \ref{ass:dp_F_invR_locally_Lipschitz_and_bounded}, \ref{ass:L_is_locally_Lipschitz_and_bounded}, and 
\begin{equation*}
	\|(\mathbf{C}_D\hat{\mathbf{J}})(\mathbf{x})\| \le \|\mathbf{C}_D(\mathbf{x})\| \le \|\mathbf{C}(\mathbf{x})\| = \|\mathbf{J}(\mathbf{x})\|.
\end{equation*}
Since the composition of finitely many Lipschitz functions is Lipschitz, it is sufficient to show that $\mathbf{p} - \mathbf{f}$ and $\mathbf{C}_D\hat{\mathbf{J}}$ are locally Lipschitz on $\mathrm{gr}(\Theta)$ for the rest of the proof.
$\mathbf{p} - \mathbf{f}$ is locally Lipschitz on $X$ because differentiable functions are locally Lipschitz.
$\det(\mathbf{C}(\mathbf{x}))>0$ for all $\mathbf{x}\in\mathrm{gr}(\Theta)$ because
\begin{equation}
	\label{eqn:inequality_of_minimum_singular_value_of_C_aa}
	0 < \omega_a \le \frac{\sigma_{\min}((\mathbf{A}_{aa} + \mathbf{A}_{aa}^T)(\mathbf{x}))}{2} 
	\le \sigma_{\min}^2(\mathbf{C}_{aa}(\mathbf{x}))\|\mathbf{L}_{aa}(\mathbf{x})\|
\end{equation}
for all $a\in\overline{1,l}$ and $\mathbf{x}\in\mathrm{gr}(\Theta)$.
Since $\mathbf{J} = \mathbf{F}_q\mathbf{R}^{-1}$ is locally Lipschitz on $X$ by \ref{ass:dp_F_invR_locally_Lipschitz_and_bounded}, $\mathbf{C}_{aa}\hat{\mathbf{J}}_a$ is locally Lipschitz on $\mathrm{gr}(\Theta)$ for all $a\in\overline{1,l}$ by \cite{An2019a} (see the discussion after Theorem 24 in \cite{An2019a}), and so is $\mathbf{C}_D\hat{\mathbf{J}}$.
It completes the proof.
In the followings, we will construct a map $\tilde{\mathbf{u}}:X\to\mathbb{R}^n$ that is continuous and linearly bounded on $X$ and satisfies $\tilde{\mathbf{u}}(\mathbf{x}) = \mathbf{u}(\mathbf{x})$ for all $\mathbf{x}\in\mathrm{gr}(\Theta)$.

\subsection{Continuous Extension}

Let $\{U_i\}_{i\in I}$ be an open cover of a set $Y\subset X$, i.e., $U_i\subset X$ is open and $Y\subset\bigcup_{i\in I}U_i$ where $I$ is an index set.
An open cover $\{U_i\}_{i\in I}$ is said to be locally finite if for each $i\in I$ the set $\{j\in I\mid U_i\cap U_j\neq\emptyset\}$ is finite.
A family of real functions $\{p_i\}_{i\in I}$ defined on $Y$ is called a locally Lipschitz partition of unity subordinate to a locally finite open cover $\{U_i\}_{i\in I}$ of $Y$ if
\begin{enumerate}
	\item $p_i$ is locally Lipschitz on $Y$ for all $i\in I$;
	\item $p_i(\mathbf{x})>0$ for $\mathbf{x}\in U_i\cap Y$ and $p_i(\mathbf{x}) = 0$ for $\mathbf{x}\in Y\setminus U_i$;
	\item $\sum_{i\in I}p_i(\mathbf{x}) = 1$ for all $\mathbf{x}\in Y$.
\end{enumerate}
For $Y,Z\subset X$, the diameter of $Y$ is defined as $\mathrm{diam}(Y) = \sup\{\|\mathbf{y} - \mathbf{y}'\|\mid\mathbf{y},\mathbf{y}'\in Y\}$ and the distance between $Y$ and $Z$ is defined as $d(Y,Z) = \inf\{\|\mathbf{y}-\mathbf{z}\|\mid\mathbf{y}\in Y,\,\mathbf{z}\in Z\}$.
When $Y = \{\mathbf{y}\}$, we write $d(\mathbf{y},Z) = d(Y,Z)$.
We will need the following two Lemmas.
\begin{lemma}
	\label{lem:partition_of_unity}
	Let $Y\subset X$.
	For any locally finite open cover $\{U_i\}_{i\in I}$ of $Y$, there exists a locally Lipschitz partition of unity subordinate to $\{U_i\}_{i\in I}$ \cite[Lemma 2.5]{Smirnov2002}.
\end{lemma}
\begin{lemma}
	\label{lem:open_cover}
	Let $Y\subset X$ be closed.
	There exist $c_1,c_2\in(0,\infty)$ and a locally finite open cover $\{U_i\}_{i\in\mathbb{N}}$ of $X\setminus Y$ such that
	\begin{equation*}
		c_1\mathrm{diam}(U_i) \le d(Y,U_i) \le c_2\mathrm{diam}(U_i)\quad(i\in\mathbb{N})
	\end{equation*}
	and any point $\mathbf{x}\in X\setminus Y$ belongs to at most $12^{n+1}$ open sets $U_i$ \cite[Lemma 2.9 and Lemma 2.10]{Smirnov2002}.
\end{lemma}

Since $\Theta$ is upper semi-continuous with closed values, $\mathrm{gr}(\Theta)$ is closed \cite[Proposition 2.1]{Smirnov2002}.
Thus, there exist $c_1,c_2\in(0,\infty)$ and a locally finite open cover $\{U_i\}_{i\in\mathbb{N}}$ of $X\setminus\mathrm{gr}(\Theta)$ satisfying
\begin{equation*}
	c_1\mathrm{diam}(U_i) \le d(\mathrm{gr}(\Theta),U_i) \le c_2\mathrm{diam}(U_i)
\end{equation*}
for all $i\in\mathbb{N}$ by Lemma \ref{lem:open_cover}.
It follows that there exists a locally Lipschitz partition of unity $\{p_i\}_{i\in\mathbb{N}}$ subordinate to $\{U_i\}_{i\in\mathbb{N}}$ by Lemma \ref{lem:partition_of_unity}.
Since $\mathrm{gr}(\Theta)$ is closed, there exists $\mathbf{x}_i = (t_i,\mathbf{q}_i)\in\mathrm{gr}(\Theta)$ for each $i\in\mathbb{N}$ satisfying $d(\mathrm{gr}(\Theta),U_i) = d(\mathbf{x}_i,U_i)$.
Define $\tilde{\mathbf{u}}:X\to\mathbb{R}^n$ as
\begin{equation*}
	\tilde{\mathbf{u}}(\mathbf{x}) = 
		\begin{dcases*} 
			\mathbf{u}(\mathbf{x}), & $\mathbf{x}\in\mathrm{gr}(\Theta)$ \\
			\sum_{i=1}^\infty \frac{p_i(\mathbf{x})\mathbf{u}(\mathbf{x}_i)}{1 + \|\mathbf{q} - \mathbf{q}_i\|}, & $\mathbf{x}\in X\setminus\mathrm{gr}(\Theta)$.
		\end{dcases*}
\end{equation*}

We prove that $\tilde{\mathbf{u}}$ is continuous and linearly bounded on $X$.
Since $\mathbf{u}$ is linearly bounded on $\mathrm{gr}(\Theta)$, there exist $\gamma,c\in(0,\infty)$ satisfying $\|\tilde{\mathbf{u}}(\mathbf{x})\| \le \gamma\|\mathbf{q}\| + c$ for all $\mathbf{x}\in\mathrm{gr}(\Theta)$.
Then, for every $\mathbf{x}\in X\setminus\mathrm{gr}(\Theta)$ we have
\begin{align*}
	\|\tilde{\mathbf{u}}(\mathbf{x})\| 
		&\le \sum_{i=1}^\infty\frac{p_i(\mathbf{x})(\gamma\|\mathbf{q}_i\| + c)}{1+\|\mathbf{q}-\mathbf{q}_i\|} \\
		&\le \sum_{i=1}^\infty p_i(\mathbf{x})\left( \frac{\gamma\|\mathbf{q}\| + c}{1+\|\mathbf{q}-\mathbf{q}_i\|} + \frac{\gamma\|\mathbf{q}-\mathbf{q}_i\|}{1+\|\mathbf{q}-\mathbf{q}_i\|}\right) \\
		&\le \sum_{i=1}^\infty p_i(\mathbf{x})(\gamma\|\mathbf{q}\| + c + \gamma) \\
		&\le 12^{n+1}(\gamma\|\mathbf{q}\| + c + \gamma).
\end{align*}
Therefore, $\tilde{\mathbf{u}}$ is linearly bounded on $X$.

Let $\mathbf{x}\in X\setminus\mathrm{gr}(\Theta)$.
There exists a finite set $I\subset\mathbb{N}$ whose cardinality is at most $12^{n+1}$ such that $\mathbf{x}\not\in U_i$ if $i\not\in I$.
For each $i\in I$, there exist $\delta_i,L_i\in(0,\infty)$ such that $\mathbf{x} + \delta_iB_X\subset U_i$ and $|p_i(\mathbf{x}') - p_i(\mathbf{x}'')| \le L_i\|\mathbf{x}' - \mathbf{x}''\|$ for all $\mathbf{x}',\mathbf{x}''\in\mathbf{x}+\delta_iB_X$ where $B_X$ is the closed unit ball in $X$.
Let $\delta = \min\{\delta_i\mid i\in I\}$ and $L = \max\{L_i\mid i\in I\}$.
Then, we have
\begin{align*}
&\|\tilde{\mathbf{u}}(\mathbf{x}') - \tilde{\mathbf{u}}(\mathbf{x}'')\| \\
&\le \sum_{i\in I} \|\mathbf{u}(\mathbf{x}_i)\|\left( \frac{|p_i(\mathbf{x}') - p_i(\mathbf{x}'')|}{1 + \|\mathbf{q}' - \mathbf{q}_i\|}\right. \\
&\quad \left.+ p_i(\mathbf{x}'')\left|\frac{1}{1 + \|\mathbf{q}' - \mathbf{q}_i\|} - \frac{1}{1 + \|\mathbf{q}'' - \mathbf{q}_i\|}\right|\right) \\
&\le \sum_{i\in I} \|\mathbf{u}(\mathbf{x}_i)\|( |p_i(\mathbf{x}') - p_i(\mathbf{x}'')| + \|\mathbf{q}' - \mathbf{q}''\|) \\
&\le \sum_{i\in I} \|\mathbf{u}(\mathbf{x}_i)\|(L+1)\|\mathbf{x}' - \mathbf{x}''\|
\end{align*}
for all $\mathbf{x}' = (t',\mathbf{q}'), \mathbf{x}'' = (t'',\mathbf{q}'') \in \mathbf{x} + \delta B_X$.
Therefore, $\tilde{\mathbf{u}}$ is locally Lipschitz on $X\setminus\mathrm{gr}(\Theta)$.
Obviously, $\tilde{\mathbf{u}}$ is locally Lipschitz on $\mathrm{int}(\mathrm{gr}(\Theta))$, so we need to prove that $\tilde{\mathbf{u}}$ is continuous at each $\mathbf{x}\in\mathrm{bd}(\mathrm{gr}(\Theta))$ for the rest of the proof.

For each $\mathbf{x}\in\mathrm{gr}(\Theta)$ and $\mathbf{x}'\in U_i$, we can find $c_1\mathrm{diam}(U_i) \le d(\mathrm{gr}(\Theta), U_i) \le \|\mathbf{x} - \mathbf{x}'\|$, $\|\mathbf{x}' - \mathbf{x}_i\| \le d(\mathrm{gr}(\Theta),U_i) + \mathrm{diam}(U_i) \le (1 + 1/c_1)d(\mathrm{gr}(\Theta),U_i)$, and
\begin{align*}
\|\mathbf{x} - \mathbf{x}_i\| &\le \|\mathbf{x} - \mathbf{x}'\| + \|\mathbf{x}' - \mathbf{x}_i\| \\
&\le \|\mathbf{x} - \mathbf{x}'\| + \left(1 + \frac{1}{c_1}\right)d(\mathrm{gr}(\Theta),U_i) \\
&\le \|\mathbf{x} - \mathbf{x}'\| + c_2\left(1 + \frac{1}{c_1}\right)\mathrm{diam}(U_i) \\
&\le c\|\mathbf{x} - \mathbf{x}'\|
\end{align*}
where $c = 1 + (c_2/c_1)(1 + 1/c_1)$.
Fix $\mathbf{x}\in\mathrm{bd}(\mathrm{gr}(\Theta))$ and $\epsilon\in(0,\infty)$.
There exist $\delta,L\in(0,\infty)$ such that $\|\tilde{\mathbf{u}}(\mathbf{x}') - \tilde{\mathbf{u}}(\mathbf{x}'')\| \le L\|\mathbf{x}' - \mathbf{x}''\|$ for all $\mathbf{x}',\mathbf{x}''\in(\mathbf{x}+\delta B_X)\cap\mathrm{gr}(\Theta)$.
For every $\mathbf{x}'\in(\mathbf{x} + \delta/c B_X)\setminus\mathrm{gr}(\Theta)$ and $U_i$ containing $\mathbf{x}'$, we have $\|\mathbf{x} - \mathbf{x}_i\| \le c\|\mathbf{x} - \mathbf{x}'\| \le \delta$.
Then,
\begin{align*}
	&\|\tilde{\mathbf{u}}(\mathbf{x}') - \tilde{\mathbf{u}}(\mathbf{x})\| \\
	&\le \sum_{i=1}^\infty p_i(\mathbf{x}')\left\|\frac{\mathbf{u}(\mathbf{x}_i)}{1+\|\mathbf{q}' - \mathbf{q}_i\|} - \mathbf{u}(\mathbf{x})\right\| \\
	&\le \sum_{i=1}^\infty p_i(\mathbf{x}') (L\|\mathbf{x}_i - \mathbf{x}\| + \|\mathbf{u}(\mathbf{x})\|\|\mathbf{x}' - \mathbf{x}_i\|) \\
	&\le 12^{n+1}(cL + (1+c)\|\mathbf{u}(\mathbf{x})\|)\|\mathbf{x}' - \mathbf{x}\| \\
	&= L'\|\mathbf{x}' - \mathbf{x}\|
\end{align*}
for all $\mathbf{x}'\in(\mathbf{x}+\delta/c B_X)\setminus\mathrm{gr}(\Theta)$.
By letting $0 < \delta' < \min\{\delta, \epsilon/L, \delta/c, \epsilon/L'\}$, we have $\|\tilde{\mathbf{u}}(\mathbf{x}') - \tilde{\mathbf{u}}(\mathbf{x})\| < \epsilon$ for all $\mathbf{x}' \in\mathbf{x}+\delta'B_X$.
Therefore, $\tilde{\mathbf{u}}$ is continuous on $X$.

\subsection{Trajectory Existence}

Since $\tilde{\mathbf{u}}$ is continuous and linearly bounded, for every $(t_0,\mathbf{q}_0)\in X$ there exists a joint trajectory $\mathbf{q}:[t_0,\infty)\to\mathbb{R}^n$ satisfying $\mathbf{q}(t_0) = \mathbf{q}_0$ and $\dot{\mathbf{q}}(t) = \tilde{\mathbf{u}}(t,\mathbf{q}(t))$ for all $t\in(t_0,\infty)$ by Lemma \ref{lem:existence_of_solution_of_differential_equation}.
In the followings, we will find a condition for $\mathbf{q}(t)\in\Theta(t)$ for all $t\in[t_0,\infty)$ such that $\dot{\mathbf{q}}(t) = \mathbf{u}(t,\mathbf{q}(t))$ for all $t\in(t_0,\infty)$.
If we find such a condition, then uniqueness of the joint trajectory will also be guaranteed in a similar way to \cite[Theorem 2.2]{Teschl2012}.

Let $\theta_a'\in(0,\theta_a]$ be arbitrary for $a\in\overline{1,l}$ and assume
\begin{assumption}
	\item\label{ass:initial_condition} $\mathbf{q}_0\in\Theta(t_0)$ and $\|\mathbf{p}_a(t_0)-\mathbf{f}_a(t_0,\mathbf{q}_0)\|<\theta_a'$ for all $a\in\overline{1,l}$.
\end{assumption}
The reason we introduce $\theta_a'$ will be apparent later in Section \ref{sec:discrete_time}, when we consider discrete time.
Let $\mathbf{q}:[t_0,\infty)\to\mathbb{R}^n$ be a joint trajectory satisfying $\mathbf{q}(t_0) = \mathbf{q}_0$ and $\dot{\mathbf{q}}(t) = \tilde{\mathbf{u}}(t,\mathbf{q}(t))$ for all $t\in(t_0,\infty)$.
Suppose that there exists $t_1\in(t_0,\infty)$ such that $\|\mathbf{p}_a(t) - \mathbf{f}_a(t,\mathbf{q}(t))\|<\theta_a'$ for all $a\in\overline{1,l}$ and $t\in[t_0,t_1)$ and $\|\mathbf{p}_a(t_1) - \mathbf{f}_a(t_1,\mathbf{q}(t_1))\| = \theta_a'$ for at least one $a\in\overline{1,l}$.
It is not difficult to see that $\mathbf{q}(t)\in\Theta(t)$ for all $t\in[t_0,t_1]$ from the fact that the map $t\mapsto d(\mathbf{q}(t),\Theta(t))$ is continuous on $[t_0,t_1]$ by \ref{ass:set_valued_map_Theta} \cite[Lemma 2.2]{Smirnov2002}.
Define $\phi_a,\rho_a,\gamma_a:[t_0,t_1]\to\mathbb{R}$ for $a\in\overline{1,l}$ as:
\begin{align}
	\phi_a(t) &= \|\mathbf{e}_a(\mathbf{x}(t))\| \label{eqn:phi_a} \\
	\rho_a(t) &= k_a\phi_a^{+2}(t)\langle \mathbf{e}_a(\mathbf{x}(t)),(\mathbf{A}_{aa}\mathbf{e}_a)(\mathbf{x}(t))\rangle \label{eqn:rho_a} \\
	\gamma_a(t) &= \phi_a^+(t)\langle \mathbf{e}_a(\mathbf{x}(t)),\mathbf{b}_a(\mathbf{x}(t))\rangle \label{eqn:gamma_a}
\end{align}
where $\phi_a^+(t) = 0$ if $\phi_a(t) = 0$ and $\phi_a^+(t) = 1/\phi_a(t)$ if $\phi_a(t)\neq 0$. 
From \ref{ass:dp_F_invR_locally_Lipschitz_and_bounded}, \ref{ass:L_is_locally_Lipschitz_and_bounded}, and \eqref{eqn:error_dynamics}, we can check that $\phi_a$ is absolutely continuous on $[t_0,t_1]$; $\dot{\phi}_a(t) = -\rho_a(t)\phi_a(t) + \gamma_a(t)$ for almost all $t\in[t_0,t_1]$; and $\rho_a$ and $\gamma_a$ are integrable on $[t_0,t_1]$ for all $a\in\overline{1,l}$ (in the sense of Lebesgue).
By \ref{ass:set_valued_map_Theta}, $\rho_a(t) \ge k_a\omega_a$ for all $t\in[t_0,t_1]$.
Since $\|\mathbf{b}_a(\mathbf{x}(\cdot))\|$ is continuous on the compact set $[t_0,t_1]$, there exists $\beta_a = \max\{\|\mathbf{b}_a(\mathbf{x}(t))\|\mid t\in[t_0,t_1]\}<\infty$ for $a\in\overline{1,l}$.
Then, we have
\begin{align}
	\phi_a(t) &= \phi_a(t_0)e^{-\int_{t_0}^t\rho_a(s)ds} + \int_{t_0}^t\gamma_a(s)e^{-\int_s^t\rho_a(r)dr}ds \nonumber\\
		&< \left( \theta_a' - \frac{\beta_a}{k_a\omega_a} \right) e^{-k_a\omega_a(t-t_0)} + \frac{\beta_a}{k_a\omega_a}\label{eqn:phi_a_inequality}
\end{align}
for all $a\in\overline{1,l}$ and $t\in[t_0,t_1]$.

By \ref{ass:dp_F_invR_locally_Lipschitz_and_bounded} and \ref{ass:L_is_locally_Lipschitz_and_bounded}, there exist $M_a,M_{ab}\in(0,\infty)$ for $a,b\in\overline{1,l}$ satisfying 
\begin{align*}
	\|\mathbf{A}_{ab}(\mathbf{x})\| &\le M_{ab} \\
	\left\|\left(\dot{\mathbf{p}}_a' - \sum_{b=1}^a\mathbf{A}_{ab}\dot{\mathbf{p}}_b'\right)(\mathbf{x})\right\| &\le M_a
\end{align*}
for all $\mathbf{x}\in\mathrm{gr}(\Theta)$.
It follows that
\begin{equation*}
	\beta_a \le M_a + \sum_{b=1}^{a-1}k_bM_{ab}\max_{t\in[t_0,t_1]}\phi_b(t) \quad (a\in\overline{1,l}).
\end{equation*}
Assume
\begin{assumption}
	\item\label{ass:lower_bound_of_k_a} $k_a > (M_a + \sum_{b=1}^{a-1}k_bM_{ab}\theta_b)/(\theta_a'\omega_a)$ for all $a\in\overline{1,l}$.
\end{assumption}
By repeatedly applying \ref{ass:lower_bound_of_k_a} to \eqref{eqn:phi_a_inequality} from $a = 1$ to $a = l$, we find that $\beta_a \le M_a + \sum_{b=1}^{a-1}k_bM_{ab}\theta_b$, $\theta_a' - \beta_a/(k_a\omega_a)>0$, and $\phi_a(t) < \theta_a'$ for all $t\in[t_0,t_1]$, a contradiction that $\phi_a(t_1) = \theta_a'$ for at least one $a\in\overline{1,l}$.
Therefore, $\|\mathbf{p}_a(t) - \mathbf{f}_a(t,\mathbf{q}(t))\|<\theta_a'$ and $\mathbf{x}(t)\in\mathrm{gr}(\Theta)$ for all $a\in\overline{1,l}$ and $t\in[t_0,\infty)$.
As we discussed earlier, it implies that the joint trajectory uniquely satisfies $\dot{\mathbf{q}}(t) = \mathbf{u}(t,\mathbf{q}(t))$ for all $t\in(t_0,\infty)$.

\subsection{Task Convergence}

Assume that
\begin{assumption}
	\item\label{ass:bounded_integral_from_t_0_to_infinity} 
	there exist $\psi_a:\mathbb{R}\to[0,\infty)$ such that $\int_{-\infty}^\infty\psi_a(t)dt<\infty$ and $\|\dot{\mathbf{p}}_a(t) - \mathbf{f}_{ta}(t,\mathbf{q})\| \le \psi_a(t)$ for all $a\in\overline{1,l}$ and $(t,\mathbf{q})\in \mathrm{gr}(\Theta)$.
\end{assumption}
Define $\bar{\phi}_a,\bar{\gamma}_a:[t_0,\infty)\to[0,\infty)$ for $a\in\overline{1,l}$ as:
\begin{align*}
	\bar{\gamma}_a(t) &= \sum_{b=1}^{a-1} M_{ab}(\psi_b(t) + k_b\bar{\phi}_b(t)) + (1+M_{aa})\psi_a(t) \\
	\bar{\phi}_a(t) &= \theta_a'e^{-k_a\omega_a(t-t_0)} + \int_{t_0}^t\bar{\gamma}_a(s)e^{-k_a\omega_a(t-s)}ds.
\end{align*}
Then, $\dot{\bar{\phi}}_a(t) = -k_a\omega_a\bar{\phi}_a(t) + \bar{\gamma}_a(t)$, $|\gamma_a(t)|\le\bar{\gamma}_a(t)$, and $\phi_a(t) \le \bar{\phi}_a(t)$ for all $a\in\overline{1,l}$ and $t\in[t_0,\infty)$.
Define $\Psi_a = \int_{-\infty}^\infty\psi_a(t)dt$, $\bar{\Gamma}_a = \sum_{b=1}^{a-1} M_{ab}(\Psi_b + k_b\bar{\Phi}_b) + (1+M_{aa})\Psi_a$, and $\bar{\Phi}_a = (\bar{\Gamma}_a + \theta_a')/(k_a\omega_a)$
for $a\in\overline{1,l}$.
Then, $\int_{t_0}^t\bar{\gamma}_a(s)ds \le \bar{\Gamma}_a$ and $\int_{t_0}^t\bar{\phi}_a(s)ds = (\int_{t_0}^t\bar{\gamma}_a(s)ds - \bar{\phi}_a(t) + \bar{\phi}_a(t_0))/(k_a\omega_a) \le \bar{\Phi}_a$ for all $a\in\overline{1,l}$ and $t\in[t_0,\infty)$.
We can check
\begin{equation*}
	\lim_{t\to\infty}\int_{t_0}^t\psi_b(s)e^{-k_a\omega_a(t-s)}ds = 0 \quad (a,b\in\overline{1,l})
\end{equation*}
from the Lebesgue's dominated convergence theorem \cite[Theorem 1.34]{Rudin1987}; the integrand is dominated by $\psi_b$ and converges to 0 as $t\to\infty$ for each $s$.
By the same reason, $\int_{t_0}^t\bar{\gamma}_b(s)e^{-k_a\omega_a(t-s)}ds \to0$ and $\int_{t_0}^t\bar{\phi}_b(s)e^{-k_a\omega_a(t-s)}ds\to0$ as $t\to\infty$ for all $a,b\in\overline{1,l}$.
Therefore, $\phi_a(t) \le \bar{\phi}_a(t) \to 0$ as $t\to\infty$ for all $a\in\overline{1,l}$.

We prove that there exists $\mathbf{q}_\infty\in\mathbb{R}^n$ such that $\|\mathbf{q}(t) - \mathbf{q}_{\infty}\|\to0$ as $t\to\infty$.
By \ref{ass:dp_F_invR_locally_Lipschitz_and_bounded} and \ref{ass:L_is_locally_Lipschitz_and_bounded}, there exists $M\in(0,\infty)$ satisfying
\begin{equation*}
	\|(\mathbf{R}^{-1}\hat{\mathbf{J}}^T\mathbf{C}_D^T\mathbf{L})(\mathbf{x})\| \le M \quad (\mathbf{x}\in \mathrm{gr}(\Theta)).
\end{equation*}
Let $t_1 < t_2 < t_3 < \cdots$ be an arbitrary divergent sequence and define $u_i = \int_{t_0}^{t_i}\|\mathbf{u}(t,\mathbf{q}(t))\|dt$ for $i\in\mathbb{N}$.
Let $\epsilon>0$ be arbitrary.
Since $u_i\to u_\infty = \int_{t_0}^\infty\|\mathbf{u}(t,\mathbf{q}(t))\|dt \le M\int_{t_0}^\infty\left(\sum_{a=1}^l(\psi_a(t) + k_a\phi_a(t))^2\right)^{1/2}dt <\infty$ as $i\to\infty$, there exists $N\in\mathbb{N}$ such that $|u_i - u_\infty|<\epsilon/2$ for all $i>N$.
Then, for every $i,j\in\mathbb{N}\setminus\overline{1,N}$ we have $\|\mathbf{q}(t_i) - \mathbf{q}(t_j)\| \le \left|\int_{t_i}^{t_j}\|\mathbf{u}(t,\mathbf{q}(t))\|dt\right| = |u_i - u_j| \le |u_i - u_\infty| + |u_j - u_\infty| < \epsilon$.
Therefore, $\{\mathbf{q}(t_i)\}$ converges in $\mathbb{R}^n$ because it is Cauchy \cite[Theorem 3.11]{Rudin1964}.
Since it holds for every divergent sequence $\{t_i\}$, $\mathbf{q}(t)$ converges to a point $\mathbf{q}_\infty\in\mathbb{R}^n$ as $t\to\infty$.

\subsection{Stability}

Since $\mathbf{f}$ is continuously differentiable and the derivative is bounded, $\mathbf{f}$ is Lipschitz on $X$ with a Lipschitz constant $L_f\in(0,\infty)$.
Since $\phi_a(t) < \theta_a'$ and $\phi_a(t) \to 0$ as $t\to\infty$, there exists $\hat{\phi}_a = \max_{t\in[t_0,\infty)}\phi_a(t) < \theta_a'$.
Let
\begin{equation*}
	0 < \delta_0 < \min\{(\theta_a' - \hat{\phi}_a)/L_f\mid a\in\overline{1,l}\}.
\end{equation*} 
Then,
\begin{equation*}
	\|\mathbf{p}_a(t) - \mathbf{f}_a(t,\mathbf{q}')\| \le \hat{\phi}_a + L_f\|\mathbf{q}(t) - \mathbf{q}'\| < \theta_a'
\end{equation*}
for all $a\in\overline{1,l}$, $t\in[t_0,\infty)$, and $\mathbf{q}'\in\mathbf{q}(t) + \delta_0 B_n$.
Define $\mathbf{v}:[0,\infty)\times \delta_0 B_n\to\mathbb{R}^n$ as
\begin{equation*}
	\mathbf{v}(s,\mathbf{z}) = \mathbf{u}(s + t_0, \mathbf{z} + \mathbf{q}(s + t_0)) - \mathbf{u}(s + t_0, \mathbf{q}(s+t_0)).
\end{equation*}
By letting $\mathbf{z} = \mathbf{q}' - \mathbf{q}(t)$ and $s = t - t_0$, we can transform the original system $\dot{\mathbf{q}} = \mathbf{u}(t,\mathbf{q})$ into 
\begin{equation}
	\label{eqn:transformed_system}
	\dot{\mathbf{z}} = \mathbf{v}(s,\mathbf{z})
\end{equation} 
in the vicinity of the joint trajectory $\mathbf{q}(t)$
and find that $\mathbf{z} = \mathbf{0}$ is the equilibrium point of the nonautonomous system \eqref{eqn:transformed_system} at $s = 0$, i.e., $\mathbf{v}(s,\mathbf{0}) = \mathbf{0}$ for all $s\in[0,\infty)$.
For every $(s_1,\mathbf{z}_1) = (t_1 - t_0, \mathbf{q}_1' - \mathbf{q}(t_1)) \in [0,\infty)\times\delta_0B_n$, we have $(t_1,\mathbf{q}_1')\in[t_0,\infty)\times(\mathbf{q}(t)+\delta_0B_n)$ and $\|\mathbf{p}_a(t_1) - \mathbf{f}_a(t_1,\mathbf{q}_1')\|<\theta_a'$ for all $a\in\overline{1,l}$.
Since $(t_1,\mathbf{q}_1')$ satisfies the assumption on the intial condition \ref{ass:initial_condition}, there exists a unique $\mathbf{q}':[t_1,\infty)\to\mathbb{R}^n$ satisfying $\mathbf{q}'(t_1) = \mathbf{q}_1'$ and $\dot{\mathbf{q}}'(t) = \mathbf{u}(t,\mathbf{q}'(t))$ for all $t\in(t_1,\infty)$.
By letting 
\begin{equation*}
	\mathbf{z}(s) = \mathbf{q}'(s+t_0) - \mathbf{q}(s+t_0) = \mathbf{q}'(t) - \mathbf{q}(t)
\end{equation*}
for $s\in[t_1-t_0,\infty)$, we see that \eqref{eqn:transformed_system} has a unique solution $\mathbf{z}(s)$ that can be continued until $s\to\infty$ or it reaches to the boundary of $\delta_0 B_n$.
For each $(t_1,\delta)\in[t_0,\infty)\times(0,\delta_0]$, we define 
\begin{align*}
	S(t_1,\delta) &= \{\mathbf{q}':[t_1,\infty)\to\mathbb{R}^n\mid \mathbf{q}'(t_1)\in\mathbf{q}(t_1) + \delta B_n \\
		&\quad\quad \text{and }\dot{\mathbf{q}}'(t) = \mathbf{u}(t,\mathbf{q}'(t))\text{ for all }t\in(t_1,\infty)\}.
\end{align*}
We study various stability notions of the equilibrium point $\mathbf{z} = \mathbf{0}$ of \eqref{eqn:transformed_system} at $s = 0$; see \cite[Definition 4.4 and 4.5]{Khalil2001} for the definition of stability.
Specifically, we will find conditions for the following stability notions:
\begin{stability}
	\item\label{sta:stability} for every $\epsilon\in(0,\infty)$ and $t_1\in[t_0,\infty)$ there exists $\delta\in(0,\infty)$ such that $\|\mathbf{q}'(t) - \mathbf{q}(t)\| < \epsilon$ for all $\mathbf{q}'\in S(t_1,\delta)$ and $t\in[t_1,\infty)$ (stability);
	\item\label{sta:uniform_stability} for every $\epsilon\in(0,\infty)$ there exists $\delta\in(0,\infty)$ such that $\|\mathbf{q}'(t) - \mathbf{q}(t)\|<\epsilon$ for all $t_1\in[t_0,\infty)$, $\mathbf{q}'\in S(t_1,\delta)$, and $t\in[t_1,\infty)$ \textit{(uniform stability)};
	\item\label{sta:uniform_asymptotic_stability} \ref{sta:uniform_stability} holds and there exists $\delta'\in(0,\infty)$ such that $\|\mathbf{q}'(t) - \mathbf{q}(t)\|\to0$ as $t\to\infty$ for all $t_1\in[t_0,\infty)$ and $\mathbf{q}'\in S(t_1,\delta')$ uniformly in $t_1$, i.e., for each $\epsilon'>0$ there exists $T\in(0,\infty)$ such that $\|\mathbf{q}'(t)-\mathbf{q}(t)\|<\epsilon'$ for all $t_1\in[t_0,\infty)$, $\mathbf{q}'\in S(t_1,\delta')$, and $t\in[t_1 + T,\infty)$ (uniform asymptotic stability);
	\item\label{sta:exponential_stability} there exist $\delta,M,\rho\in(0,\infty)$ such that $\|\mathbf{q}'(t) - \mathbf{q}(t)\|\le M\|\mathbf{q}'(t_1) - \mathbf{q}(t_1)\|e^{-\rho(t-t_1)}$ for all $t_1\in[t_0,\infty)$, $\mathbf{q}'\in S(t_1,\delta)$, and $t\in[t_1,\infty)$ (exponential stability).
\end{stability}

\subsubsection{Stability}

For each $t_1\in[t_0,\infty)$ and $\mathbf{q}'\in S(t_1,\delta_0)$, we define $\phi_a',\rho_a',\gamma_a':[t_1,\infty)\to\mathbb{R}$ for $a\in\overline{1,l}$ as in \eqref{eqn:phi_a} to \eqref{eqn:gamma_a} by replacing $\mathbf{x}(t)$ with $\mathbf{x}'(t) = (t,\mathbf{q}'(t))$.
Then, $|\gamma_a'(t)|\le\bar{\gamma}_a(t-t_1+t_0)$ and $\phi_a'(t) \le \bar{\phi}_a(t-t_1+t_0)$ for all $a\in\overline{1,l}$, $t_1\in[t_0,\infty)$, $\mathbf{q}'\in S(t_1,\delta_0)$, and $t\in[t_1,\infty)$.
We also define $\zeta':[t_1,\infty)\to[0,\infty)$ as
\begin{equation}
	\label{eqn:zeta}
	\zeta'(t) = \int_t^\infty\left(\sum_{a=1}^l(\psi_a(s) + k_a\phi_a'(s))^2\right)^{1/2}ds.
\end{equation}
We define $\zeta:[t_1,\infty)\to[0,\infty)$ similarly by replacing $\phi_a'$ with $\phi_a$.
Obviously, $\zeta'$ monotonically decreases and converges to $0$.
Moreover, for every $\epsilon\in(0,\infty)$ there exists $T\in(0,\infty)$ such that $\zeta'(t)<\epsilon$ for all $t_1\in[t_0,\infty)$, $\mathbf{q}'\in S(t_1,\delta_0)$, and $t\in[t_1+T,\infty)$. 
To see this, let $\epsilon\in(0,\infty)$ be arbitrary.
Since $\int_{t_0}^\infty\psi_a(t)dt\le\Psi_a$ and $\int_{t_0}^\infty\bar{\phi}_a(t)dt \le \bar{\Phi}_a$, there exists $t_2\in[t_0,\infty)$ such that $\int_t^\infty\psi_a(s)ds < \epsilon/(2l)$ and $\int_t^\infty\bar{\phi}_a(s)ds < \epsilon/(2k_al)$ for all $a\in\overline{1,l}$ and $t\in[t_2,\infty)$.
Let $T = t_2 - t_0$.
Then, we have $t_2 \le t_1 + T$ and
\begin{align*}
	\zeta'(t) 
		&\le \sum_{a=1}^l\left(\int_t^\infty\psi_a(s)ds + k_a\int_t^\infty\bar{\phi}_a(s - t_1 + t_0)ds \right) \\
		&\le \sum_{a=1}^l\left(\int_{t_2}^\infty\psi_a(s)ds + k_a\int_{t_2}^\infty\bar{\phi}_a(r)dr \right) < \epsilon
\end{align*}
for all $t_1\in[t_0,\infty)$, $\mathbf{q}'\in S(t_1,\delta_0)$, and $t\in[t_1+T,\infty)$.

Let $t_1\in[t_0,\infty)$ be arbitrary.
Then,
\begin{equation*}
	\|\mathbf{q}'(t) - \mathbf{q}(t_1)\| \le r_0 = \delta_0 + M\sum_{a=1}^l(\Psi_a + k_a\bar{\Phi}_a)
\end{equation*}
for all $\mathbf{q}'\in S(t_1,\delta_0)$ and $t\in[t_1,\infty)$.
Let $t_2\in(t_1,\infty)$ be arbitrary and define
\begin{equation*}
	C_{[t_1,t_2]} = \mathrm{gr}(\Theta)\cap ([t_1,t_2]\times(\mathbf{q}(t_1) + r_0B_n)).
\end{equation*}
Obviously, $(t,\mathbf{q}'(t))\in C_{[t_1,t_2]}$ for all $\mathbf{q}'\in S(t_1,\delta_0)$ and $t\in[t_1,t_2]$.
Since $\mathrm{gr}(\Theta)$ is closed, $C_{[t_1,t_2]}$ is compact.
It follows that $\mathbf{u}$ is Lipscitz on $C_{[t_1,t_2]}$ with a Lipschitz constant $L_{[t_1,t_2]}\in(0,\infty)$.
Then, for each $\mathbf{q}'\in S(t_1,\delta_0)$ and $t\in[t_1,t_2]$ we have
\begin{equation*}
	\|\mathbf{q}'(t) - \mathbf{q}(t)\| \le \|\mathbf{q}'(t_1) - \mathbf{q}(t_1)\|e^{L_{[t_1,t_2]}(t-t_1)}
\end{equation*}
by the Gronwall's inequality.
Also, for each $\mathbf{q}'\in S(t_1,\delta_0)$ and $t\in(t_2,\infty)$ we have
\begin{align*}
	&\|\mathbf{q}'(t) - \mathbf{q}(t)\| \\
	&\le \|\mathbf{q}'(t_2) - \mathbf{q}(t_2)\| + \int_{t_2}^t \|\mathbf{u}(s,\mathbf{q}'(s)) - \mathbf{u}(s,\mathbf{q}(s))\|ds \\
	&\le \|\mathbf{q}'(t_1) - \mathbf{q}(t_1)\|e^{L_{[t_1,t_2]}(t_2-t_1)} + M(\zeta'(t_2)+\zeta(t_2)).
\end{align*}
Let $\epsilon\in(0,\infty)$ be arbitrary.
There exist $T\in(0,\infty)$ and $\delta\in(0,\delta_0]$ satisfying
\begin{equation*}
	\delta e^{L_{[t_1,t_2]}T} + M(\zeta'(t_1+T) + \zeta(t_1+T)) < \min\{\delta_0,\epsilon\}
\end{equation*}
for all $\mathbf{q}'\in S(t_1,\delta)$.

\subsubsection{Uniform Stability}

Define $\hat{\Theta}:\mathbb{R}\to2^{\mathbb{R}^n}$ as
\begin{equation*}
\hat{\Theta}(t) = \{\mathbf{q}'\in\Theta(t)\mid\mathbf{f}_a(t,\mathbf{q}')\in\mathbf{p}_a(t)+\theta_aB_{m_a},\,a\in\overline{1,l}\}.
\end{equation*}
Then, $\mathbf{q}'(t)\in\hat{\Theta}(t)$ for all $t_1\in[t_0,\infty)$, $\mathbf{q}'\in S(t_1,\delta_0)$, and $t\in[t_1,\infty)$.
Assume that
\begin{assumption}
	\item\label{ass:dp_F_invR_L_are_Lipschitz} $\dot{\mathbf{p}}$ is Lipschitz on $\mathbb{R}$ and $\mathbf{F}$, $\mathbf{R}^{-1}$, and $\mathbf{L}$ are Lipschitz on $\mathrm{gr}(\Theta)$.
\end{assumption}
Since $\mathbf{p}-\mathbf{f}$ is bounded on $\mathrm{gr}(\hat{\Theta})$, $\mathbf{u}$ is Lipschitz on $\mathrm{gr}(\hat{\Theta})$ with a Lipschitz constant $L_u\in(0,\infty)$.
It follows that
\begin{equation*}
	\|\mathbf{q}'(t) - \mathbf{q}(t)\| \le \|\mathbf{q}'(t_1) - \mathbf{q}(t_1)\|e^{L_u(t_2-t_1)} + M(\zeta'(t_2)+\zeta(t_2))
\end{equation*}
for all $t_1\in[t_0,\infty)$, $\mathbf{q}'\in S(t_1,\delta_0)$, $t_2\in(t_1,\infty)$, and $t\in[t_1,\infty)$.
Let $\epsilon\in(0,\infty)$ be arbitrary.
There exist $T\in(0,\infty)$ and $\delta\in(0,\delta_0]$ satisfying
\begin{equation*}
	\delta e^{L_uT} + M(\zeta'(t_1+T)+\zeta(t_1+T)) < \min\{\delta_0,\epsilon\}
\end{equation*}
for all $t_1\in[t_0,\infty)$, $\mathbf{q}'\in S(t_1,\delta)$, and $t\in[t_1,\infty)$.

\subsubsection{Uniform Asymptotic Stability}

Since $\mathbf{C}$, $\mathbf{L}$, and $\mathbf{R}^{-1}$ are bounded on $\mathrm{gr}(\Theta)$, there exist $M_C,M_L,M_R\in(0,\infty)$ satisfying
\begin{align*}
\|\mathbf{C}(\mathbf{x})\| &\le M_C, & \|\mathbf{L}(\mathbf{x})\| &\le M_L, & \|\mathbf{R}^{-1}(\mathbf{x})\| &\le M_R
\end{align*}
for all $\mathbf{x}\in\mathrm{gr}(\Theta)$.
Let $c_{ij}$ be the $(i,j)$-th entry of $\mathbf{C}$ for $i,j\in\overline{1,m}$.
Since $\mathbf{C}(\mathbf{x})$ is lower triangular with nonnegative diagonals by \cite[Lemma 1]{An2019}, we have 
\begin{align*}
	\sigma_{\min}(\mathbf{C}(\mathbf{x})) 
		&= \frac{\prod_{i=1}^mc_{ii}(\mathbf{x})}{\prod_{i=1}^{m-1}\sigma_i(\mathbf{C}(\mathbf{x}))} 
		\ge \frac{\prod_{a=1}^l\sigma_{\min}^{m_a}(\mathbf{C}_{aa}(\mathbf{x}))}{\|\mathbf{C}(\mathbf{x})\|^{m-1}} \\
		&\ge \frac{\prod_{a=1}^l\omega_a^{m_a/2}}{\|\mathbf{C}(\mathbf{x})\|^{m-1}\|\mathbf{L}(\mathbf{x})\|^{l/2}} 
		\ge \frac{\prod_{a=1}^l\omega_a^{m_a/2}}{M_C^{m-1}M_L^{l/2}} \\
		&= m_C > 0
\end{align*}
for all $\mathbf{x}\in\mathrm{gr}(\Theta)$ by the Weyl's product inequality \cite[Problems 7.3.P17]{Horn2013} and \eqref{eqn:inequality_of_minimum_singular_value_of_C_aa}.

Let $Y\subset\mathrm{gr}(\Theta)$ be a convex set and $\mathbf{x}_0,\mathbf{x}_1\in Y$.
Let $\mathbf{v} = \mathbf{f}(\mathbf{x}_1) - \mathbf{f}(\mathbf{x}_0) - \mathbf{F}(\mathbf{x}_0)(\mathbf{x}_1-\mathbf{x}_0)$ and $\hat{\mathbf{v}} = \mathbf{v}\|\mathbf{v}\|^+$.
Define $\mathbf{x}:[0,1]\to\mathbb{R}^n$ and $g:[0,1]\to\mathbb{R}$ as $\mathbf{x}(s) = \mathbf{x}_0 + s(\mathbf{x}_1-\mathbf{x}_0)$ and $g(s) = \hat{\mathbf{v}}^T(\mathbf{f}(\mathbf{x}(s)) - \mathbf{f}(\mathbf{x}_0) - \mathbf{F}(\mathbf{x}_0)(\mathbf{x}(s) - \mathbf{x}_0))$.
Since $g$ is continuous on $[0,1]$ and differentiable in $(0,1)$, there exists $s_0\in(0,1)$ satisfying $g(1) - g(0) = dg(s_0)/ds = \hat{\mathbf{v}}^T(\mathbf{F}(\mathbf{x}(s_0)) - \mathbf{F}(\mathbf{x}_0))(\mathbf{x}_1-\mathbf{x}_0)$ by the mean value theorem and the chain rule.
By \ref{ass:dp_F_invR_L_are_Lipschitz}, $\mathbf{F}$ is Lipschitz on $\mathrm{gr}(\Theta)$ with a Lipschitz constant $L_F\in(0,\infty)$.
Then,
\begin{align*}
O(\mathbf{x}_0,\mathbf{x}_1) &= \|\mathbf{f}(\mathbf{x}_1) - \mathbf{f}(\mathbf{x}_0) - \mathbf{F}(\mathbf{x}_0)(\mathbf{x}_1-\mathbf{x}_0)\| \\
&\le \|\mathbf{F}(\mathbf{x}(s_0)) - \mathbf{F}(\mathbf{x}_0)\|\|\mathbf{x}_1-\mathbf{x}_0\| \\
&\le L_F\|\mathbf{x}_1-\mathbf{x}_0\|^2.
\end{align*}

Let $\epsilon\in(0,m_C/(L_FM_R))$.
By the uniform stability, there exists $\delta\in(0,\delta_0]$ satisfying $\|\mathbf{q}'(t) - \mathbf{q}(t)\| < \min\{\delta_0,\epsilon\}$ for all $t_1\in[t_0,\infty)$, $\mathbf{q}'\in S(t_1,\delta)$, and $t\in[t_1,\infty)$.
Then, $\mathbf{x}(t), \mathbf{x}'(t) = (t,\mathbf{q}'(t)) \in \{t\}\times(\mathbf{q}(t) + \delta_0 B_n) \subset \mathrm{gr}(\Theta)$ for all $t\in[t_1,\infty)$.
Assume
\begin{assumption}
	\item\label{ass:m=n} $m = n$.
\end{assumption}
Then, for every $\Delta\mathbf{q}(t) = \mathbf{q}'(t) - \mathbf{q}(t)\neq \mathbf{0}$ we have
\begin{align}
	&\|\mathbf{p}(t) - \mathbf{f}(\mathbf{x}'(t))\| + \|\mathbf{p}(t) - \mathbf{f}(\mathbf{x}(t))\| \nonumber \\
	&\ge \|\mathbf{f}(\mathbf{x}'(t)) - \mathbf{f}(\mathbf{x}(t))\| \nonumber \\
	&\ge \left( \frac{\|\mathbf{F}_q(\mathbf{x}(t))\Delta\mathbf{q}(t)\|}{\|\Delta\mathbf{q}(t)\|} - \frac{O(\mathbf{x}(t),\mathbf{x}'(t))}{\|\Delta\mathbf{q}(t)\|} \right)\|\Delta\mathbf{q}(t)\| \nonumber \\
	&\ge \left( \frac{\sigma_{\min}(\mathbf{C}(\mathbf{x}(t)))}{M_R} - L_F\|\Delta\mathbf{q}(t)\| \right)\|\Delta\mathbf{q}(t)\| \nonumber \\
	&\ge \left( \frac{m_C}{M_R} - L_F\epsilon \right)\|\Delta\mathbf{q}(t)\| \nonumber \\
	&= C_\epsilon\|\Delta\mathbf{q}(t)\|. \label{eqn:convergence_of_delta_q}
\end{align}
Let $\epsilon'\in(0,\infty)$ be arbitrary.
Since $\bar{\phi}_a(t)\to0$ as $t\to\infty$ for all $a\in\overline{1,l}$, there exists $T\in(0,\infty)$ satisfying $\sum_{a=1}^l\bar{\phi}_a(t) < \epsilon' C_\epsilon/2$ for all $t\in[t_0+T,\infty)$.
Then, we have $\|\mathbf{q}'(t) - \mathbf{q}(t)\| \le C_\epsilon^{-1}\sum_{a=1}^l(\phi'_a(t) + \phi_a(t)) \le 2C_\epsilon^{-1}\sum_{a=1}^l\bar{\phi}_a(t-t_1+t_0) < \epsilon'$ for all $t_1\in[t_0,\infty)$, $\mathbf{q}'\in S(t_1,\delta)$, and $t\in[t_1+T,\infty)$.

\subsubsection{Exponential Stability}

Assume
\begin{assumption}
	\item\label{ass:zero_integral_of_psi_from_t_0_to_infinity_and_zero_initial_task_error} $\int_{-\infty}^{\infty}\psi_a(t)dt = \phi_a(t_0) = 0$ for all $a\in\overline{1,l}$.
\end{assumption}
Let $0 < \rho < \min\{k_a\omega_a\mid a\in\overline{1,l}\}$ and define $\boldsymbol{\Phi} = [\Phi_{ij}] \in\mathbb{R}^{l\times l}$ as
\begin{equation*}
	\Phi_{ab} = \begin{dcases*}
		0,& $a < b$ \\
		1,& $a = b$ \\
		\sum_{i=b}^{a-1}\frac{k_iM_{ai}\Phi_{ib}}{k_a\omega_a-\rho},& $a > b$.
	\end{dcases*}
\end{equation*}
Let $t_1\in[t_0,\infty)$, $\mathbf{q}'\in S(t_1,\delta_0)$, and $t\in[t_1,\infty)$ be arbitrary.
Then, $\phi_1'(t) \le \Phi_{11}\phi_1'(t_1)e^{-\rho(t-t_1)}$.
Let $a\in\overline{2,l}$ and assume $\phi_b'(t) \le \sum_{i=1}^b\Phi_{bi}\phi_i'(t_1)e^{-\rho(t-t_1)}$ for all $b\in\overline{1,a-1}$.
Then,
\begin{align*}
&\phi_a'(t)e^{\rho(t-t_1)} \\
&\le \int_{t_1}^t|\gamma_a'(s)|e^{\rho(s-t_1)}e^{-(k_a\omega_a-\rho)(t-s)}ds + \phi_a'(t_1) \\
&\le \sum_{b=1}^{a-1}\sum_{i=1}^b\frac{k_bM_{ab}\Phi_{bi}}{k_a\omega_a-\rho}\phi_i'(t_1) + \phi_a'(t_1) \\
&= \sum_{b=1}^{a-1}\sum_{i=b}^{a-1}\frac{k_iM_{ai}\Phi_{ib}}{k_a\omega_a-\rho}\phi_b'(t_1) + \phi_a'(t_1) \\
&= \sum_{b=1}^a\Phi_{ab}\phi_b'(t_1).
\end{align*}
Thus, $\|\mathbf{p}(t) - \mathbf{f}(\mathbf{x}(t))\| = 0$.
Fix $\epsilon\in(0,m_C/(L_FM_R))$.
There exists $\delta\in(0,\delta_0]$ satisfying $\|\mathbf{q}'(t) - \mathbf{q}(t)\| < \min\{\delta_0,\epsilon\}$ for all $t_1\in[t_0,\infty)$, $\mathbf{q}'\in S(t_1,\delta)$, and $t\in[t_1,\infty)$.
By \eqref{eqn:convergence_of_delta_q},
\begin{align}
	&\|\mathbf{q}'(t) - \mathbf{q}(t)\| \nonumber \\
	&\le C_\epsilon^{-1}\|\mathbf{p}(t) - \mathbf{f}(\mathbf{x}'(t))\| \nonumber \\
	&\le C_\epsilon^{-1}\|\boldsymbol{\Phi}\|\|\mathbf{p}(t_1) - \mathbf{f}(\mathbf{x}'(t_1))\|e^{-\rho(t-t_1)} \nonumber \\
	&= C_\epsilon^{-1}\|\boldsymbol{\Phi}\|\|\mathbf{f}(\mathbf{x}(t_1)) - \mathbf{f}(\mathbf{x}'(t_1))\|e^{-\rho(t-t_1)} \nonumber \\
	&\le C_\epsilon^{-1}L_f\|\boldsymbol{\Phi}\|\|\mathbf{q}'(t_1) - \mathbf{q}(t_1)\|e^{-\rho(t-t_1)}. \label{eqn:exponential_stability}
\end{align}

We summarize analysis results of trajectory existence, task convergence, and stability in continuous time as follows:
\begin{theorem}
	\label{thm:task_convergence_and_stability_in_continuous_time}
	Assume \ref{ass:p_bounded_and_f_linearly_bounded} to \ref{ass:lower_bound_of_k_a}.
	Then, there exists a unique $\mathbf{q}:[t_0,\infty)\to\mathbb{R}^n$ satisfying $\mathbf{q}(t_0) = \mathbf{q}_0$ and $\dot{\mathbf{q}}(t) = \mathbf{u}(t,\mathbf{q}(t))$ for all $t\in(t_0,\infty)$ such that $\|\mathbf{p}_a(t) - \mathbf{f}_a(t,\mathbf{q}(t))\| < \theta_a'$ for all $a\in\overline{1,l}$ and $t\in[t_0,\infty)$.
	Assume additionally \ref{ass:bounded_integral_from_t_0_to_infinity}.
	Then, $\int_{t_0}^\infty\|\mathbf{p}_a(t) - \mathbf{f}_a(t,\mathbf{q}(t))\|dt < \infty$ and $\|\mathbf{p}_a(t) - \mathbf{f}_a(t,\mathbf{q}(t))\| \to 0$ as $t\to\infty$ for all $a\in\overline{1,l}$ and there exists $\mathbf{q}_\infty\in\mathbb{R}^n$ such that $\|\mathbf{q}(t) - \mathbf{q}_\infty\|\to0$ as $t\to\infty$.
	Various stability notions of the equilibrium point $\mathbf{z} = \mathbf{0}$ of \eqref{eqn:transformed_system} at $s = 0$ can be satisfied by the following conditions:
	\begin{enumerate}
		\item \ref{ass:p_bounded_and_f_linearly_bounded} to \ref{ass:bounded_integral_from_t_0_to_infinity} imply \ref{sta:stability};
		\item \ref{ass:p_bounded_and_f_linearly_bounded} to \ref{ass:dp_F_invR_L_are_Lipschitz} imply \ref{sta:uniform_stability};
		\item \ref{ass:p_bounded_and_f_linearly_bounded} to \ref{ass:m=n} imply \ref{sta:uniform_asymptotic_stability};
		\item \ref{ass:p_bounded_and_f_linearly_bounded} to \ref{ass:zero_integral_of_psi_from_t_0_to_infinity_and_zero_initial_task_error} imply \ref{sta:exponential_stability}.
	\end{enumerate}
\end{theorem}

\section{Discrete Time}
\label{sec:discrete_time}

In most practical applications, the PIK problem is solved in discrete time.
In other words, the PIK solution is calculated at a sequence of points in time and hold the value between two consecutive points of time.
As a result, we have an approximated joint trajectory that is a piecewise linear function.
Undoubtedly, we need to check if an approximated joint trajectory converges to a joint trajectory as the calculation points in time become finer.
Also, we should find conditions for task convergence and stability of approximated joint trajectories.

\subsection{Approximated Joint Trajectory}

For $t\in\mathbb{R}$, a set $P_t = \{\tau_i\in\mathbb{R}\mid i\in\mathbb{N}\cup\{0\}\}$ is said to be a partition of $[t,\infty)$ if $t = \tau_0 < \tau_1 < \tau_2 < \cdots$ and $\tau_i\to\infty$ as $i\to\infty$.
The norm of a partition $P_t$ is defined as $\|P_t\| = \sup\{\tau_i - \tau_{i-1}\mid i\in\mathbb{N}\}$.
A function $\tilde{\mathbf{q}}:[t_0,\infty)\to\mathbb{R}^n$ is said to be an approximated joint trajectory of the differntial equation \eqref{eqn:differential_equation} with an initial value $(t_0,\mathbf{q}_0)$ given by a partition $P_{t_0}$ if $\tilde{\mathbf{q}}(\tau_0) = \mathbf{q}_0$ and
\begin{equation*}
	\tilde{\mathbf{q}}(t) = \tilde{\mathbf{q}}(\tau_{i-1}) + (t-\tau_{i-1})\mathbf{u}(\tilde{\mathbf{x}}(\tau_{i-1}))
\end{equation*}
for all $t\in(\tau_{i-1},\tau_i]$ and $i\in\mathbb{N}$ where $\tilde{\mathbf{x}}(t) = (t,\tilde{\mathbf{q}}(t))$.
A sequence $\{\hat{\mathbf{q}}_i = \tilde{\mathbf{q}}(\tau_i)\mid i\in\mathbb{N}\cup\{0\}\}$ is said to be a discrete joint trajectory of \eqref{eqn:differential_equation} with $(t_0,\mathbf{q}_0)$ given by $P_{t_0}$. 

Assume \ref{ass:p_bounded_and_f_linearly_bounded} to \ref{ass:lower_bound_of_k_a} and
\begin{assumption}
	\item\label{ass:theta_a_prime_is_less_than_theta_a} $\theta_a'\in(0,\theta_a)$ for all $a\in\overline{1,l}$.
\end{assumption}
Let $\mathbf{q}:[t_0,\infty)\to\mathbb{R}^n$ be the joint trajectory of \eqref{eqn:differential_equation} with $\mathbf{q}(t_0) = \mathbf{q}_0$.
For each $\eta\in (0,\infty)$, we define $\tilde{S}(\eta)$ as the set of all approximated joint trajectories of \eqref{eqn:differential_equation} with $(t_0,\mathbf{q}_0)$ given by a partition $P_{t_0}$ satisfying $\|P_{t_0}\|<\eta$.
We prove that for every $T\in(0,\infty)$ there exists $\eta_T\in(0,\infty)$ such that $\|\mathbf{p}_a(t) - \mathbf{f}_a(t,\tilde{\mathbf{q}}(t))\|<\theta_a$ for all $\tilde{\mathbf{q}}\in\tilde{S}(\eta_T)$, $a\in\overline{1,l}$, and $t\in[t_0,t_0+T]$.
Let $T\in(0,\infty)$ be arbitrary, $P_{t_0}$ be an arbitrary partition of $[t_0,\infty)$, and $\tilde{\mathbf{q}}:[t_0,\infty)\to\mathbb{R}^n$ be the approximated joint trajectory of \eqref{eqn:differential_equation} with $\tilde{\mathbf{q}}(t_0) = \mathbf{q}_0$ given by $P_{t_0}$.
Since $\dot{\mathbf{p}}'$ is assumed to be bounded, there exists $M_p\in(0,\infty)$ satisfying $\|\dot{\mathbf{p}}'(\mathbf{x})\|\le M_p$ for all $\mathbf{x}\in\mathrm{gr}(\Theta)$.
Let $M(M_p + (\sum_{a=1}^lk_a^2\theta_a^2)^{1/2})\le M_u<\infty$, $r_T = TM_u$, and $C_T = \mathrm{gr}(\Theta)\cap([t_0,t_0+T]\times (\mathbf{q}_0 + r_TB_n))$.
Since $\mathbf{u}$ is locally Lipschitz on $\mathrm{gr}(\Theta)$ and $C_T$ is compact, $\mathbf{u}$ is Lipschitz on $C_T$ with a Lipschitz constant $L_T\in(0,\infty)$.
Since $\mathbf{F}$ is bounded, there exist $M_{F_a}\in(0,\infty)$ for $a\in\overline{1,l}$ satisfying $\|\mathbf{F}_a(\mathbf{x})\|\le M_{F_a}$ for all $a\in\overline{1,l}$ and $\mathbf{x}\in X$.
Let
\begin{equation}
	\label{eqn:upper_bound_of_eta_for_phi_to_be_smaller_than_theta_for_finete_time_period}
	\|P_{t_0}\| < \eta_T \le \min_{a\in\overline{1,l}}\frac{k_a\omega_a(\theta_a-\theta_a')}{M_{F_a}L_T(1+M_u)}.
\end{equation}

Suppose that there exists $t_1\in(t_0,t_0+T]$ such that $\|\mathbf{p}_a(t) - \mathbf{f}_a(t,\tilde{\mathbf{q}}(t))\| < \theta_a$ for all $a\in\overline{1,l}$ and $t\in[t_0,t_1)$ and $\|\mathbf{p}_a(t_1) - \mathbf{f}_a(t_1,\tilde{\mathbf{q}}(t_1))\| = \theta_a$ for at least one $a\in\overline{1,l}$.
Let $N\in\mathbb{N}$ be such that $t_1\in(\tau_{N-1},\tau_N]$.
Then, $\|\mathbf{u}(\tilde{\mathbf{x}}(\tau_{i-1}))\| \le M_u$ and $\|\tilde{\mathbf{q}}(t) - \tilde{\mathbf{q}}(t_0)\| \le r_T$ 
for all $t\in[\tau_{i-1},\tau_i]\cap[t_0,t_1]$ and $i\in\overline{1,N}$.
Define $\tilde{\phi}_a,\tilde{\rho}_a,\tilde{\gamma}_a:[t_0,t_1]\to\mathbb{R}$ for $a\in\overline{1,l}$ as in \eqref{eqn:phi_a} to \eqref{eqn:gamma_a} by replacing $\mathbf{x}(t)$ with $\tilde{\mathbf{x}}(t)$.
Define $\tilde{\delta}_a:[t_0,t_1]\to\mathbb{R}$ for $a\in\overline{1,l}$ as
\begin{equation}
	\label{eqn:delta_a}
	\tilde{\delta}_a(t) = \tilde{\phi}_a^+(t) \left\langle \mathbf{e}_a(\tilde{\mathbf{x}}(t)), \mathbf{F}_{qa}(\tilde{\mathbf{x}}(t))(\mathbf{u}(\tilde{\mathbf{x}}(t)) - \mathbf{u}(\tilde{\mathbf{x}}(\tau_{i-1})))\right\rangle
\end{equation}
if $t\in(\tau_{i-1},\tau_i]$.
Since $P_{t_0}\cap[t_0,t_1]$ is finite, we can check that $\tilde{\phi}_a$ is absolutely continuous on $[t_0,t_1]$; $\dot{\tilde{\phi}}_a(t) = -\tilde{\rho}_a(t)\tilde{\phi}_a(t) + \tilde{\gamma}_a(t) + \tilde{\delta}_a(t)$ for almost all $t\in[t_0,t_1]$; and $\tilde{\rho}_a$, $\tilde{\gamma}_a$, and $\tilde{\delta}_a$ are integrable on $[t_0,t_1]$ for all $a\in\overline{1,l}$ (in the sense of Lebesgue).
Since $\|\mathbf{b}_a(\tilde{\mathbf{x}}(\cdot))\|$ is continuous on the compact set $[t_0,t_1]$, there exists $\tilde{\beta}_a = \max\{\|\mathbf{b}_a(\tilde{\mathbf{x}}(t))\|\mid t\in[t_0,t_1]\}$ for $a\in\overline{1,l}$.
Then, we have
\begin{align*}
	\tilde{\phi}_a(t)
		&\le \left(\theta_a' - \frac{\tilde{\beta}_a}{k_a\omega_a}\right)e^{-k_a\omega_a(t-t_0)} + \frac{\tilde{\beta}_a}{k_a\omega_a} \\
		&\quad + M_{F_a}L_T\sum_{j=1}^{i-1}\int_{\tau_{j-1}}^{\tau_j}\|\tilde{\mathbf{x}}(s) - \tilde{\mathbf{x}}(\tau_{j-1})\|e^{-k_a\omega_a(t-s)}ds \\
		&\quad + M_{F_a}L_T\int_{\tau_{i-1}}^t\|\tilde{\mathbf{x}}(s) - \tilde{\mathbf{x}}(\tau_{i-1})\|e^{-k_a\omega_a(t-s)}ds \\
		&< \theta_a' + \frac{M_{F_a}L_T(1 + M_u)}{k_a\omega_a}\|P_{t_0}\| < \theta_a
\end{align*}
for all $a\in\overline{1,l}$, $t\in[\tau_{i-1},\tau_i]\cap[t_0,t_1]$, and $i\in\overline{1,N}$ by \ref{ass:lower_bound_of_k_a}, a contradiction that $\tilde{\phi}_a(t_1) = \theta_a$ for at least one $a\in\overline{1,l}$.
Therefore, $\|\mathbf{p}_a(t) - \mathbf{f}_a(t,\tilde{\mathbf{q}}(t))\| < \theta_a$ for all $a\in\overline{1,l}$ and $t\in[t_0,t_0+T]$.

We prove that for every $T,\epsilon\in(0,\infty)$ there exists $\eta_{T,\epsilon}\in(0,\infty)$ such that $\|\mathbf{q}(t) - \tilde{\mathbf{q}}(t)\|<\epsilon$ for all $\tilde{\mathbf{q}}\in\tilde{S}(\eta_{T,\epsilon})$ and $t\in[t_0,t_0+T]$.
Let $T,\epsilon\in(0,\infty)$ be arbitrary and $\|P_{t_0}\|<\eta_T$.
Then,
\begin{align*}
	&\|\mathbf{q}(t) - \tilde{\mathbf{q}}(t)\| \\
	&\le \sum_{j=1}^{i-1}\int_{\tau_{j-1}}^{\tau_j}\|\mathbf{u}(\mathbf{x}(s)) - \mathbf{u}(\tilde{\mathbf{x}}(\tau_{j-1}))\|ds \\
	&\quad + \int_{\tau_{i-1}}^t\|\mathbf{u}(\mathbf{x}(s)) - \mathbf{u}(\tilde{\mathbf{x}}(\tau_{i-1}))\|ds \\
	&\le (1+M_u)TL_T\|P_{t_0}\| + L_T\int_{t_0}^t\|\mathbf{q}(s) - \tilde{\mathbf{q}}(s)\|ds
\end{align*}
for all $\tilde{\mathbf{q}}\in\tilde{S}(\eta_T)$ and $t\in[t_0,t_0+T]$.
Let
\begin{equation*}
	\|P_{t_0}\| < \eta_{T,\epsilon} \le \min\left\{\eta_T, \frac{\epsilon}{(1+M_u)T(e^{TL_T}-1)}\right\}.
\end{equation*}
By the Gronwall's inequality,
\begin{equation*}
	\|\mathbf{q}(t) - \tilde{\mathbf{q}}(t)\| \le (1+M_u)T(e^{TL_T}-1)\|P_{t_0}\| < \epsilon
\end{equation*}
for all $\tilde{\mathbf{q}}\in\tilde{S}(\eta_{T,\epsilon})$ and $t\in[t_0,t_0+T]$.
Therefore, $\|\mathbf{q}(t) - \tilde{\mathbf{q}}(t)\|\to0$ as $\|P_{t_0}\|\to0$ uniformly on $[t_0,t_0+T]$.

\subsection{Task Convergence}

Assume \ref{ass:dp_F_invR_L_are_Lipschitz} and let 
\begin{equation}
	\label{eqn:upper_bound_of_eta_for_phi_to_be_smaller_than_theta_for_all_time}
	\|P_{t_0}\| < \eta_\infty \le \min_{a\in\overline{1,l}}\frac{k_a\omega_a(\theta_a-\theta_a')}{M_{F_a}L_u(1+M_u)}.
\end{equation}
It is not difficult to see that $\|\mathbf{p}_a(t) - \mathbf{f}_a(t,\tilde{\mathbf{q}}(t))\|<\theta_a$ for all $\tilde{\mathbf{q}}\in\tilde{S}(\eta_\infty)$, $a\in\overline{1,l}$, and $t\in[t_0,\infty)$.
Assume \ref{ass:bounded_integral_from_t_0_to_infinity} and
\begin{assumption}
	\item\label{ass:psi_a_monotonic_decrease} $\psi_a$ monotonically decreases on $[t_0-\eta_\infty,\infty)$ for all $a\in\overline{1,l}$.
\end{assumption}
We prove that there exists $\eta_0\in(0,\eta_\infty)$ such that for every $\tilde{\mathbf{q}}\in\tilde{S}(\eta_0)$, we have $\int_{t_0}^\infty\|\mathbf{p}_a(t) - \mathbf{f}_a(t,\tilde{\mathbf{q}}(t))\|dt<\infty$ and $\|\mathbf{p}_a(t) - \mathbf{f}_a(t,\tilde{\mathbf{q}}(t))\|\to0$ as $t\to\infty$ for all $a\in\overline{1,l}$ and there exists $\tilde{\mathbf{q}}_\infty\in\mathbb{R}^n$ such that $\|\tilde{\mathbf{q}}(t) - \tilde{\mathbf{q}}_\infty\|\to0$ as $t\to\infty$.

Let $\psi = \|(\psi_1,\dots,\psi_l)\|$, $\tilde{\phi} = \|(\tilde{\phi}_1,\dots,\tilde{\phi}_l)\|$, $K = \|\mathbf{K}\|$, and $M_F,M_A\in(0,\infty)$ be such that $\|\mathbf{F}_q(\mathbf{x})\|\le M_F$ and $\|\mathbf{A}(\mathbf{x})\|\le M_A$ for all $\mathbf{x}\in\mathrm{gr}(\Theta)$.
Since the product of bounded Lipschitz functions is Lipschitz, $\mathbf{M} = \mathbf{R}^{-1}\hat{\mathbf{J}}^T\mathbf{C}_D^T\mathbf{L}$ is Lipschitz on $\mathrm{gr}(\Theta)$ with a Lipschitz constant $L_M\in(0,\infty)$.
Let $\rho_{\max} = \max\{k_1\omega_1,\dots,k_l\omega_l\}$ and $\rho\in[0,\rho_{\max}]$ be arbitrary.
Then,
\begin{align*}
	&\|\mathbf{e}(\tilde{\mathbf{x}}(t)) - \mathbf{e}(\tilde{\mathbf{x}}(\tau_{i-1}))\| \\
	&\le \int_{\tau_{i-1}}^t\|(\dot{\mathbf{p}}' - \mathbf{A}(\dot{\mathbf{p}}' + \mathbf{K}\mathbf{e}))(\tilde{\mathbf{x}}(s))\|ds \\
	&\quad + \int_{\tau_{i-1}}^t\|\mathbf{F}_q(\tilde{\mathbf{x}}(s))(\mathbf{u}(\tilde{\mathbf{x}}(s)) - \mathbf{u}(\tilde{\mathbf{x}}(\tau_{i-1})))\|ds \\
	&\le (1+M_A)\int_{\tau_{i-1}}^t\psi(s)ds + M_AK\int_{\tau_{i-1}}^t\tilde{\phi}(s)ds \\
	&\quad + M_F\int_{\tau_{i-1}}^t\|\mathbf{u}(\tilde{\mathbf{x}}(s)) - \mathbf{u}(\tilde{\mathbf{x}}(\tau_{i-1}))\|ds
\end{align*}
and
\begin{align*}
	&\int_{\tau_{i-1}}^{t} \|\mathbf{u}(\tilde{\mathbf{x}}(s)) - \mathbf{u}(\tilde{\mathbf{x}}(\tau_{i-1}))\|e^{\rho s}ds \\
	&\le M\int_{\tau_{i-1}}^{t}(\|\dot{\mathbf{p}}'(\tilde{\mathbf{x}}(s))\| + \|\dot{\mathbf{p}}'(\tilde{\mathbf{x}}(\tau_{i-1}))\|)e^{\rho s}ds \\
	&\quad + \int_{\tau_{i-1}}^{t}\|\mathbf{M}(\tilde{\mathbf{x}}(s)) - \mathbf{M}(\tilde{\mathbf{x}}(\tau_{i-1}))\|\|\mathbf{K}\|\|\mathbf{e}(\tilde{\mathbf{x}}(s))\|e^{\rho s}ds \\
	&\quad + \int_{\tau_{i-1}}^{t} \|\mathbf{M}(\tilde{\mathbf{x}}(\tau_{i-1}))\|\|\mathbf{K}\|\|\mathbf{e}(\tilde{\mathbf{x}}(s)) - \mathbf{e}(\tilde{\mathbf{x}}(\tau_{i-1}))\|e^{\rho s}ds \\
	&\le M(1 + K(1+M_A)\|P_{t_0}\|e^{\rho_{\max} \|P_{t_0}\|})\int_{\tau_{i-1}}^t\psi(s)e^{\rho s}ds \\
	&\quad + M\int_{\tau_{i-1}}^t\psi(s - \|P_{t_0}\|)e^{\rho s}ds \\
	&\quad + N_1(\|P_{t_0}\|)\|P_{t_0}\|\int_{\tau_{i-1}}^t\|\mathbf{e}(\tilde{\mathbf{x}}(s))\|e^{\rho s}ds \\
	&\quad + N_2(\|P_{t_0}\|)\int_{\tau_{i-1}}^t\|\mathbf{u}(\tilde{\mathbf{x}}(s)) - \mathbf{u}(\tilde{\mathbf{x}}(\tau_{i-1}))\|e^{\rho s}ds
\end{align*}
for all $t\in[\tau_{i-1},\tau_i)$ and $i\in\mathbb{N}$ where $N_1(\|P_{t_0}\|) = (L_M(1+M_u) + MM_AKe^{\rho_{\max} \|P_{t_0}\|})K$ and $N_2(\|P_{t_0}\|) = MM_FK\|P_{t_0}\|e^{\rho_{\max} \|P_{t_0}\|}$.
Let
\begin{equation*}
	\|P_{t_0}\|< \eta_0' < \max\{\eta\in(0,\eta_\infty]\mid N_2(\eta)\le1\}.
\end{equation*}
Then,
\begin{align*}
	&\int_{\tau_{i-1}}^{t} \|\mathbf{u}(\tilde{\mathbf{x}}(s)) - \mathbf{u}(\tilde{\mathbf{x}}(\tau_{i-1}))\|e^{\rho s}ds \\
	&\le \frac{(1 + K(1+M_A)\eta_0'e^{\rho_{\max}\eta_0'})M}{1 - N_2(\eta_0')}\int_{\tau_{i-1}}^t\psi(s)e^{\rho s}ds \\
	&\quad + \frac{M}{1 - N_2(\eta_0')}\int_{\tau_{i-1}}^t\psi(s - \eta_0')e^{\rho s}ds \\
	&\quad + \frac{N_1(\eta_0')}{1 - N_2(\eta_0')}\|P_{t_0}\|\int_{\tau_{i-1}}^t\tilde{\phi}(s)e^{\rho s}ds \\
	&= \int_{\tau_{i-1}}^t\left(C_1\psi(s) + C_2\psi(s - \eta_0') + C_3\|P_{t_0}\|\tilde{\phi}(s)\right)e^{\rho s}ds
\end{align*}
for all $t\in[\tau_{i-1},\tau_i)$ and $i\in\mathbb{N}$.

Let
\begin{equation*}
	\|P_{t_0}\| \le \eta_0 < \min\{\eta_0',1/\mathrm{sr}(\mathbf{Y}^{-1}\mathbf{Z})\}
\end{equation*}
where
\begin{align*}
	\mathbf{Y} &= \begin{bmatrix} 1 & 0 & \cdots & 0 \\ -\frac{M_{21}k_1}{k_2\omega_2} & 1 & \cdots & 0 \\ \vdots & \vdots & \ddots & \vdots \\ -\frac{M_{l1}k_1}{k_l\omega_l} & -\frac{M_{l2}k_2}{k_l\omega_l} & \cdots & 1 \end{bmatrix} \in \mathbb{R}^{l\times l}, \\
	\mathbf{Z} &= C_3\begin{bmatrix} \frac{M_{F_1}}{k_1\omega_1} & \frac{M_{F_1}}{k_1\omega_1} & \cdots & \frac{M_{F_1}}{k_1\omega_1} \\ \frac{M_{F_2}}{k_2\omega_2} & \frac{M_{F_2}}{k_2\omega_2} & \cdots & \frac{M_{F_2}}{k_2\omega_2} \\ \vdots & \vdots & \ddots & \vdots \\ \frac{M_{F_l}}{k_l\omega_l} & \frac{M_{F_l}}{k_l\omega_l} & \cdots & \frac{M_{F_l}}{k_l\omega_l} \end{bmatrix} \in \mathbb{R}^{l\times l},
\end{align*}
and $\mathrm{sr}(\cdot)$ is the spectral radius of a square matrix.
Define $\hat{\phi}_a,\hat{\gamma}_a,\hat{\delta}_a:[t_0,\infty)\to[0,\infty)$ for $a\in\overline{1,l}$ as:
\begin{align*}
	\hat{\gamma}_a(t) &= \sum_{b=1}^{a-1}M_{ab}(\psi_b(t)+k_b\hat{\phi}_a(t)) + (1+M_{aa})\psi_a(t) \\
	\hat{\delta}_a(t) &= M_{F_a}\left(C_1\psi(t) + C_2\psi(t - \eta_0') + C_3\eta_0\sum_{b=1}^l\hat{\phi}_b(t)\right) \\
	\hat{\phi}_a(t) &= \theta_a'e^{-k_a\omega_a(t-t_0)} + \int_{t_0}^t(\hat{\gamma}_a(s) + \hat{\delta}_a(s))e^{-k_a\omega_a(t-s)}ds.
\end{align*}
Then, $\dot{\hat{\phi}}_a(t) = -k_a\omega_a\hat{\phi}_a(t) + \hat{\gamma}_a(t) + \hat{\delta}_a(t)$,
\begin{align*}
	\tilde{\phi}_a(t) - \hat{\phi}_a(t) 
	&\le (\tilde{\phi}_a(t_0) - \hat{\phi}_a(t_0))e^{-k_a\omega_a(t-t_0)} \\
	&\quad + \int_{t_0}^t(|\tilde{\gamma}_a(s)| - \hat{\gamma}_a(s))e^{-k_a\omega_a(t-s)}ds \\
	&\quad + \int_{t_0}^t(|\tilde{\delta}_a(s)| - \hat{\delta}_a(s))e^{-k_a\omega_a(t-s)}ds \\
	&\le \sum_{b=1}^{a-1}M_{ab}k_b\int_{t_0}^t(\tilde{\phi}_b(s) - \hat{\phi}_b(s))ds \\
	&\quad + M_{F_a}C_3\eta_0\sum_{b=1}^l\int_{t_0}^t(\tilde{\phi}_b(s) - \hat{\phi}_b(s))ds,
\end{align*}
and
\begin{align*}
	\int_{t_0}^t\hat{\phi}_a(s)ds 
	&= \frac{1}{k_a\omega_a}\int_{t_0}^t\left(\hat{\gamma}_a(s) + \hat{\delta}_a(s) - \dot{\hat{\phi}}_a(s)\right)ds \\
	&\le \int_{t_0}^t(\hat{\psi}_a(s) + \hat{\alpha}_a(t-s))ds \\
	&\quad + \sum_{b=1}^{a-1}\frac{M_{ab}k_b}{k_a\omega_a}\int_{t_0}^t\hat{\phi}_b(s)ds \\
	&\quad + \frac{M_{F_a}C_{3}}{k_a\omega_a}\eta_0\sum_{b=1}^l\int_{t_0}^t\hat{\phi}_b(s)ds
\end{align*}
for all $a\in\overline{1,l}$ and $t\in[t_0,\infty)$ where $\hat{\alpha}_a(t) = \theta_a'e^{-k_a\omega_at}$ and $\hat{\psi}_a(t) = (\sum_{b=1}^{a-1}M_{ab}\psi_b(t) + (1+M_{aa})\psi_a(t) + M_{F_a}C_{1}\psi(t) + M_{F_a}C_{2}\psi(t - \eta_0'))/(k_a\omega_a)$.
By the Gronwall's inequality, $\sum_{a=1}^l(\tilde{\phi}_a(t) - \hat{\phi}_a(t)) \le 0$ and $\tilde{\phi}_a(t) \le \hat{\phi}_a(t)$ for all $a\in\overline{1,l}$ and $t\in[t_0,\infty)$.

Let $\hat{\boldsymbol{\phi}} = (\hat{\phi}_1,\dots,\hat{\phi}_l)$, $\hat{\boldsymbol{\psi}} = (\hat{\psi}_1,\dots,\hat{\psi}_l)$, $\hat{\boldsymbol{\alpha}} = (\hat{\alpha}_1,\dots,\hat{\alpha}_l)$, and $\mathbf{X} = \mathbf{Y} - \eta_0\mathbf{Z}$.
Then, we have
\begin{equation*}
	\mathbf{X}\int_{t_0}^t\hat{\boldsymbol{\phi}}(s)ds \le \int_{t_0}^t(\hat{\boldsymbol{\psi}}(s) + \hat{\boldsymbol{\alpha}}(t-s))ds
\end{equation*}
for all $t\in[t_0,\infty)$ where the symbol $\le$ denotes the element-wise order of two vectors.
Since all off-diagonal entries of $\mathbf{X}$ is negative, $\mathbf{X}$ is a Z-matrix.
Since $\mathbf{Y}^{-1}$ and $\eta_0\mathbf{Z}$ are nonnegative, $\mathbf{X}$ is an M-matrix if and only if $\mathrm{sr}(\eta_0\mathbf{Y}^{-1}\mathbf{Z})<1$ \cite{Plemmons1977}.
Since $\mathbf{X}$ is an M-matrix, $\mathbf{X}^{-1}$ is nonnegative.
Therefore,
\begin{equation*}
	0 \le \hat{\boldsymbol{\Phi}} = \mathbf{X}^{-1}\lim_{t\to\infty}\int_{t_0}^t(\hat{\boldsymbol{\psi}}(s) + \hat{\boldsymbol{\alpha}}(t-s))ds < \infty
\end{equation*}
and
\begin{equation*}
	\int_{t_0}^t\hat{\boldsymbol{\phi}}(s)ds \le \mathbf{X}^{-1}\int_{t_0}^t(\hat{\boldsymbol{\psi}}(s) + \hat{\boldsymbol{\alpha}}(t-s))ds \le \hat{\boldsymbol{\Phi}}
\end{equation*}
for all $t\in[t_0,\infty)$ where $0\le\hat{\boldsymbol{\Phi}} = (\hat{\Phi}_1,\dots,\hat{\Phi}_l) < \infty$ denotes $0\le\hat{\Phi}_a<\infty$ for all $a\in\overline{1,l}$.
It follows that $\int_{t_0}^\infty(\hat{\gamma}_a(t) + \hat{\delta}_a(t))dt<\infty$ and
\begin{equation*}
	\lim_{t\to\infty}\int_{t_0}^t(\hat{\gamma}_a(s) + \hat{\delta}_a(s))e^{-k_a\omega_a(t-s)}ds = 0
\end{equation*}
for all $a\in\overline{1,l}$ by the Lebesgue's dominated convergence theorem.
Therefore, $\tilde{\phi}_a(t) \le \hat{\phi}_a(t) \to 0$ as $t\to\infty$ for all $a\in\overline{1,l}$.

For every $t,t'\in[\tau_i,\infty)$ and $i\in\mathbb{N}$, we have
\begin{align*}
	&\|\tilde{\mathbf{q}}(t') - \tilde{\mathbf{q}}(t)\| \\
	&\le \|\tilde{\mathbf{q}}(t') - \tilde{\mathbf{q}}(\tau_i)\| + \|\tilde{\mathbf{q}}(t) - \tilde{\mathbf{q}}(\tau_i)\| \\
	&\le 2M\int_{\tau_i}^\infty\left(\sum_{a=1}^l(\psi_a(s) + k_a\tilde{\phi}_a(s))^2\right)^{1/2}ds \\
	&\quad + 2\int_{\tau_i}^\infty\left(C_1\psi(s) + C_2\psi(s - \eta_0') + C_3\|P_{t_0}\|\tilde{\phi}(s)\right)ds.
\end{align*}
Since $\int_{t_0}^\infty\psi_a(t)dt<\infty$ and $\int_{t_0}^\infty\tilde{\phi}(t)dt \le \int_{t_0}^\infty\hat{\phi}_a(t)dt<\infty$, for every $\epsilon\in(0,\infty)$ there exists $i\in\mathbb{N}$ such that $\|\tilde{\mathbf{q}}(t') - \tilde{\mathbf{q}}(t)\| < \epsilon$ for all $t,t'\in[\tau_i,\infty)$.
It follows that for every divergent sequence $t_0 \le t_1 < t_2 < \cdots$, the sequence $\{\tilde{\mathbf{q}}(t_i)\}$ converges in $\mathbb{R}^n$ because it is Cauchy.
Therefore, $\tilde{\mathbf{q}}(t)$ converges to a point $\tilde{\mathbf{q}}_\infty\in\mathbb{R}^n$ as $t\to\infty$.

\subsection{Stability}

For each $(t_1,\delta,\eta)\in[t_0,\infty)\times(0,\delta_0]\times(0,\eta_0]$, we define $\tilde{S}(t_1,\delta,\eta)$ as the set of all approximated joint trajectories $\tilde{\mathbf{q}}':[t_1,\infty)\to\mathbb{R}^n$ of \eqref{eqn:differential_equation} with an initial value $(t_1,\tilde{\mathbf{q}}_1')\in X$ given by a partition $P_{t_1}$ satisfying $\tilde{\mathbf{q}}_1'\in\mathbf{q}(t_1) + \delta B_n$ and $\|P_{t_1}\|<\eta$.
We study stability of approximated joint trajectories $\tilde{\mathbf{q}}'\in \tilde{S}(t_1,\delta_0,\eta_0)$ on the joint trajectory $\mathbf{q}(t)$.
Specifically, we will find conditions for the following stability notions:
\begin{stability}
	\item\label{sta2:uniform_stability} for every $\epsilon\in(0,\infty)$ there exist $\delta,\eta\in(0,\infty)$ such that $\|\tilde{\mathbf{q}}'(t) - \mathbf{q}(t)\| < \epsilon$ for all $t_1\in[t_0,\infty)$, $\tilde{\mathbf{q}}'\in \tilde{S}(t_1,\delta,\eta)$, and $t\in[t_1,\infty)$ (uniform stability);
	\item\label{sta2:uniform_asymptotic_stability} \ref{sta2:uniform_stability} holds and there exist $\delta',\eta'\in(0,\infty)$ such that $\|\tilde{\mathbf{q}}'(t) - \mathbf{q}(t)\|\to0$ as $t\to\infty$ for all $t_1\in[t_0,\infty)$ and $\tilde{\mathbf{q}}'\in \tilde{S}(t_1,\delta',\eta')$ uniformly in $t_1$, i.e., for each $\epsilon'>0$ there exists $T\in(0,\infty)$ such that $\|\tilde{\mathbf{q}}'(t)-\mathbf{q}(t)\|<\epsilon'$ for all $t_1\in[t_0,\infty)$, $\tilde{\mathbf{q}}'\in \tilde{S}(t_1,\delta',\eta')$, and $t\in[t_1 + T,\infty)$ (uniform asymptotic stability);
	\item\label{sta2:exponential_stability} there exist $\delta,\eta,M,\rho\in(0,\infty)$ such that $\|\tilde{\mathbf{q}}'(t) - \mathbf{q}(t)\|\le M\|\tilde{\mathbf{q}}'(t_1) - \mathbf{q}(t_1)\|e^{-\rho(t-t_1)}$ for all $t_1\in[t_0,\infty)$, $\tilde{\mathbf{q}}'\in \tilde{S}(t_1,\delta,\eta)$, and $t\in[t_1,\infty)$ (exponential stability).
\end{stability}

\subsubsection{Uniform Stability}

For each $t_1\in[t_0,\infty)$ and $\tilde{\mathbf{q}}'\in \tilde{S}(t_1,\delta_0,\eta_0)$, we define $\tilde{\phi}_a',\tilde{\rho}_a',\tilde{\gamma}_a',\tilde{\delta}_a':[t_1,\infty)\to\mathbb{R}$ for $a\in\overline{1,l}$ as in \eqref{eqn:phi_a} to \eqref{eqn:gamma_a} and \eqref{eqn:delta_a} by replacing $\mathbf{x}(t)$ and $\tilde{\mathbf{x}}(t)$ with $\tilde{\mathbf{x}}'(t) = (t,\tilde{\mathbf{q}}'(t))$.
We also define $\tilde{\zeta}':[t_1,\infty)\to\mathbb{R}$ as in \eqref{eqn:zeta} by replacing $\phi_a'$ with $\tilde{\phi}_a'$.
Obviously, $\tilde{\zeta}'$ monotonically decreases and converges to $0$.
We can prove, similarly as before, that for every $\epsilon\in(0,\infty)$ there exists $T\in(0,\infty)$ such that $\tilde{\zeta}'(t)<\epsilon$ for all $t_1\in[t_0,\infty)$, $\tilde{\mathbf{q}}'\in\tilde{S}(t_1,\delta_0,\eta_0)$, and $t\in[t_1+T,\infty)$.
For each $t_1\in[t_0,\infty)$ and $\tilde{\mathbf{q}}'\in\tilde{S}(t_1,\delta_0,\eta_0)$, we write $P'_{t_1} = \{t_1 = \tau'_0 < \tau'_1 < \tau'_2 < \cdots \}$ to denote the partition of $[t_1,\infty)$ that generates $\tilde{\mathbf{q}}'$.

Let $\epsilon\in(0,\infty)$ be arbitrary.
There exists $\delta\in(0,\delta_0]$ satisfying $\|\mathbf{q}'(t) - \mathbf{q}(t)\|<\min\{\delta_0,\epsilon/2\}$ for all $t_1\in[t_0,\infty)$, $\mathbf{q}'\in S(t_1,\delta)$, and $t\in[t_1,\infty)$ by Theorem \ref{thm:task_convergence_and_stability_in_continuous_time}.
Let $t_1\in[t_0,\infty)$ and $\tilde{\mathbf{q}}'\in\tilde{S}(t_1,\delta,\eta_0)$ be arbitrary.
There exists $\mathbf{q}'\in S(t_1,\delta)$ satisfying $\tilde{\mathbf{q}}'(t_1) = \mathbf{q}'(t_1)$.
We proved that there exist $\mathbf{q}'_\infty,\tilde{\mathbf{q}}'_\infty\in\mathbb{R}^n$ such that $\|\mathbf{q}'(t) - \mathbf{q}'_\infty\|\to0$ and $\|\tilde{\mathbf{q}}'(t) - \tilde{\mathbf{q}}'_\infty\|\to0$ as $t\to\infty$.
So, we can write $\mathbf{q}'_\infty = \mathbf{q}'(t_1) + \int_{t_1}^\infty\dot{\mathbf{q}}'(t)dt$ and $\tilde{\mathbf{q}}'_\infty = \tilde{\mathbf{q}}'(t_1) + \int_{t_1}^\infty\dot{\tilde{\mathbf{q}}}(t)dt$.
Then,
\begin{equation*}
	\|\mathbf{q}'(t) - \mathbf{q}'_\infty\| 
		\le \int_{\tau'_i}^\infty\|\mathbf{u}(s,\mathbf{q}'(s))\|ds 
		\le M\zeta'(\tau'_i)
\end{equation*}
and
\begin{align*}
	&\|\tilde{\mathbf{q}}'(t) - \tilde{\mathbf{q}}'_\infty\| \\
	&\le \int_{\tau'_i}^\infty\|\dot{\tilde{\mathbf{q}}}'(s) - \mathbf{u}(s,\tilde{\mathbf{q}}'(s))\|ds + \int_{\tau'_i}^\infty\|\mathbf{u}(s,\tilde{\mathbf{q}}'(s))\|ds \\
	&\le \int_{\tau'_i}^\infty\left(C_1\psi(s) + C_2\psi(s-\eta'_0) + C_3\eta_0\tilde{\phi}'(s)\right)ds + M\tilde{\zeta}'(\tau'_i)
\end{align*}
for all $t\in[\tau'_i,\tau'_{i+1}]$ and $i\in\mathbb{N}$ where $\tilde{\phi}' = \|(\tilde{\phi}'_1,\dots,\tilde{\phi}'_l)\|$.
Since $\int_t^\infty\tilde{\phi}'(s)ds \le \sum_{a=1}^l\int_t^\infty\hat{\phi}_a(s-t_1+t_0)ds\to0$ as $t\to\infty$, there exists $T\in(0,\infty)$ that only depends on $\epsilon$ such that $\|\mathbf{q}'(t) - \mathbf{q}'_\infty\|<\epsilon/10$ and $\|\tilde{\mathbf{q}}'(t) - \tilde{\mathbf{q}}'_\infty\| < \epsilon/10$ for all $t\in[t_1+T,\infty)$.
Let
\begin{equation*}
	\|P'_{t_1}\| < \eta = \min\left\{\eta_0, \frac{\epsilon/10}{(1+M_u)T(e^{TL_u}-1)}\right\}.
\end{equation*}
Then, $\|\mathbf{q}'(t) - \tilde{\mathbf{q}}'(t)\| < \epsilon/10$ for all $t\in[t_1,t_1+T]$ and
\begin{align*}
	&\|\mathbf{q}'(t) - \tilde{\mathbf{q}}'(t)\| \le \|\mathbf{q}'(t) - \mathbf{q}'_\infty\| + \|\mathbf{q}'_\infty - \mathbf{q}'(t_1+T)\| \\
	&\quad\quad\quad + \|\mathbf{q}'(t_1+T) - \tilde{\mathbf{q}}'(t_1+T)\| + \|\tilde{\mathbf{q}}'(t_1+T) - \tilde{\mathbf{q}}'_\infty\| \\
	&\quad\quad\quad + \|\tilde{\mathbf{q}}'_\infty - \tilde{\mathbf{q}}'(t)\| < \epsilon/2
\end{align*}
for all $t\in(t_1+T,\infty)$.
Therefore,
\begin{equation*}
	\|\tilde{\mathbf{q}}'(t) - \mathbf{q}(t)\| \le \|\tilde{\mathbf{q}}'(t) - \mathbf{q}'(t)\| + \|\mathbf{q}'(t) - \mathbf{q}(t)\| < \epsilon
\end{equation*}
for all $t\in[t_1,\infty)$.

\subsubsection{Uniform Asymptotic Stability}

Assume \ref{ass:m=n} and let $\epsilon\in(0,m_C/(L_FM_R))$.
By the uniform stability, there exist $\delta\in(0,\delta_0]$ and $\eta\in(0,\eta_0]$ satisfying $\|\tilde{\mathbf{q}}'(t) - \mathbf{q}(t)\|<\min\{\delta_0,\epsilon\}$ for all $t_1\in[t_0,\infty)$, $\tilde{\mathbf{q}}'\in\tilde{S}(t_1,\delta,\eta)$, and $t\in[t_1,\infty)$.
Then, $\mathbf{x}(t),\tilde{\mathbf{x}}'(t) = (t,\tilde{\mathbf{q}}'(t))\in\{t\}\times(\mathbf{q}(t) + \delta_0B_n)\subset\mathrm{gr}(\Theta)$ for all $t\in[t_1,\infty)$.
Let $\epsilon'\in(0,\infty)$ be arbitrary.
Since $\hat{\phi}_a(t) + \bar{\phi}_a(t)\to0$ as $t\to\infty$ for all $a\in\overline{1,l}$, there exists $T\in(0,\infty)$ satisfying $\sum_{a=1}^l(\hat{\phi}_a(t) + \bar{\phi}_a(t)) < \epsilon'C_\epsilon/2$ for all $t\in[t_0+T,\infty)$.
By replacing $\mathbf{x}'(t)$ with $\tilde{\mathbf{x}}'(t)$ in \eqref{eqn:convergence_of_delta_q}, we have $\|\tilde{\mathbf{q}}'(t) - \mathbf{q}(t)\| \le C_\epsilon^{-1}\sum_{a=1}^l(\tilde{\phi}'_a(t) + \phi_a(t)) \le C_\epsilon^{-1}\sum_{a=1}^l(\hat{\phi}_a(t-t_1+t_0) + \bar{\phi}_a(t-t_1+t_0)) <\epsilon'$ for all $t\in[t_1 + T,\infty)$.

\subsubsection{Exponential Stability}

Assume \ref{ass:zero_integral_of_psi_from_t_0_to_infinity_and_zero_initial_task_error}.
Let $\eta\in(0,\eta_0]$ and $\rho\in(0,\min\{k_1\omega_1,\dots,k_l\omega_l\})$ be such that 
\begin{equation*}
\mathrm{sr}(\mathbf{Y}^{-1}(\eta\mathbf{Z} + \rho\tilde{\mathbf{Z}})) < 1
\end{equation*}
where $\tilde{\mathbf{Z}} = \mathrm{diag}(1/(k_1\omega_1),\dots,1/(k_l\omega_l))\in\mathbb{R}^{l\times l}$.
Let $\tilde{\mathbf{X}} = \mathbf{Y} - \eta\mathbf{Z} - \rho\tilde{\mathbf{Z}}$, $\tilde{\boldsymbol{\Phi}} = \mathbf{I}_l + \tilde{\mathbf{Z}}^{-1}(\mathbf{I}_l - \mathbf{Y} + \eta\mathbf{Z})\tilde{\mathbf{X}}^{-1}\tilde{\mathbf{Z}}$, and $\tilde{\boldsymbol{\phi}}' = (\tilde{\phi}_1',\dots,\tilde{\phi}_l')$.
Then, we can check, similarly as before, that $\tilde{\mathbf{X}}$ is an M-matrix,
\begin{equation*}
\int_{t_1}^t\tilde{\boldsymbol{\phi}}'(s)e^{\rho(s-t_1)}ds \le \tilde{\mathbf{X}}^{-1}\tilde{\mathbf{Z}}\tilde{\boldsymbol{\phi}}'(t_1),
\end{equation*}
and
\begin{equation*}
\tilde{\boldsymbol{\phi}}'(t) \le \tilde{\boldsymbol{\Phi}}\tilde{\boldsymbol{\phi}}'(t_1)e^{-\rho(t-t_1)}
\end{equation*}
for all $t_1\in[t_0,\infty)$, $\mathbf{q}'\in\tilde{S}(t_1,\delta_0,\eta)$, and $t\in[t_1,\infty)$ from
\begin{align*}
&\int_{t_1}^t\tilde{\phi}_a'(s)e^{\rho(s-t_1)}ds \\
&\le \frac{1}{k_a\omega_a}\int_{t_1}^t\left(|\tilde{\gamma}_a'(s)| + |\tilde{\delta}_a'(s)|\right)e^{\rho(s-t_1)}ds \\
&\quad + \frac{\rho}{k_a\omega_a}\int_{t_1}^t\tilde{\phi}_a'(s)e^{\rho(s-t_1)}ds + \frac{\tilde{\phi}_a'(t_1)}{k_a\omega_a} \\
&\le \sum_{b=1}^{a-1}\frac{M_{ab}k_b}{k_a\omega_a}\int_{t_1}^t\tilde{\phi}_b'(s)e^{\rho(s-t_1)}ds \\
&\quad + \frac{\eta M_{F_a}C_3}{k_a\omega_a}\sum_{b=1}^l\int_{t_1}^t\tilde{\phi}_b'(s)e^{\rho(s-t_1)}ds \\
&\quad + \frac{\rho}{k_a\omega_a}\int_{t_1}^t\tilde{\phi}_a'(s)e^{\rho(s-t_1)}ds + \frac{\tilde{\phi}_a'(t_1)}{k_a\omega_a}
\end{align*}
and
\begin{align*}
\tilde{\phi}_a'(t)e^{\rho(t-t_1)} 
&\le \tilde{\phi}_a'(t_1) + \int_{t_1}^t(|\tilde{\gamma}_a'(s)| + |\tilde{\delta}_a'(s)|)e^{\rho(s-t_1)}ds \\
&\le \tilde{\phi}_a'(t_1) + \sum_{b=1}^{a-1}M_{ab}k_b\int_{t_1}^t\tilde{\phi}_b'(s)e^{\rho(s-t_1)}ds \\
&\quad + \eta M_{F_a}C_3\sum_{b=1}^l\int_{t_1}^t\tilde{\phi}_b'(s)e^{\rho(s-t_1)}ds.
\end{align*}
Fix $\epsilon\in(0,m_C/(L_FM_R))$.
By the uniform stability, there exists $\delta\in(0,\delta_0]$ satisfying $\|\tilde{\mathbf{q}}'(t) - \mathbf{q}(t)\| < \min\{\delta_0,\epsilon\}$ for all $t_1\in[t_0,\infty)$, $\tilde{\mathbf{q}}'\in\tilde{S}(t_1,\delta,\eta)$, and $t\in[t_1,\infty)$. 
Similarly as in \eqref{eqn:convergence_of_delta_q} and \eqref{eqn:exponential_stability}, we can show
\begin{equation*}
\|\tilde{\mathbf{q}}'(t) - \mathbf{q}(t)\| 
\le C_\epsilon^{-1}L_f\|\tilde{\boldsymbol{\Phi}}\|\|\tilde{\mathbf{q}}'(t_1) - \mathbf{q}(t_1)\|e^{-\rho(t-t_1)}
\end{equation*}
for all $t_1\in[t_0,\infty)$, $\tilde{\mathbf{q}}'\in\tilde{S}(t_1,\delta,\eta)$, and $t\in[t_1,\infty)$.

We summarize analysis results of task convergence and stability in discrete time as follows:
\begin{theorem}
	\label{thm:task_convergence_and_stability_in_discrete_time}
	Assume \ref{ass:p_bounded_and_f_linearly_bounded} to \ref{ass:lower_bound_of_k_a} and let $\mathbf{q}:[t_0,\infty)\to\mathbb{R}^n$ be such that $\mathbf{q}(t_0) = \mathbf{q}_0$ and $\dot{\mathbf{q}}(t) = \mathbf{u}(t,\mathbf{q}(t))$ for all $t\in(t_0,\infty)$; existence and uniqueness of such $\mathbf{q}$ is guaranteed by Theorem \ref{thm:task_convergence_and_stability_in_continuous_time}.
	Assume additionally \ref{ass:theta_a_prime_is_less_than_theta_a}.
	Then, for every $T\in(0,\infty)$ there exists $\eta_T\in(0,\infty)$ such that $\|\mathbf{p}_a(t) - \mathbf{f}_a(t,\tilde{\mathbf{q}}(t))\| < \theta_a$ for all $a\in\overline{1,l}$, $\tilde{\mathbf{q}}\in\tilde{S}(\eta_T)$, and $t\in[t_0,\infty)$.
	Also, for every $T,\epsilon\in(0,\infty)$ there exists $\eta_{T,\epsilon}\in(0,\infty)$ such that $\|\tilde{\mathbf{q}}(t) - \mathbf{q}(t)\|<\epsilon$ for all $\tilde{\mathbf{q}}\in\tilde{S}(\eta_{T,\epsilon})$ and $t\in[t_0,t_0+T]$.
	Assume additionally \ref{ass:dp_F_invR_L_are_Lipschitz}.
	Then, there exists $\eta_\infty\in(0,\infty)$ such that $\|\mathbf{p}_a(t) - \mathbf{f}_a(t,\tilde{\mathbf{q}}(t))\|<\theta_a$ for all $a\in\overline{1,l}$, $\tilde{\mathbf{q}}\in\tilde{S}(\eta_\infty)$, and $t\in[t_0,\infty)$.
	Assume additionally \ref{ass:bounded_integral_from_t_0_to_infinity} and \ref{ass:psi_a_monotonic_decrease}.
	Then, there exists $\eta_0\in(0,\eta_\infty)$ such that for every $\tilde{\mathbf{q}}\in\tilde{S}(\eta_0)$, we have $\int_{t_0}^\infty\|\mathbf{p}_a(t) - \mathbf{f}_a(t,\tilde{\mathbf{q}}(t))\|dt<\infty$ and $\|\mathbf{p}_a(t) - \mathbf{f}_a(t,\tilde{\mathbf{q}}(t))\|\to0$ as $t\to\infty$ for all $a\in\overline{1,l}$ and there exists $\tilde{\mathbf{q}}_\infty\in\mathbb{R}^n$ such that $\|\tilde{\mathbf{q}}(t) - \tilde{\mathbf{q}}_\infty\|\to0$ as $t\to\infty$.
	Various stability notions of approximated joint trajectories on the joint trajectory $\mathbf{q}(t)$ can be satisfied by the following conditions:
	\begin{enumerate}
		\item \ref{ass:p_bounded_and_f_linearly_bounded} to \ref{ass:psi_a_monotonic_decrease} except for \ref{ass:m=n} and \ref{ass:zero_integral_of_psi_from_t_0_to_infinity_and_zero_initial_task_error} imply \ref{sta2:uniform_stability};
		\item \ref{ass:p_bounded_and_f_linearly_bounded} to \ref{ass:psi_a_monotonic_decrease} except for \ref{ass:zero_integral_of_psi_from_t_0_to_infinity_and_zero_initial_task_error} imply \ref{sta2:uniform_asymptotic_stability};
		\item \ref{ass:p_bounded_and_f_linearly_bounded} to \ref{ass:psi_a_monotonic_decrease} imply \ref{sta2:exponential_stability}.
	\end{enumerate}
\end{theorem}

\section{Preconditioning}
\label{sec:preconditioning}

In this section, we discuss how preconditioning can be used to overcome a limitation of the PIK problem.
The $a$-th error dynamics \eqref{eqn:error_dynamics} can be rewritten as
\begin{equation*}
	\dot{\mathbf{e}}_a = - k_a\mathbf{A}_{aa}\mathbf{e}_a + (\mathbf{I}_{m_a} - \mathbf{A}_{aa})\dot{\mathbf{p}}_a' - \sum_{b=1}^{a-1}\mathbf{A}_{ab}(\dot{\mathbf{p}}_b' + k_b\mathbf{e}_b).
\end{equation*}
The first term on the right hand side is the feedback that pushes $\mathbf{e}_a$ to zero, the second is the disturbance generated from the imperfect inversion, and the third is the influence from higher priority tasks.
It is apparent that we have the best result when $\mathbf{A}_{aa}(\mathbf{x}) = \mathbf{I}_{m_a}$ and $\mathbf{A}_{ab}(\mathbf{x}) = \mathbf{0}$ for all $1\le b < a\le l$ and $\mathbf{x}\in X$ but it is not possible in general.
Then, we should try to make $\mathbf{A}_{aa}(\mathbf{x}) \approx \mathbf{I}_{m_a}$ and $\mathbf{A}_{ab}(\mathbf{x}) \approx \mathbf{0}$ on some domain, e.g., $\mathrm{gr}(\Theta)$ in this paper, where existence of joint trajectories is guaranteed.
Indeed, the condition $\mathbf{A}_{ab}(\mathbf{x})\approx\mathbf{0}$ is crucial in order to overcome a limitation of the PIK problem.
To see this, observe from \ref{ass:lower_bound_of_k_a} that the lower bound of $k_a$ contains $\sum_{b=1}^{a-1}k_bM_{ab}\theta_b/(\theta_a'\omega_a)$ such that if $M_{ab}\gg0$, then $k_a$ tends to increase as $a$ increases.
It can be interpreted that if $M_{ab}\gg0$, then there is a large influence from higher priority tasks such that a large feedback gain $k_a$ is necessary to cancel it out.
It results that the upper bound of $\|\mathbf{u}(t,\mathbf{q}(t))\|$, defined as $M_u = M(M_p + (\sum_{a=1}^lk_a^2\theta_a^2)^{1/2})$, becomes large if $l$ becomes large.
Then, we can observe from \eqref{eqn:upper_bound_of_eta_for_phi_to_be_smaller_than_theta_for_finete_time_period} or \eqref{eqn:upper_bound_of_eta_for_phi_to_be_smaller_than_theta_for_all_time} that $\|P_{t_0}\|$ tends to decrease as $l$ increases.
However, we cannot decrease $\|P_{t_0}\|$ arbitrarily in most practical applications by various reasons, e.g., round-off error, limited computation time, etc.
Usually, $P_{t_0}$ is predefined as $\{t_0,t_0+\eta,t_0+2\eta,\dots\}$ with a fixed time step $\eta\in(0,\infty)$ in the real-time applications.
Thus, the number of tasks is restricted in most practical cases.
We may alleviate this limitation by preconditioning $\mathbf{F}_q$ properly. 
Our goal of preconditioning is to make $\mathbf{C}(\mathbf{x}) \approx \mathbf{I}_m$ such that $\mathbf{A}_{ab}(\mathbf{x}) \approx\mathbf{0}$ for $1\le b<a\le l$.
To measure how a nonzero square matrix is close to a diagonal matrix, we define a function $\mathrm{dn}:\{\text{all nonzero square matrices}\}\to[0,1]$, which we call {\it diagonalization number}, as
\begin{equation*}
\mathrm{dn}([m_{ij}]) = \frac{\sum_am_{aa}^2}{\sum_{a,b}m_{ab}^2} = \frac{\sum_am_{aa}^2}{\sum_a\sigma_a^2([m_{ij}])}.
\end{equation*}
Obviously, $[m_{ij}] \to \mathbf{I}$ as $\mathrm{dn}([m_{ij}]) \to 1$ and $m_{ii}\to1$ for all $i$.
We denote the range space and the null space of matrices as $\mathcal{R}(\cdot)$ and $\mathcal{N}(\cdot)$, respectively.

\subsection{Analysis of Preconditioning}

Let $\bar{\mathbf{J}}\in\mathbb{R}^{m\times n}$ with $m\le n$ and $r = \mathrm{rank}(\bar{\mathbf{J}})$. 
Let $\bar{\mathbf{J}} = \bar{\mathbf{C}}\hat{\bar{\mathbf{J}}}$ be the reduced QR decomposition of $\bar{\mathbf{J}}^T$ given by \cite[Lemma 1]{An2019} and $\bar{\mathbf{J}} = \mathbf{U}\bar{\boldsymbol{\Sigma}}\bar{\mathbf{V}}^T = \mathbf{U}\begin{bmatrix} \bar{\boldsymbol{\Sigma}}_m & \mathbf{0} \end{bmatrix}\begin{bmatrix} \bar{\mathbf{V}}_m & \bar{\mathbf{V}}_{n-m} \end{bmatrix}^T = \mathbf{U}\bar{\boldsymbol{\Sigma}}_m\bar{\mathbf{V}}_m^T$ be the singular value decomposition satisfying $\mathcal{R}(\hat{\bar{\mathbf{J}}}^T) = \mathcal{R}(\bar{\mathbf{V}}_m)$ where $\mathbf{U},\bar{\boldsymbol{\Sigma}}_m,\bar{\mathbf{C}} = [\bar{c}_{ij}]\in\mathbb{R}^{m\times m}$, $\bar{\boldsymbol{\Sigma}},\bar{\mathbf{V}}_m^T,\hat{\bar{\mathbf{J}}}\in\mathbb{R}^{m\times n}$, and $\bar{\mathbf{V}}\in\mathbb{R}^{n\times n}$. 
Such singular value decomposition always exists because we can freely choose right-singular vectors that correspond to the zero singular values from $\mathcal{N}(\bar{\mathbf{J}})$. 
Let $\bar{\mathbf{V}}_C = \hat{\bar{\mathbf{J}}}\bar{\mathbf{V}}_m$. 
Since $\mathcal{R}(\hat{\bar{\mathbf{J}}}^T) = \mathcal{R}(\bar{\mathbf{V}}_m)$ is assumed, we have $\bar{\mathbf{V}}_C\bar{\mathbf{V}}_C^T = \bar{\mathbf{V}}_C^T\bar{\mathbf{V}}_C = \mathbf{I}_m$ and find the singular value decomposition $\bar{\mathbf{C}} = \mathbf{U}\bar{\boldsymbol{\Sigma}}_m\bar{\mathbf{V}}_C^T$. 

Let $w \in (0,\infty)$. 
We define the preconditioner $\mathbf{R}\in\mathbb{R}^{n\times n}$ from the Cholesky decomposition $\mathbf{W} = \bar{\mathbf{J}}^T\bar{\mathbf{J}} + w^2\mathbf{I}_n = \mathbf{R}^T\mathbf{R}$ as in \cite{An2014} where $\mathbf{W}$ is symmetric and positive definite and $\mathbf{R}$ is upper triangular with positive diagonals. 
Define $\mathbf{J} = \bar{\mathbf{J}}\mathbf{R}^{-1}$ and let $\mathbf{R} = \mathbf{U}_R\boldsymbol{\Sigma}_R\mathbf{V}_R^T$ be the singular value decomposition. 
Then, we find the eigenvalue decomposition $\mathbf{W} = \bar{\mathbf{V}}(\bar{\boldsymbol{\Sigma}}^T\bar{\boldsymbol{\Sigma}} + w^2 \mathbf{I}_n)\bar{\mathbf{V}}^T = \mathbf{V}_R\boldsymbol{\Sigma}_R^2\mathbf{V}_R^T$. 
Since eigenvalues of $\mathbf{W}$ are listed in decreasing order, we have
\begin{equation*}
\boldsymbol{\Sigma}_R = (\bar{\boldsymbol{\Sigma}}^T\bar{\boldsymbol{\Sigma}} + w^2\mathbf{I}_n)^{1/2} = \begin{bmatrix} (\bar{\boldsymbol{\Sigma}}_m^2 + w^2\mathbf{I}_m)^{1/2} & \mathbf{0} \\ \mathbf{0} & w\mathbf{I}_{n-m} \end{bmatrix}.
\end{equation*}
Let $\sigma_{R,a}$ be the $a$-th diagonal entry of $\boldsymbol{\Sigma}_R$ and $\bar{\mathbf{v}}_a$ and $\mathbf{v}_{R,a}$ be the $a$-th columns of $\bar{\mathbf{V}}$ and $\mathbf{V}_R$, respectively. 
If two consecutive eigenvalues of $\mathbf{W}$ are distinctive ($\sigma_{R,a}^2 \neq \sigma_{R,a+1}^2$), then we have $\begin{bmatrix} \bar{\mathbf{v}}_a & \bar{\mathbf{v}}_{a+1} \end{bmatrix}^T\begin{bmatrix} \mathbf{v}_{R,a} & \mathbf{v}_{R,a+1} \end{bmatrix} = \mathrm{diag}(\pm1,\pm1)$ because the eigenspaces that correspond to different eigenvalues are orthogonal \cite{Trefethen1997}. 
If there are repeated eigenvalues ($\sigma_{R,a}^2 = \sigma_{R,a+1}^2$), then we can write $\mathrm{diag}(\sigma_{R,a}^{-1},\sigma_{R,a+1}^{-1}) = \sigma_{R,a}^{-1}\mathbf{I}_2$. 
From these properties, we can derive $(\bar{\mathbf{V}}^T\mathbf{V}_R)\boldsymbol{\Sigma}_R^{-1} = \boldsymbol{\Sigma}_R^{-1}(\bar{\mathbf{V}}^T\mathbf{V}_R)$ and $\mathbf{J} = \bar{\mathbf{J}}\mathbf{R}^{-1} = \mathbf{U}(\bar{\boldsymbol{\Sigma}}\boldsymbol{\Sigma}_R^{-1})(\mathbf{U}_R\mathbf{V}_R^T\bar{\mathbf{V}})^T$. 
Define $\boldsymbol{\Sigma} = \bar{\boldsymbol{\Sigma}}\boldsymbol{\Sigma}_R^{-1}$ and $\mathbf{V} = \mathbf{U}_R\mathbf{V}_R^T\bar{\mathbf{V}}$. 
Let $\bar{\sigma}_a$ and $\sigma_a$ be the $a$-th diagonal entries of $\bar{\boldsymbol{\Sigma}}$ and $\boldsymbol{\Sigma}$, respectively. 
Then, $\sigma_a = \bar{\sigma}_a/\sqrt{\bar{\sigma}_a^2 + w^2}$ is nonnegative and monotonically increasing with respect to $\bar{\sigma}_a$. 
Since $\mathbf{V}$ is orthogonal, we have the singular value decomposition $\mathbf{J} = \mathbf{U}\boldsymbol{\Sigma}\mathbf{V}^T = \mathbf{U}\begin{bmatrix} \boldsymbol{\Sigma}_m & \mathbf{0} \end{bmatrix} \begin{bmatrix} \mathbf{V}_m & \mathbf{V}_{n-m} \end{bmatrix}^T = \mathbf{U}\boldsymbol{\Sigma}_m\mathbf{V}_m^T$ where $\boldsymbol{\Sigma}_m = \bar{\boldsymbol{\Sigma}}_m(\bar{\boldsymbol{\Sigma}}_m^2+w^2\mathbf{I}_m)^{-1/2}$.

Let $\mathbf{J} = \mathbf{C}\hat{\mathbf{J}}$ be the reduced QR decomposition of $\mathbf{J}^T$ given by \cite[Lemma 1]{An2019} where $\mathbf{C}=[c_{ij}]\in\mathbb{R}^{m\times m}$ and $\hat{\mathbf{J}}\in\mathbb{R}^{m\times n}$. 
$\bar{\mathbf{C}}$ and $\mathbf{C}$ have same zero diagonal entries because right multiplying $\bar{\mathbf{J}}$ by $\mathbf{R}^{-1}$ does not change linear dependence of two rows of $\bar{\mathbf{J}}$. 
Construct $\mathbf{C}_r\in\mathbb{R}^{m\times r}$ by removing zero columns from $\mathbf{C}$ and $\hat{\mathbf{J}}_r\in\mathbb{R}^{r\times n}$ by removing rows from $\hat{\mathbf{J}}$ that correspond to zero columns of $\mathbf{C}$ such that $\mathbf{J} = \mathbf{C}_r\hat{\mathbf{J}}_r$. 
Let $\mathbf{U}_r$ be the left $(m\times r)$ block of $\mathbf{U}$; $\bar{\boldsymbol{\Sigma}}_r$ and $\boldsymbol{\Sigma}_r$ be the top left $(r\times r)$ blocks of $\bar{\boldsymbol{\Sigma}}$ and $\boldsymbol{\Sigma}$, respectively; and $\mathbf{V}_r$ be the left $(n\times r)$ block of $\mathbf{V}$ such that $\mathbf{J} = \mathbf{U}_r\boldsymbol{\Sigma}_r\mathbf{V}_r^T$ and $\boldsymbol{\Sigma}_r = \bar{\boldsymbol{\Sigma}}_r(\bar{\boldsymbol{\Sigma}}_r^2 + w^2\mathbf{I}_r)^{1/2}$. 
Then, $\mathbf{C}_r = \mathbf{U}_r\boldsymbol{\Sigma}_r(\hat{\mathbf{J}}_r\mathbf{V}_r)^T$. 
Since $\mathcal{R}(\mathbf{J}^T) = \mathcal{R}(\hat{\mathbf{J}}_r^T) = \mathcal{R}(\mathbf{V}_r)$, $\hat{\mathbf{J}}_r\mathbf{V}_r$ is orthogonal and we can derive
\begin{align*}
\mathbf{C}\mathbf{C}^T &= \mathbf{C}_r\mathbf{C}_r^T \\
&= \mathbf{U}_r\bar{\boldsymbol{\Sigma}}_r(\bar{\boldsymbol{\Sigma}}_r^2 + w^2\mathbf{I}_r)^{-1}\bar{\boldsymbol{\Sigma}}_r\mathbf{U}_r^T \\
&= \mathbf{U}\bar{\boldsymbol{\Sigma}}_m(\bar{\boldsymbol{\Sigma}}_m^2 + w^2\mathbf{I}_m)^{-1}\bar{\boldsymbol{\Sigma}}_m\mathbf{U}^T \\
&= \bar{\mathbf{C}}(\bar{\mathbf{C}}^T\bar{\mathbf{C}} + w^2\mathbf{I}_m)^{-1}\bar{\mathbf{C}}^T.
\end{align*}
Since $\bar{\mathbf{C}}^T\bar{\mathbf{C}} + w^2\mathbf{I}_m$ is symmetric and positive definite, we can consider the reverse Cholesky decomposition $\bar{\mathbf{C}}^T\bar{\mathbf{C}} + w^2\mathbf{I}_m = \tilde{\mathbf{C}}^T\tilde{\mathbf{C}}$ where $\tilde{\mathbf{C}}=[\tilde{c}_{ij}]\in\mathbb{R}^{m\times m}$ is lower triangular with positive diagonals. 
Then, we have $\mathbf{C}\mathbf{C}^T = (\bar{\mathbf{C}}\tilde{\mathbf{C}}^{-1})(\bar{\mathbf{C}}\tilde{\mathbf{C}}^{-1})^T$. 
Observe that if $\bar{c}_{aa} = 0$, then all entries of the $a$-th column and row of $\bar{\mathbf{C}}^T\bar{\mathbf{C}} + w^2\mathbf{I}_m$ are zeros except for the $a$-th diagonal that is $w^2$, so all entries of the $a$-th columns of $\tilde{\mathbf{C}}$ and $\tilde{\mathbf{C}}^{-1}$ are zeros except for $a$-th diagonals that are $w$ for $\tilde{\mathbf{C}}$ and $1/w$ for $\tilde{\mathbf{C}}^{-1}$. 

We prove $\mathbf{C} = \bar{\mathbf{C}}\tilde{\mathbf{C}}^{-1}$ by induction. 
It is obvious if $m = 1$. 
Assume that it holds for some $m - 1 \ge 1$. 
Let $\mathbf{A} = \begin{bmatrix} \mathbf{a}_1 & \cdots & \mathbf{a}_{m-1}\end{bmatrix}\in\mathbb{R}^{(m-1)\times(m-1)}$ and $\mathbf{b} = (b_1,\dots,b_{m-1})\in\mathbb{R}^{1\times(m-1)}$ be such that $\mathbf{C} = \begin{bmatrix} \mathbf{A} & \mathbf{0} \\ \mathbf{b} & c_{mm} \end{bmatrix}$. 
$\bar{\mathbf{A}}$, $\bar{\mathbf{b}}$, $\tilde{\mathbf{A}}$, and $\tilde{\mathbf{b}}$ have the same meaning for $\bar{\mathbf{C}}$ and $\tilde{\mathbf{C}}$. 
From $\mathbf{C}\mathbf{C}^T = (\bar{\mathbf{C}}\tilde{\mathbf{C}}^{-1})(\bar{\mathbf{C}}\tilde{\mathbf{C}}^{-1})^T$ and
\begin{align*}
\bar{\mathbf{C}}\tilde{\mathbf{C}}^{-1} 
&= \begin{bmatrix} \bar{\mathbf{A}} & \mathbf{0} \\ \bar{\mathbf{b}} & c_{mm} \end{bmatrix} \begin{bmatrix} \tilde{\mathbf{A}}^{-1} & \mathbf{0} \\ \frac{1}{\tilde{c}_{mm}}\tilde{\mathbf{b}}\tilde{\mathbf{A}}^{-1} & \frac{1}{\tilde{c}_{mm}} \end{bmatrix} \\
&= \begin{bmatrix} \bar{\mathbf{A}}\tilde{\mathbf{A}}^{-1} & \mathbf{0} \\ \left(\bar{\mathbf{b}} + \frac{\bar{c}_{mm}}{\tilde{c}_{mm}}\tilde{\mathbf{b}}\right)\tilde{\mathbf{A}}^{-1} & \frac{\bar{c}_{mm}}{\tilde{c}_{mm}} \end{bmatrix},
\end{align*}
we find the next three relations:
\begin{align*}
\mathbf{A}\mathbf{A}^T &= (\bar{\mathbf{A}}\tilde{\mathbf{A}}^{-1})(\bar{\mathbf{A}}\tilde{\mathbf{A}}^{-1})^T \\
\mathbf{b}\mathbf{A}^T &= \hat{\mathbf{b}}(\bar{\mathbf{A}}\tilde{\mathbf{A}}^{-1})^T \\
\mathbf{b}\mathbf{b}^T + c_{mm}^2 &= \hat{\mathbf{b}}\hat{\mathbf{b}}^T + \frac{\bar{c}_{mm}^2}{\tilde{c}_{mm}^2}
\end{align*}
where $\hat{\mathbf{b}} = (\hat{b}_1,\dots,\hat{b}_{m-1}) = \left(\bar{\mathbf{b}} + \frac{\bar{c}_{mm}}{\tilde{c}_{mm}}\tilde{\mathbf{b}}\right)\tilde{\mathbf{A}}^{-1}$. 
By the assumption, the first and second relations give $\mathbf{A} = \bar{\mathbf{A}}\tilde{\mathbf{A}}^{-1}$ and $\mathbf{A}(\mathbf{b} - \hat{\mathbf{b}})^T = \sum_{a=1}^{m-1}\mathbf{a}_a(b_a-\hat{b}_a) = \mathbf{0}$. 
If $c_{aa} = 0$ for some $a\in\overline{1,m-1}$, then $b_a = \bar{b}_a = \tilde{b}_a = \hat{b}_a = 0$ because all entries of the $a$-th columns of $\mathbf{C}$ and $\bar{\mathbf{C}}$ are zeros and all entries of the $a$-th columns of $\tilde{\mathbf{C}}$ and $\tilde{\mathbf{C}}^{-1}$ are zeros except for the $a$-th diagonal. 
So, we find $\mathbf{b} = \hat{\mathbf{b}}$ from
\begin{equation*}
\mathbf{A}(\mathbf{b} - \hat{\mathbf{b}})^T = \sum_{\substack{c_{aa}\neq 0\\a\in\overline{1,m-1}}} (b_a-\hat{b}_a)\mathbf{a}_a = \mathbf{0}
\end{equation*}
because $\{\mathbf{a}_a\mid c_{aa} \neq 0,\,a\in\overline{1,m-1}\}$ is linearly independent. 
Then, the third relation gives $c_{mm} = \bar{c}_{mm}/\tilde{c}_{mm}$. 
Therefore, $\mathbf{C} = \bar{\mathbf{C}}\tilde{\mathbf{C}}^{-1}$.

We find bounds and tendencies of $c_{aa}^2$ and $\mathrm{dn}(\mathbf{C})$ when $w\to\infty$ or $w\to0$.
Let $\tilde{\mathbf{C}} = \tilde{\mathbf{U}}\tilde{\boldsymbol{\Sigma}}\tilde{\mathbf{V}}^T$ be the singular value decomposition. 
Then, $\mathbf{W}_C = \bar{\mathbf{C}}^T\bar{\mathbf{C}} + w^2\mathbf{I}_m = \tilde{\mathbf{C}}^T\tilde{\mathbf{C}}$ can be written as $\mathbf{W}_C = \bar{\mathbf{V}}_C(\bar{\boldsymbol{\Sigma}}_m^2 + w^2\mathbf{I}_m)\bar{\mathbf{V}}_C^T = \tilde{\mathbf{V}}\tilde{\boldsymbol{\Sigma}}^2\tilde{\mathbf{V}}^T$ using $\bar{\mathbf{C}} = \mathbf{U}\bar{\boldsymbol{\Sigma}}_m\bar{\mathbf{V}}_C^T$. 
Similarly as before, we find $\tilde{\boldsymbol{\Sigma}} = (\bar{\boldsymbol{\Sigma}}_m^2 + w^2\mathbf{I}_m)^{1/2}$ and the singular value decomposition of $\mathbf{C} = \bar{\mathbf{C}}\tilde{\mathbf{C}}^{-1}$ as $\mathbf{C} = \mathbf{U}(\bar{\boldsymbol{\Sigma}}_m\tilde{\boldsymbol{\Sigma}}^{-1})(\tilde{\mathbf{U}}\tilde{\mathbf{V}}^T\bar{\mathbf{V}}_C)^T = \mathbf{U}\boldsymbol{\Sigma}_m\mathbf{V}_C^T$. 
Let $\bar{\mathbf{C}}_{a:a',b:b'}$ and $\tilde{\mathbf{C}}_{a:a',b:b'}$ be the blocks of $\bar{\mathbf{C}}$ and $\tilde{\mathbf{C}}$ with the top left entries $\bar{c}_{a,b}$ and $\tilde{c}_{a,b}$ and the bottom right entries $\bar{c}_{a',b'}$ and $\tilde{c}_{a',b'}$, respectively. 
By expanding $\bar{\mathbf{C}}^T\bar{\mathbf{C}} + w^2\mathbf{I}_m = \tilde{\mathbf{C}}^T\tilde{\mathbf{C}}$, we find a set of relations
\begin{align*}
\tilde{c}_{aa}^2 &= \bar{c}_{aa}^2 + w^2 + \|\bar{\mathbf{d}}_a\|^2 - \|\tilde{\mathbf{d}}_a\|^2 \\
\tilde{\mathbf{C}}_a^T\tilde{\mathbf{d}}_a &= \bar{\mathbf{C}}_a^T\bar{\mathbf{d}}_a \\
\tilde{\mathbf{C}}_a^T\tilde{\mathbf{C}}_a &= \bar{\mathbf{C}}_a^T\bar{\mathbf{C}}_a + w^2\mathbf{I}_{m-a}
\end{align*}
for $a\in\overline{1,m}$ where $\bar{\mathbf{d}}_a = \bar{\mathbf{C}}_{a+1:m,a:a}$, $\tilde{\mathbf{d}}_a = \tilde{\mathbf{C}}_{a+1:m,a:a}$, $\bar{\mathbf{C}}_a = \bar{\mathbf{C}}_{a+1:m,a+1:m}$, $\tilde{\mathbf{C}}_a = \tilde{\mathbf{C}}_{a+1:m,a+1:m}$, and $\bar{\mathbf{d}}_m = \tilde{\mathbf{d}}_m = \bar{\mathbf{C}}_m = \tilde{\mathbf{C}}_m = \mathbf{I}_0 = 0$. 
Let $\bar{\mathbf{C}}_a = \bar{\mathbf{U}}_{C_a}\bar{\boldsymbol{\Sigma}}_{C_a}\bar{\mathbf{V}}_{C_a}^T$ and $\tilde{\mathbf{C}}_a = \tilde{\mathbf{U}}_{C_a}\tilde{\boldsymbol{\Sigma}}_{C_a}\tilde{\mathbf{V}}_{C_a}^T$ be the singular value decompositions. 
Similarly as before, we find $\tilde{\boldsymbol{\Sigma}}_{C_a} = (\bar{\boldsymbol{\Sigma}}_{C_a}^2 + w^2\mathbf{I}_{m-a})^{1/2}>0$ from the third relation. 
So, $\tilde{\mathbf{C}}_a^T$ is invertible and we can rewrite the second relation as $\tilde{\mathbf{d}}_a = \left( \bar{\mathbf{U}}_{C_a}(\bar{\boldsymbol{\Sigma}}_{C_a}\tilde{\boldsymbol{\Sigma}}_{C_a}^{-1})(\tilde{\mathbf{U}}_{C_a}\tilde{\mathbf{V}}_{C_a}^T\bar{\mathbf{V}}_{C_a})^T\right)^T\bar{\mathbf{d}}_a$ and find the bounds of $\nu_a^2 = \|\bar{\mathbf{d}}_a\|^2 - \|\tilde{\mathbf{d}}_a\|^2$ as
\begin{equation}
\label{eqn:bounds_of_nu}
\underline{\nu}_a^2 = \frac{w^2\|\bar{\mathbf{d}}_a\|^2}{\sigma_{\max}^2(\bar{\mathbf{C}}_a) +w^2} \le \nu_a^2 \le \frac{w^2\|\bar{\mathbf{d}}_a\|^2}{\sigma_{\min}^2(\bar{\mathbf{C}}_a) +w^2} = \overline{\nu}_a^2.
\end{equation}
From $\mathbf{C} = \bar{\mathbf{C}}\tilde{\mathbf{C}}^{-1}$ and the first relation, we can write $c_{aa}^2 = \bar{c}_{aa}^2/\tilde{c}_{aa}^2 = \bar{c}_{aa}^2/(\bar{c}_{aa}^2 + w^2 + \nu_a^2)$. 
Then, we have
\begin{equation*}
\mathrm{dn}(\mathbf{C}) = \frac{\sum_{a=1}^mc_{aa}^2}{\sum_{a=1}^m\sigma_a^2} = \frac{\sum_{a=1}^m\frac{\bar{c}_{aa}^2}{\bar{c}_{aa}^2+w^2+\nu_a^2}}{\sum_{a=1}^m\frac{\bar{\sigma}_a^2}{\bar{\sigma}_a^2+w^2}}
\end{equation*}
where $\sigma_a$ and $\bar{\sigma}_a$ be the $a$-th singular values of $\mathbf{C}$ and $\bar{\mathbf{C}}$, respectively.
We can find the bounds of $c_{aa}^2$ from \eqref{eqn:bounds_of_nu} as
\begin{align*}
\alpha_a^2
= \frac{\bar{c}_{aa}^2}{\bar{c}_{aa}^2 + w^2 + \overline{\nu}_a^2}
\le c_{aa}^2
\le \frac{\bar{c}_{aa}^2}{\bar{c}_{aa}^2 + w^2 + \underline{\nu}_a^2}
= \beta_a^2
\end{align*}
and the bounds of $\mathrm{dn}(\mathbf{C})$ as
\begin{equation*}
\frac{1}{\rho^2}\sum_{a=1}^m\alpha_a^2 \le \mathrm{dn}(\mathbf{C}) \le \frac{1}{\rho^2}\sum_{a=1}^m\beta_a^2
\end{equation*}
where $\rho^2 = \sum_{a=1}^m\bar{\sigma}_a^2/(\bar{\sigma}_a^2+w^2)$.
Since $\nu_a^2 \to \|\bar{\mathbf{d}}_a\|^2$ and $\nu_a^2/w^2\to0$ as $w\to\infty$, we have
\begin{align*}
\lim_{w\to\infty}c_{aa}^2 &= \lim_{w\to\infty}\frac{\bar{c}_{aa}^2}{\bar{c}_{aa}^2 + w^2 + \nu_a^2} = 0 \\
\lim_{w\to\infty}\mathrm{dn}(\mathbf{C}) &= \lim_{w\to\infty}\frac{\sum_{a=1}^m\frac{\bar{c}_{aa}^2}{\bar{c}_{aa}^2/w^2 + 1 + \nu_a^2/w^2}}{\sum_{a=1}^m\frac{\bar{\sigma}_a^2}{\bar{\sigma}_a^2/w^2 + 1}} = \mathrm{dn}(\bar{\mathbf{C}}).
\end{align*}
If $r = m$, then $\bar{c}_{aa}>0$ and $\sigma_{\min}(\bar{\mathbf{C}}_a) > 0$ that leads to $\nu_a^2\to 0$ as $w\to0$, so we have
\begin{align*}
\lim_{w\to\infty}c_{aa}^2 &= \frac{\bar{c}_{aa}^2}{\bar{c}_{aa}^2} = 1 \\
\lim_{w\to\infty}\mathrm{dn}(\mathbf{C}) &= \frac{\sum_{a=1}^m\bar{c}_{aa}^2/\bar{c}_{aa}^2}{\sum_{a=1}^m\bar{\sigma}_a^2/\bar{\sigma}_a^2} = 1.
\end{align*}
We summarize analysis results of preconditioning as follows:
\begin{theorem}
	\label{thm:preconditioning}
	Let $\bar{\mathbf{J}} = \bar{\mathbf{C}}\hat{\bar{\mathbf{J}}}\in\mathbb{R}^{m\times n}$ be the reduced QR decomposition given by \cite[Lemma 1]{An2019} with $m\le n$; $\bar{\mathbf{J}}^T\bar{\mathbf{J}} + w^2\mathbf{I}_n = \mathbf{R}^T\mathbf{R}$ be the Cholesky decomposition with $w\in(0,\infty)$ and an upper triangular matrix $\mathbf{R}\in\mathbb{R}^{n\times n}$ whose diagonals are positive; and $\mathbf{J} = \bar{\mathbf{J}}\mathbf{R}^{-1} = \mathbf{C}\hat{\mathbf{J}}$ be the reduced QR decomposition given by \cite[Lemma 1]{An2019}. 
	Then, the preconditioning of $\bar{\mathbf{J}}$ by $\mathbf{R}$ has the following properties:
	\begin{enumerate}
		\item $\alpha_a^2 \le c_{aa}^2 \le \beta_a^2$ and $\displaystyle\frac{1}{\rho^2}\sum_{a=1}^m\alpha_a^2 \le \mathrm{dn}(\mathbf{C}) \le \frac{1}{\rho^2}\sum_{a=1}^m\beta_a^2$;
		\item $c_{aa}^2\to0$ and $\mathrm{dn}(\mathbf{C})\to\mathrm{dn}(\bar{\mathbf{C}})$ as $w\to\infty$;
		\item if $\mathrm{rank}(\bar{\mathbf{J}}) = m$, then $c_{aa}^2\to1$ and $\mathrm{dn}(\mathbf{C})\to1$ as $w\to0$.
	\end{enumerate}
	where $\bar{c}_{aa}$ and $c_{aa}$ are the $a$-th diagonals of $\bar{\mathbf{C}}$ and $\mathbf{C}$, respectively, $\bar{\sigma}_a$ is the $a$-th singular value of $\bar{\mathbf{J}}$, $\overline{\nu}_a$ and $\underline{\nu}_a$ are as in \eqref{eqn:bounds_of_nu}, $\alpha_a^2 = \bar{c}_{aa}^2/(\bar{c}_{aa}^2 + w^2 + \overline{\nu}_a^2)$, $\beta_a^2 = \bar{c}_{aa}^2/(\bar{c}_{aa}^2+w^2+\underline{\nu}_a^2)$, and $\rho^2=\sum_{a=1}^m\bar{\sigma}_a^2/(\bar{\sigma}_a^2+w^2)$.
\end{theorem}

\begin{corollary}
	\label{cor:preconditioning_norm_bounds}
	Let $1\le a\le a'\le m$ and $1\le b\le b'\le m$. 
	The preconditioning given by Theorem \ref{thm:preconditioning} has the following properties:
	\begin{enumerate}
		\item $|c_{ab}| < 1$;
		\item $\|\mathbf{C}_{a:a',b:b'}\| < 1$;
		\item $\|\mathbf{C}_{a:a',b:b'}\|_1 < a'-a+1$;
		\item $\|\mathbf{C}_{a:a',b:b'}\|_\infty < b'-b+1$;
		\item $\|\mathbf{C}_{a:a',b:b'}\|_F < \min\{\sqrt{a'-a+1},\sqrt{b'-b+1}\}$.
	\end{enumerate}
\end{corollary}
\begin{proof}
	From the singular value decomposition $\mathbf{C} = \mathbf{U}\boldsymbol{\Sigma}_m\mathbf{V}_C^T$, we can write $c_{ab} = \mathbf{u}_a\boldsymbol{\Sigma}_m\mathbf{v}_{C,b}^T$ where $\mathbf{u}_a$ and $\mathbf{v}_{C,b}$ are the $a$-th and $b$-th rows of $\mathbf{U}$ and $\mathbf{V}_C$, respectively. 
	Then, $|c_{ab}| \le \|\mathbf{u}_a\|\|\boldsymbol{\Sigma}_m\|\|\mathbf{v}_{C,b}\| = \sigma_{\max}(\boldsymbol{\Sigma}_m) = \bar{\sigma}_1/\sqrt{\bar{\sigma}_1^2+w^2} < 1$. 
	Similarly, $\|\mathbf{C}_{a:a',b:b'}\| = \|\mathbf{U}_{a:a',1:m}\boldsymbol{\Sigma}_m\mathbf{V}_{C,b:b',1:m}^T\| \le \|\mathbf{U}_{a:a',1:m}\|\|\boldsymbol{\Sigma}_m\|\|\mathbf{V}_{C,b:b',1:m}\| < 1$. 
	The rest of the properties are directly obtained from the first two properties by definitions.
\end{proof}

Let $\bar{\mathbf{J}} = \mathbf{F}_q:X\to\mathbb{R}^{m\times n}$ with $m\le n$ be continuous on $X$ and $w\in(0,\infty)$. 
We recall that a matrix-valued function is continuous if and only if all entries are continuous.
Also, the pseudoinverse of a continuous matrix-valued function is continuous at a point if and only if the function has a local constant rank \cite{Stewart1969}.
Define a preconditioner function $\mathbf{R}_w:X\to\mathbb{R}^{n\times n}$ from the Cholesky decomposition $\mathbf{W}_w(\mathbf{x}) = (\bar{\mathbf{J}}^T\bar{\mathbf{J}})(\mathbf{x}) + w^2\mathbf{I}_n = (\mathbf{R}_w^T\mathbf{R}_w)(\mathbf{x})$ for each $\mathbf{x}\in X$ where $\mathbf{R}_w(\mathbf{x})$ is upper triangular with positive diagonals. 
Since the algorithm to find the Cholesky decomposition consists of continuous functions of continuous entries of $\mathbf{W}_w$, $\mathbf{R}_w$ is continuous on $X$ and so are $\mathbf{R}_w^{-1}$ and $\mathbf{J}_w = \bar{\mathbf{J}}\mathbf{R}_w^{-1}$. 
Define $\mathcal{G} = \{\mathbf{x}\in X\mid\mathrm{rank}(\bar{\mathbf{J}}(\mathbf{x})) = m\}$.
Observe that $\mathcal{G}\subset X$ is open because $\bar{\mathbf{J}}$ is continuous on $X$.
Since the preconditioning does not change linear dependence of rows of $\bar{\mathbf{J}}$, $\mathcal{G} = \{\mathbf{x}\in X\mid\mathrm{rank}(\mathbf{J}_w(\mathbf{x})) = m\}$ for all $w\in(0,\infty)$. 
If $\mathbf{x}\in\mathcal{G}$, then the reduced QR decompositions $\bar{\mathbf{J}}(\mathbf{x}) = (\bar{\mathbf{C}}\hat{\bar{\mathbf{J}}})(\mathbf{x})$ and $\mathbf{J}_w(\mathbf{x}) = (\mathbf{C}_w\hat{\mathbf{J}}_w)(\mathbf{x})$ can be found by the modified Gram-Schmidt orthogonalization that consists of continuous functions of continuous entries of $\bar{\mathbf{J}}$ and $\mathbf{J}_w$.
Thus, $\bar{\mathbf{C}} = [\bar{c}_{ij}]$ and $\mathbf{C}_w = [c_{w,ij}]$ are continuous on $\mathcal{G}$ for all $w\in(0,\infty)$.
The next corollary shows that $\mathbf{C}_w$ converges to $\mathbf{I}_m$ uniformly on any compact set $\Omega\subset \mathcal{G}$ as $w\to0$.

\begin{corollary}
	\label{cor:preconditioning}
	For every compact set $\Omega\subset\mathcal{G}$ and every $\epsilon>0$ there exists $w_0\in(0,\infty)$ such that if $w\in(0,w_0)$, $a\in\overline{1,m}$, and $\mathbf{x}\in\Omega$, then $|1 - c_{w,aa}^2(\mathbf{x})|<\epsilon$ and $|1-\mathrm{dn}(\mathbf{C}_w(\mathbf{x}))|<\epsilon$.
\end{corollary}
\begin{proof}
	Let $\Omega\subset\mathcal{G}$ be a compact set and $\epsilon\in(0,1)$ be arbitrary.
	Since $\bar{\mathbf{C}}$ is continuous and $\mathrm{rank}(\bar{\mathbf{C}}(\mathbf{x})) = m$ on the compact set $\Omega$, there exists $m_\Omega \in(0,\infty)$ satisfying $\sigma_{\min}(\bar{\mathbf{C}}(\mathbf{x}))\ge m_\Omega$ for all $\mathbf{x}\in\Omega$.
	Let $w_0 = m_\Omega(\epsilon/(1-\epsilon))^{1/2}$.
	Then,
	\begin{equation*}
	|1-c_{w,aa}^2(\mathbf{x})| \le 1 - \sigma_{\min}^2(\mathbf{C}_w(\mathbf{x})) \le 1 - \frac{m_\Omega^2}{m_\Omega^2 + w^2} < \epsilon
	\end{equation*}
	and
	\begin{equation*}
	|1-\mathrm{dn}(\mathbf{C}_w(\mathbf{x}))|
	\le \frac{1}{m}\sum_{a=1}^m(1 - c_{w,aa}^2(\mathbf{x}))
	< \epsilon
	\end{equation*}
	for all $w\in(0,w_0)$, $a\in\overline{1,m}$, and $\mathbf{x}\in\Omega$.
\end{proof}

\subsection{Effect of Preconditioning}

It would be interesting to see the effect of preconditioning on a specific PIK solution.
Let $\mathbf{u} = \mathbf{R}^{-1}\hat{\mathbf{J}}^T\mathbf{C}_D^T(\dot{\mathbf{p}}' + \mathbf{K}\mathbf{e})$; the proper objective function $\boldsymbol{\pi}$ that generates this PIK solution can be found in \cite{An2019}.
In many practical applications, we could reasonably assume that $\mathbf{f}$ is linearly bounded and twice differentiable; the derivatives of $\mathbf{f}$ are bounded; $\mathbf{f}(t,\cdot) = \mathbf{f}(0,\cdot)$ for all $t\in\mathbb{R}$; and $m = n$.
In such a case, we write $\mathbf{f}(\mathbf{q}) = \mathbf{f}(t,\mathbf{q})$ and $\mathbf{F}_q(\mathbf{q}) = \mathbf{F}_q(t,\mathbf{q})$ for the sake of simplicity in the notation.
We need to define a desired task trajectory $\mathbf{p}:\mathbb{R}\to\mathbb{R}^m$ and a set-valued map $\Theta:\mathbb{R}\to2^{\mathbb{R}^n}$ that satisfy \ref{ass:p_bounded_and_f_linearly_bounded}, \ref{ass:dp_F_invR_locally_Lipschitz_and_bounded}, and \ref{ass:set_valued_map_Theta}.
Let $\Theta_0\subset\mathbb{R}^n$ be a compact and convex set and $r_\Theta\in(0,\infty)$ be such that $\mathrm{int}(\Theta_0) \neq\emptyset$; $\mathbf{f}$ is one-to-one on $\Theta_0 + r_\Theta B_n$; and $\mathrm{rank}(\mathbf{F}_q(\mathbf{q})) = m$ for all $\mathbf{q}\in\Theta_0$.
Since $\mathbf{f}(\mathrm{int}(\Theta_0))$ is open \cite[Theorem 9.25]{Rudin1964}, there exist $\hat{\mathbf{p}} = (\hat{\mathbf{p}}_1,\dots,\hat{\mathbf{p}}_l)\in\mathbb{R}^m$ and $\theta\in(0,\infty)$ satisfying $\bigtimes_{a=1}^l(\hat{\mathbf{p}}_a + \theta B_{m_a})\subset\mathbf{f}(\Theta_0)$.
Since $\mathbf{f}$ is assumed to be one-to-one on $\Theta_0+r_\Theta B_n$, we have $\bigtimes_{a=1}^l(\hat{\mathbf{p}}_a+\theta B_{m_a})\cap\mathbf{f}((\Theta_0+r_\Theta B_n)\setminus\Theta_0) = \emptyset$.
Let $\mathbf{p}(t) \equiv \hat{\mathbf{p}}$ and $\Theta(t) \equiv \Theta_0$.

Let $\mathbf{F}_q = \bar{\mathbf{C}}\hat{\bar{\mathbf{J}}}$ be the reduced QR decomposition given by \cite[Lemma 1]{An2019}.
Since $\bar{\mathbf{C}}$ is continuous and $\mathrm{rank}(\bar{\mathbf{C}}(\mathbf{q})) = m$ on the compact set $\Theta_0$, there exist $M_\Theta,m_\Theta\in(0,\infty)$ satisfying $\|\mathbf{F}_q(\mathbf{q})\| = \|\bar{\mathbf{C}}(\mathbf{q})\|\le M_\Theta$ and $\sigma_{\min}(\bar{\mathbf{C}}(\mathbf{q}))\ge m_\Theta$ for all $\mathbf{q}\in\Theta_0$.
Then, $m_\Theta \le \bar{c}_{aa}(\mathbf{q})\le M_\Theta$ for all $a\in\overline{1,m}$ and $\mathbf{q}\in\Theta_0$ by the Weyl's product inequality.
Also, $m_\Theta\le\sigma_i(\bar{\mathbf{C}}_{aa}(\mathbf{q}))\le M_\Theta$ for all $a\in\overline{1,l}$, $i\in\overline{1,m_a}$, and $\mathbf{q}\in\Theta_0$ because $\bar{\mathbf{C}}_{aa}$ and $\bar{\mathbf{C}}_{aa}^{-1}$ are the $a$-th diagonal block of $\bar{\mathbf{C}}$ and $\bar{\mathbf{C}}^{-1}$, respectively.
Let $w\in(0,\infty)$ be arbitrary, $\mathbf{W} = \mathbf{F}_q^T\mathbf{F}_q + w^2\mathbf{I}_n = \mathbf{R}^T\mathbf{R}$ be the Cholesky decomposition, and $\mathbf{J} = \mathbf{F}_q\mathbf{R}^{-1} = \mathbf{C}\hat{\mathbf{J}}$ be the reduced QR decomposition given by \cite[Lemma 1]{An2019}.
Then, $1 > c_{aa}^2(\mathbf{q}), \sigma_a^2(\mathbf{C}(\mathbf{q})) \ge \sigma_{\min}^2(\mathbf{C}(\mathbf{q})) \ge \omega = m_\Theta^2/(m_\Theta^2+w^2)$ for all $a\in\overline{1,m}$ and $\mathbf{q}\in\Theta_0$.
It follows that $\sigma_{\min}(\mathbf{C}_{aa}(\mathbf{q})) \ge \sqrt{\omega}$ and $(\mathbf{A}_{aa} + \mathbf{A}_{aa}^T)(\mathbf{q}) - 2\omega\mathbf{I}_{m_a}\ge0$ for all $a\in\overline{1,l}$ and $\mathbf{q}\in\Theta_0$.
Obviously, $\mathbf{L} = \mathbf{I}_m$ is Lipschitz and bounded.
Let $\theta'\in(0,\theta)$, $t_0\in\mathbb{R}$, and $\mathbf{q}_0\in\Theta_0$ be such that $\|\hat{\mathbf{p}}_a - \mathbf{f}_a(\mathbf{q}_0)\|<\theta'$ for all $a\in\overline{1,l}$.
Let $L_\Theta\in(0,\infty)$ be such that $\|\mathsf{D}_q\mathbf{F}_q(\mathbf{q})\|\le L_\Theta$ for all $\mathbf{q}\in\Theta_0$ where $\mathsf{D}_q\mathbf{F}_q = \partial\mathbf{F}_q/\partial\mathbf{q}$.
Since $\Theta_0$ is convex, $\mathbf{F}_q$ is Lipschitz on $\Theta_0$ with the Lipschitz constant $L_\Theta$. 
From $\|\mathbf{A}_{ab}(\mathbf{q})\| \le \|\mathbf{C}_{ab}(\mathbf{q})\| \le \|\mathbf{C}(\mathbf{q})\|_F\sqrt{1-\mathrm{dn}(\mathbf{C}(\mathbf{q}))} \le \sqrt{m(1-\omega)}$ for all $a\in\overline{1,l}$, $b\in\overline{1,a-1}$, and $\mathbf{q}\in\Theta_0$,
one can easily check that \ref{ass:lower_bound_of_k_a} can be satisfied by letting $k_a = k(1 + \theta\sqrt{m(1-\omega)}/(\theta'\omega))^{a-1}$ for $a\in\overline{1,l}$ with an arbitrary $k\in(0,\infty)$. 
Therefore, there exists a unique $\mathbf{q}:[t_0,\infty)\to\mathbb{R}^n$ satisfying $\mathbf{q}(t_0) = \mathbf{q}_0$ and $\dot{\mathbf{q}}(t) = \mathbf{u}(\mathbf{q}(t))$ for all $t\in(t_0,\infty)$ such that $\|\hat{\mathbf{p}}_a - \mathbf{f}_a(\mathbf{q}(t))\|<\theta'$ for all $a\in\overline{1,l}$ and $t\in[t_0,\infty)$ by Theorem \ref{thm:task_convergence_and_stability_in_continuous_time}.

Since $\|(\mathbf{R}^{-1}\hat{\mathbf{J}}^T\mathbf{C}_D^T\mathbf{L})(\mathbf{q})\|\le (m_\Theta^2 + w^2)^{-1/2}$ for all $\mathbf{q}\in\mathbb{R}^n$ and $k_a\to k$ as $w\to0$ for all $a\in\overline{1,l}$, 
we have $\|\mathbf{u}(\mathbf{q})\| \le M_u = \theta(\sum_{a=1}^lk_a^2)^{1/2}/\sqrt{m_\Theta^2+w^2} \to \theta k\sqrt{l}/m_\Theta$ as $w\to0$ for all $\mathbf{q}\in\hat{\Theta}_0 = \{\mathbf{q}'\in\Theta_0\mid \mathbf{f}_a(\mathbf{q}')\in\hat{\mathbf{p}}_a+\theta B_{m_a},\,a\in\overline{1,l}\}$.
Define $\Phi_u:[m_{ijk}]\mapsto[u_{ijk}]$ and $\Phi_l:[m_{ijk}]\mapsto[l_{ijk}]$ as:
\begin{equation*}
u_{ijk} = \begin{dcases*} m_{ijk}, & $i < j$ \\ m_{iik}/2, & $i = j$ \\ 0, & $i > j$ \end{dcases*},\quad l_{ijk} = \begin{dcases*} 0, & $i < j$ \\ m_{iik}/2, & $i = j$ \\ m_{ijk}, & $i > j$ \end{dcases*}
\end{equation*}
for all three dimensional arrays $[m_{ijk}]$, $[u_{ijk}]$, and $[l_{ijk}]$ with the same dimension.
As in \cite{Murray2016}, we can derive $\mathsf{D}_q\mathbf{R}^{-1}(\mathbf{q}) = - (\mathbf{R}^{-1}\Phi_u(\mathbf{A} + \mathbf{A}^T))(\mathbf{q})$, $\mathsf{D}_q\mathbf{C}(\mathbf{q}) = (\mathbf{C}\Phi_l(\mathbf{B}+\mathbf{B}^T))(\mathbf{q})$, and $\mathsf{D}_q\hat{\mathbf{J}}(\mathbf{q}) = (\mathbf{C}^{-1}\mathsf{D}_q\mathbf{F}_q\mathbf{R}^{-1} - \hat{\mathbf{J}}\Phi_u(\mathbf{A} + \mathbf{A}^T) - \Phi_l(\mathbf{B} + \mathbf{B}^T)\hat{\mathbf{J}})(\mathbf{q})$ for all $\mathbf{q}\in\Theta_0$ where $\mathbf{A} = \mathbf{R}^{-T}\mathsf{D}_q\mathbf{F}_q^T\mathbf{J}$ and $\mathbf{B} = (\mathbf{C}^{-1} - \mathbf{C}^T)\mathsf{D}_q\mathbf{F}_q\mathbf{R}^{-1}\hat{\mathbf{J}}^T$.
Then, $\|\mathsf{D}_q(\mathbf{R}^{-1}\hat{\mathbf{J}}^T\mathbf{C}_D^T)(\mathbf{q})\| \le L_M = L_\Theta(m_\Theta^{-1}(m_\Theta^2 + w^2)^{-1/2} + 2(m_\Theta^2 + w^2)^{-1} + \sqrt{2}nw^2m_\Theta^{-1}(m_\Theta^2+w^2)^{-3/2})$ and $\|\mathsf{D}_q\mathbf{u}(\mathbf{q})\| \le L_u = k_l(\theta L_M\sqrt{l} + M_\Theta(m_\Theta^2 + w^2)^{-1/2}) \to k(3\theta L_\Theta\sqrt{l}m_\Theta^{-2} + M_\Theta m_\Theta^{-1})$ as $w\to0$ for all $\mathbf{q}\in\hat{\Theta}_0$.
If we let $M_{F_a} = M_\Theta$, then we can find that, as $w\to0$, the upper bound of $\eta_\infty$ in \eqref{eqn:upper_bound_of_eta_for_phi_to_be_smaller_than_theta_for_all_time} converges to
\begin{equation*}
\frac{\theta-\theta'}{M_\Theta\left(\frac{M_\Theta}{m_\Theta} + \frac{3\theta L_\Theta\sqrt{l}}{m_\Theta^2}\right)\left(1+\frac{\theta k\sqrt{l}}{m_\Theta}\right)}.
\end{equation*}

In order to justify this result, we also find the upper bound of $\eta_\infty$ when there is no preconditioning.
Let $\mathbf{R} = \mathbf{I}_n$ and $\omega = m_\Theta^2$.
It is not difficult to see that $(\mathbf{A}_{aa}+\mathbf{A}_{aa}^T)(\mathbf{q}) - 2\omega\mathbf{I}_{m_a} \ge 0$ for all $a\in\overline{1,l}$ and $\mathbf{q}\in\Theta_0$.
\ref{ass:lower_bound_of_k_a} can be rewritten as $k_1 > 0$ and $k_a > \sum_{b=1}^{a-1}k_bM_{ab}/m_\Theta^2$ for $a\in\overline{2,l}$.
It would be important to find a tight upper bound $M_{ab}\in(0,\infty)$ for $a,b\in\overline{1,l}$ satisfying $\|\mathbf{A}_{ab}(\mathbf{q})\| = \|(\bar{\mathbf{C}}_{ab}\bar{\mathbf{C}}_{bb}^T)(\mathbf{q})\| \le M_{ab}$ for all $\mathbf{q}\in\Theta_0$, not to increase $k_a$ too much.
Let $k = \min\{k_1,\dots,k_l\}$.
We have $\|\mathbf{u}(\mathbf{q})\| \le M_u = \theta M_\Theta(\sum_{a=1}^lk_a^2)^{1/2}$ for all $\mathbf{q}\in\hat{\Theta}_0$.
Similarly as before, we can find $\mathsf{D}_q\bar{\mathbf{C}}(\mathbf{q}) = (\bar{\mathbf{C}}\Phi_l(\bar{\mathbf{C}}^{-1}\mathsf{D}_q(\mathbf{F}_q\mathbf{F}_q^T)\bar{\mathbf{C}}^{-T}))(\mathbf{q})$ and $\mathsf{D}_q\hat{\bar{\mathbf{J}}}(\mathbf{q}) = \mathsf{D}_q(\bar{\mathbf{C}}^{-1}\mathbf{F}_q)(\mathbf{q})$ for all $\mathbf{q}\in\Theta_0$.
From this, we can formulate $\|\mathsf{D}_q(\hat{\bar{\mathbf{J}}}^T\bar{\mathbf{C}}_D^T)(\mathbf{q})\| \le L_M = (1+\sqrt{2n})M_\Theta L_\Theta/m_\Theta$ and $\|\mathsf{D}_q\mathbf{u}(\mathbf{q})\| \le L_u = \|\mathbf{K}\|(\theta L_M\sqrt{l} + M_\Theta^2)$ for all $\mathbf{q}\in\hat{\Theta}_0$.
By using $(\sum_{a=1}^lk_a^2)^{1/2} > k(\sum_{a=1}^l(\sum_{b=1}^{a-1}M_{ab})^2)^{1/2}/m_\Theta^2$, we can find that the upper bound of $\eta_\infty$ in \eqref{eqn:upper_bound_of_eta_for_phi_to_be_smaller_than_theta_for_all_time} is less than or equal to 
\begin{equation*}
\frac{\theta-\theta'}{M_\Theta\left(\frac{M_\Theta}{m_\Theta} + \frac{(1+\sqrt{2n})\theta L_\Theta\sqrt{l}}{m_\Theta^2}\right)\left(1 +\frac{\theta \|\mathbf{K}\|\sqrt{S}}{m_\Theta}\right)}
\end{equation*}
where $S = \sum_{a=1}^l(\sum_{b=1}^{a-1}M_{ab})^2$.
Observe that if $l,S\gg 0$, then $\|\mathbf{K}\|$ becomes much bigger than $k$.
Therefore, if $l,n,S\gg0$ and $w\ll1$, then the preconditioning given by Theorem \ref{thm:preconditioning} makes the upper bound of $\eta_\infty$ in \eqref{eqn:upper_bound_of_eta_for_phi_to_be_smaller_than_theta_for_all_time} higher than the one without preconditioning.

\section{Conclusion}
\label{sec:conclusion}

We have presented a method that guarantees that a joint trajectory generated from a class of PIK solutions exists uniquely and stays in a nonsingular configuration space.
So far as we know, this is the first time in the study of the PIK problem that a joint trajectory is guaranteed to have this property.
Based on this result, we could find incremental sufficient conditions for task convergence, stability, uniform stability, uniform asymptotic stability, and exponential stability in continuous time.
We further extended the study to discrete time on the consideration of practical applications by finding incremental sufficient conditions for task convergence, uniform stability, uniform asymptotic stability, and exponential stability.
The analysis results in discrete time show that the total number of tasks is limited in the practical applications.
We analyzed preconditioning and showed how preconditioning can be used in order to overcome this limitation.

\ifCLASSOPTIONcaptionsoff
  \newpage
\fi


\begin{IEEEbiography}[{\includegraphics[width=1in,height=1.25in,clip,keepaspectratio]{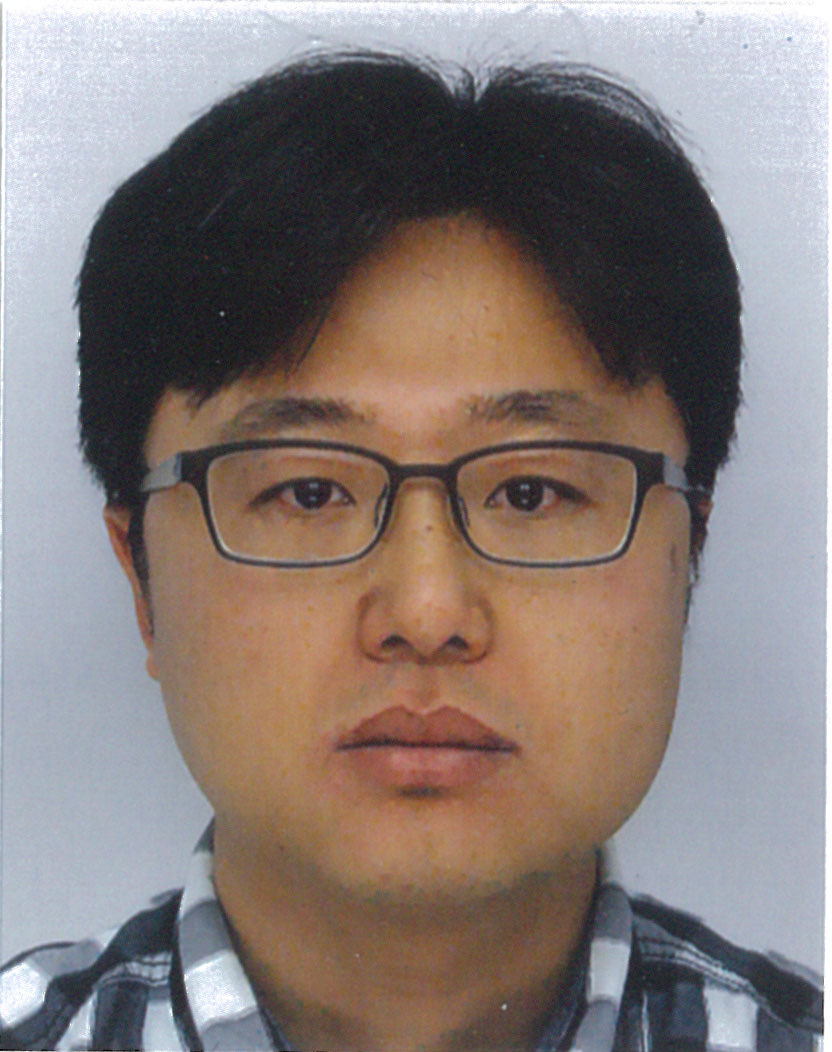}}]{Sang-ik An}
	received the B.S. in both mechanical and electronic engineerings from Korea Aerospace University, Korea in 2008 and the M.S. in mechanical engineering from Korea Advanced Institute of Science and Technology, Korea in 2010. He is currently working towards the Ph.D. degree in electrical and computer engineering at Technical University of Munich, Germany.
\end{IEEEbiography}
\begin{IEEEbiography}[{\includegraphics[width=1in,height=1.25in,clip,keepaspectratio]{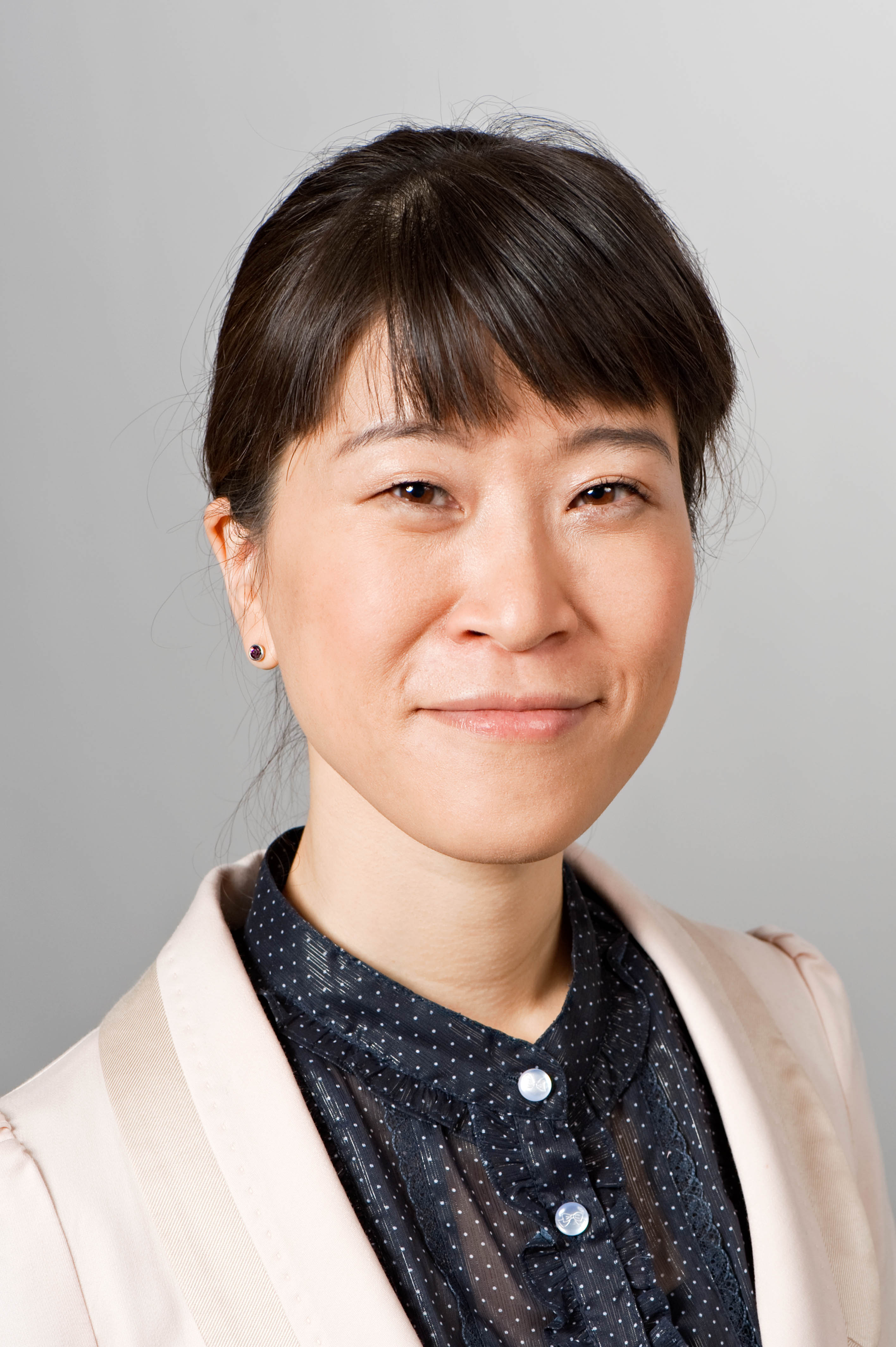}}]{Dongheui Lee}
	is Associate Professor of Human-centered Assistive Robotics at the TUM Department of Electrical and Computer Engineering. She is also director of a Human-centered assistive robotics group at the German Aerospace Center (DLR). Her research interests include human motion understanding, human robot interaction, machine learning in robotics, and assistive robotics.
	
	Prior to her appointment as Associate Professor, she was an Assistant Professor at TUM (2009-2017), Project Assistant Professor at the University of Tokyo (2007-2009), and a research scientist at the Korea Institute of Science and Technology (KIST) (2001-2004). After completing her B.S. (2001) and M.S. (2003) degrees in mechanical engineering at Kyung Hee University, Korea, she went on to obtain a PhD degree from the department of Mechano-Informatics, University of Tokyo, Japan in 2007. She was awarded a Carl von Linde Fellowship at the TUM Institute for Advanced Study (2011) and a Helmholtz professorship prize (2015). She is coordinator of both the euRobotics Topic Group on physical Human Robot Interaction and of the TUM Center of Competence Robotics, Autonomy and Interaction.
\end{IEEEbiography}
\vfill

\end{document}